\UseRawInputEncoding
\documentclass[11pt, reqno]{amsart}
\usepackage{amsfonts,latexsym, enumerate}
\usepackage{amsmath}
\usepackage{amscd}
\usepackage{float,amsmath,amssymb,mathrsfs,bm,multirow,graphics}
\usepackage[dvips]{graphicx}
\usepackage[percent]{overpic}
\usepackage{amsaddr}
\usepackage{todonotes}
\usepackage[numbers,sort&compress]{natbib}
\usepackage{grffile}

\numberwithin{equation}{section}
\newcommand{\R}{{\Bbb R}}

\newcommand{\C}{{\Bbb C}}

\newcommand{\Z}{{\Bbb Z}}

\DeclareMathOperator{\tr}{tr}

\DeclareMathOperator{\im}{Im}
\DeclareMathOperator{\re}{Re}

\DeclareMathOperator{\sgn}{sgn}

\DeclareMathOperator{\arctanh}{arctanh}
\DeclareMathOperator{\sech}{sech}

\newcommand{\BCS}{\text{\upshape BCS}}
\newcommand{\half}{\text{\upshape half}}
\newcommand{\AF}{\text{\upshape AF}}
\newcommand{\qm}{\text{\upshape qm}}

\def\XXint#1#2#3{{\setbox0=\hbox{$#1{#2#3}{\int}$}
\vcenter{\hbox{$#2#3$}}\kern-.5\wd0}}

\newtheorem{theorem}{Theorem}[section]

\newtheorem{lemma}[theorem]{Lemma}

\newtheorem{remark}[theorem]{Remark}

\newtheorem{figuretext}{Figure}

\usepackage[colorlinks=true]{hyperref}
\hypersetup{urlcolor=blue, citecolor=red, linkcolor=blue}

\input epsf
\title[Universality of mean-field antiferromagnetic order]{Universality of mean-field antiferromagnetic order in an anisotropic 3D Hubbard model at half-filling}

\author{E. Langmann$^{1}$ and J. Lenells$^{2}$}
\address{$^1$Department of Physics, KTH Royal Institute of Technology, \\ 106 91 Stockholm, Sweden
	\\
$^{2}$Department of Mathematics, KTH Royal Institute of Technology, \\ 100 44 Stockholm, Sweden}
\email{langmann@kth.se}
\email{jlenells@kth.se}
    
\begin{document}
\begin{abstract}
We study the 3D anisotropic Hubbard model on a cubic lattice with hopping parameter $t$ in the $x$- and $y$-directions and a possibly different hopping parameter $t_z$ in the $z$-direction; this model interpolates between the 2D and 3D Hubbard models corresponding to the limiting cases $t_z=0$ and $t_z=t$, respectively.
We first derive all-order asymptotic expansions for the density of states. Using these expansions and units such that $t=1$, we analyze how the N\'eel temperature and the antiferromagnetic mean field depend on the coupling parameter, $U$, and on the hopping parameter $t_z$. We derive asymptotic formulas valid in the weak coupling regime, and we study in particular the transition from the three-dimensional to the two-dimensional model as $t_z \to 0$. 
It is found that the asymptotic formulas are qualitatively different for $t_z = 0$ (the two-dimensional case) and $t_z > 0$ (the case of nonzero hopping in the $z$-direction). 
Our results show that certain universality features of the three-dimensional Hubbard model are lost in the limit $t_z \to 0$ in which the three-dimensional model reduces to the two-dimensional model.   
\end{abstract} 

\maketitle




\section{Introduction}
The Hubbard model was independently introduced by Gutzwiller \cite{G1963}, Hubbard \cite{H1963}, and Kanamori \cite{K1963} in 1963 as a model for electrons in transition and rare-earth metals. It is often described as being of similar importance to interacting electron physics as the Ising model is to statistical mechanics, see e.g. \cite{QSACG2022}. Interest in the Hubbard model rose in the late 1980's following a proposal by Anderson \cite{A1987} that the 2D Hubbard model provides a prototype model for high-temperature superconductivity.
Although a lot of effort has been devoted to the study of the Hubbard model in the roughly sixty years since it was first described, many basic questions  remain unsettled, see \cite{ABKR2022} for a recent review. 

In this paper, we consider the thermodynamic limit of the Hubbard model with coupling parameter $U > 0$ and nearest-neighbor hopping on a hypercubic lattice $\Lambda_L \subset \Z^n$ with $L^n$ lattice points in two ($n = 2$) and three ($n = 3$) dimensions; see Section~\ref{sec:Hubbard} for precise definitions. In two dimensions, we take the nearest-neighbor hopping parameter $t$ to be the same for all bonds; in three dimensions, we allow for a different hopping parameter, $t_z$, in the $z$-direction, which makes the model anisotropic. 
We are particularly interested in the limit $t_z \to 0$ in which the anisotropic three-dimensional model reduces to the two-dimensional model. Our motivation for considering the limit $t_z \to 0$ is the following: The 2D Hubbard model (i.e., the model with $t_z = 0$) is often used as a model for high-temperature superconductivity, because the superconductivity is believed to be confined to two-dimensional copper-oxide planes with the surrounding layers merely playing the role of charge reservoirs. This assumption is based on the fact that the distance between the copper-oxide planes is larger than the copper-oxygen interspacing within a single layer \cite{PFCP2014}. However, it is known that there is also electron interaction between different planes (see e.g. \cite{TSKSS1989, SMRB1995}), meaning that the hopping parameter $t_z$ in the $z$-direction is small, but not zero. Setting $t_z = 0$ is an adequate approximation if $t_z \to 0$ is a regular limit, however, as we will argue in this paper, the limit $t_z \to 0$ is, at least in some respects, singular. Even though our results only concern a mean-field approximation at half-filling, they suggest that the anisotropic 3D Hubbard model with a small but nonzero $t_z$ value in some circumstances can serve as a more robust model for high-temperature superconductivity than the 2D Hubbard model.

\subsection{Mean-field equation for antiferromagnetic order}
We will consider the Hubbard model at half-filling (meaning that there is on average one electron per lattice site) and we will restrict ourselves to states whose expectation value of the spin operator has the form
$$\vec{m}(\mathbf{x}) = m_{\AF} (-1)^{x_1 + \cdots + x_n}  \vec{e},$$
where $m_{\AF} \in [0,1]$ is independent of $\mathbf{x} = (x_1, x_2, \dots, x_n) \in\Lambda_L$, and $\vec{e} \in \R^3$ is a fixed unit vector. 
Such states describe uniform antiferromagnetic order. 

In the analysis of the Hubbard model, one is typically forced to resort to various approximation schemes, and we will adopt the mean-field approximation obtained from Hartree--Fock theory. In the thermodynamic limit $L \to \infty$ of an infinitely large lattice, this gives the following equation for the antiferromagnetic order parameter,  
\begin{align}\label{Deltadef} 
\Delta_{\AF} := \frac{Um_{\AF}}{2}, 
\end{align}
which is our starting point:
\begin{align}\label{meanfieldtprimezero}
\Delta_{\AF} = \Delta_{\AF} \int_{[-\pi,\pi]^n} 
\frac{U\tanh(\frac{\sqrt{\Delta_{\AF}^2 + \varepsilon(\mathbf{k})^2}}{2T})}{2\sqrt{\Delta_{\AF}^2 + \varepsilon(\mathbf{k})^2}} \frac{d\mathbf{k}}{(2\pi)^n},
\end{align}
where $T > 0$ is the temperature and
\begin{align}\label{epsilondef}
\varepsilon(\mathbf{k}) := \begin{cases} 
- 2 t(\cos{k_1} + \cos{k_2}) & \text{for $n = 2$}, \\
- 2 t(\cos{k_1} + \cos{k_2}) - 2 t_z \cos{k_3} & \text{for $n = 3$}. 
\end{cases}
\end{align}
If $T = 0$, the mean-field equation for $\Delta_{\AF}$ is given by
\begin{align}\label{meanfieldtprimezeroT0}
\Delta_{\AF} = \Delta_{\AF} \int_{[-\pi,\pi]^n} 
\frac{U}{2\sqrt{\Delta_{\AF}^2 + \varepsilon(\mathbf{k})^2}} \frac{d\mathbf{k}}{(2\pi)^n},
\end{align}
which is the natural limit of (\ref{meanfieldtprimezero}) as $T$ tends to $0$. We note that $\Delta_{\AF}$ has a natural physics interpretation as the antiferromagnetic energy gap, and $\varepsilon(\mathbf{k})$ is known as the tight-binding band relation.

A systematic account of Hartree--Fock theory for the Hubbard model can be found in \cite{BLS1994} and a derivation of equations (\ref{meanfieldtprimezero}) and (\ref{meanfieldtprimezeroT0}) is provided in Section \ref{derivationsec}. It is worth mentioning that, for any dimension $n \geq 2$, the Hartree--Fock approximation of the Hubbard model becomes exact in the limit of small $U$, in the sense that the difference between the Hartree--Fock mean-field energy and the true ground state energy vanishes in this limit \cite{BP1997}. 

When $t_z$ is set to zero, equation (\ref{meanfieldtprimezero}) with $n = 3$ reduces to the corresponding equation with $n = 2$, because when $t_z = 0$ the integrand is independent of $k_3$ and $\frac{1}{2\pi} \int_{[-\pi, \pi]} dk_3 = 1$. In other words, the equation for $\Delta_{\AF}$ for the three-dimensional model reduces to the corresponding equation for the two-dimensional model in the limit $t_z \to 0$.
We will therefore henceforth assume that $n = 3$; the two-dimensional model will be included in our analysis by allowing for $t_z = 0$.

\subsection{N\'eel temperature}
It turns out that for any $U> 0$, $t > 0$, and $t_z \geq 0$, there is a unique temperature $T_N > 0$, referred to as the N\'eel temperature, such that (\ref{meanfieldtprimezero}) has a unique positive solution $\Delta_{\AF}>0$ whenever $T \in (0, T_N)$ whereas no such solution exists if $T \geq T_N$. 
The N\'eel temperature (for $n = 3$) can be characterized as the unique solution $T_N > 0$ of the equation
\begin{align}\label{TNeelequation}
1 = \int_{[-\pi,\pi]^3} 
\frac{U\tanh(\frac{\varepsilon(\mathbf{k})}{2T_N})}{2\varepsilon(\mathbf{k})} \frac{dk_1dk_2dk_3}{(2\pi)^3}.
\end{align}
We include proofs of these facts in Section \ref{existencesubsec}. Clearly, $\Delta_{\AF}=0$ is a solution of equation (\ref{meanfieldtprimezero}) for any temperature $T$. However, for $T < T_N$, the magnetically ordered state with $\Delta_{\AF}>0$ is energetically preferred. The phase transition from the ordered to the disordered state occurs at the N\'eel temperature $T_N$.

\subsection{Scaling properties}
Equation (\ref{meanfieldtprimezero}) depends on the four parameters $(U,t,t_z, T)$ and equation (\ref{TNeelequation}) depends on the three parameters $(U,t,t_z)$.
It is easy to see that if $\Delta_{\AF}$ and $T_N$ are solutions of (\ref{meanfieldtprimezero}) and (\ref{TNeelequation}) associated to the parameters $(U,t,t_z, T)$ and $(U, t, t_z)$, respectively, then, for any $\alpha > 0$, $\alpha\Delta_{\AF}$ and $\alpha T_N$ are solutions of (\ref{meanfieldtprimezero}) and (\ref{TNeelequation})  associated to the parameters $(\alpha U, \alpha t, \alpha t_z, \alpha T)$ and $(\alpha U, \alpha t, \alpha t_z)$, respectively. It is convenient to use these scaling properties with $\alpha = 1/t$ to set $t = 1$.\footnote{Some formulas below, such as \eqref{N0expansionfirstfew} and \eqref{Ntzat0expansion}, become slightly simpler with a different choice of $t$, namely $t = 1/16$. However, we choose $t = 1$ because it is the most common normalization in the literature.}
In the rest of the paper, we therefore assume that $t = 1$ and write $\Delta_{\AF}(U, t_z, T)$ and $T_N(U, t_z)$ for the solutions of (\ref{meanfieldtprimezero}) and (\ref{TNeelequation}) associated to the parameters $(U,1,t_z, T)$ and $(U, 1, t_z)$, respectively. A similar statement applies to the solution of equation (\ref{meanfieldtprimezeroT0}) at temperature $T = 0$ which we denote by $\Delta_{\AF}(U, t_z, 0)$.

\begin{figure}
\bigskip
\begin{center}
\begin{overpic}[width=.46\textwidth]{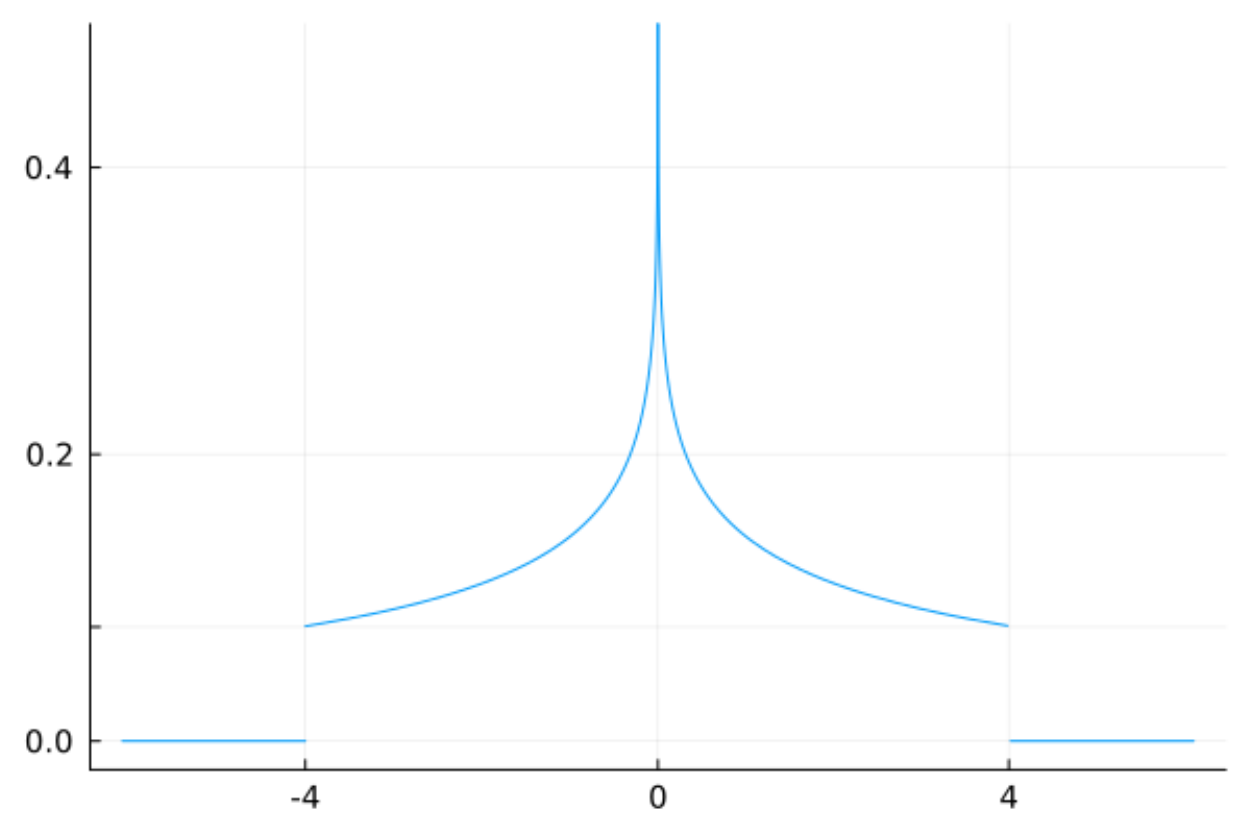}
      \put(2,68){\footnotesize $N_{0}(\epsilon)$}
      \put(101,3.7){\footnotesize $\epsilon$}
       \put(1,15.4){\footnotesize $\frac{1}{4\pi}$}
   \end{overpic}
   \hspace{0.5cm}
\begin{overpic}[width=.46\textwidth]{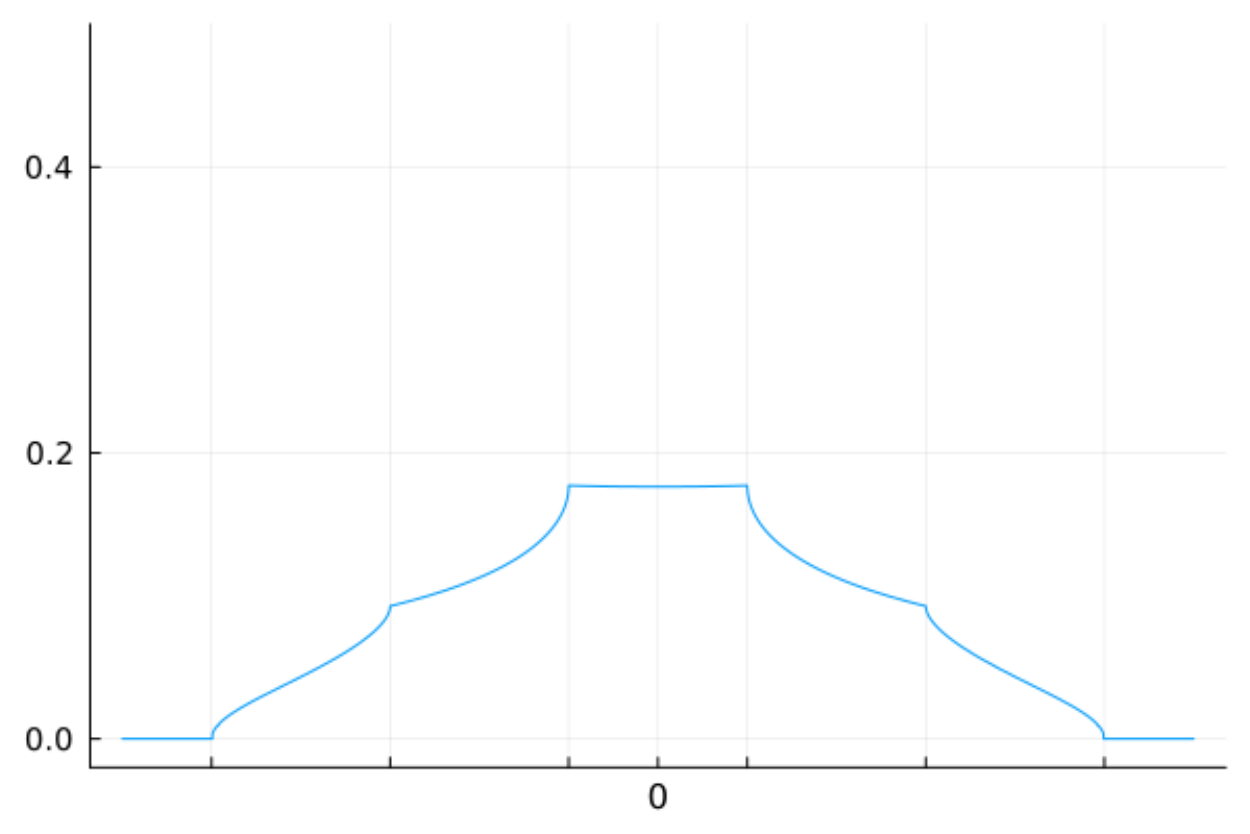}
      \put(1,68){\footnotesize $N_{t_z}(\epsilon)$}
      \put(101,3.7){\footnotesize $\epsilon$}
      \put(10,0.1){\footnotesize $^{-4-2t_z}$}
      \put(25,0.1){\footnotesize $^{-4+2t_z}$}
      \put(40.5,0.1){\footnotesize $^{-2t_z}$}
      \put(58,0.1){\footnotesize $^{2t_z}$}
      \put(70,0.1){\footnotesize $^{4-2t_z}$}
      \put(85,0.1){\footnotesize $^{4+2t_z}$}
    \end{overpic}
     \begin{figuretext}\label{N0N3Dschematicfig}
       The function $\epsilon \mapsto N_{t_z}(\epsilon)$ for $t_z = 0$ (left) and $t_z = 1/2$ (right). 
      \end{figuretext}
     \end{center}
\end{figure}

\subsection{Density of states}
The density of states $N_{t_z}(\epsilon)$ is defined by
\begin{align}\nonumber
N_{t_z}(\epsilon)  
& = \int_{[-\pi,\pi]^3} \delta(\varepsilon(\mathbf{k}) - \epsilon) \frac{dk_1dk_2dk_3}{(2\pi)^3}
	\\\label{3Ddensityofstates}
& = \int_{[-\pi,\pi]^3} \delta\big(2 (\cos{k_1} + \cos{k_2}) + 2 t_z \cos{k_3} + \epsilon\big) \frac{dk_1dk_2dk_3}{(2\pi)^3},
\end{align}
see Figure \ref{N0N3Dschematicfig}. We note that $N_{t_z}(\epsilon)$ is an even function of $\epsilon$, and that
\begin{align}\label{2Ddensityofstates}
N_{0}(\epsilon) = \int_{[-\pi,\pi]^2} \delta\big(2 (\cos{k_1} + \cos{k_2}) + \epsilon\big) \frac{dk_1dk_2}{(2\pi)^2}
\end{align}
is the density of states for the two-dimensional model. 
In terms of $N_{t_z}$, 
equation (\ref{meanfieldtprimezero}) for $\Delta_{\AF}$ can be written as
\begin{align}\label{mAFeq}
\frac{1}{U} = \int_0^\infty N_{t_z}(\epsilon)  
\frac{\tanh(\frac{\sqrt{\Delta_{\AF}^2 + \epsilon^2}}{2T})}{\sqrt{\Delta_{\AF}^2 + \epsilon^2}}  d\epsilon
\end{align}
and equation (\ref{TNeelequation}) for $T_N$ can be written as
\begin{align}\label{TNeelequationNtz}
\frac{1}{U} = \int_0^\infty N_{t_z}(\epsilon) 
\frac{\tanh(\frac{\epsilon}{2T_N})}{\epsilon} d\epsilon.
\end{align}
When writing (\ref{mAFeq}), we have ignored the trivial zero solution $\Delta_{\AF}=0$.\footnote{See Lemma \ref{Gminlemma} for a proof of the fact that the nonzero solution is the physically preferred solution for temperatures below the N\'eel temperature.} Moreover, we adopt the convention that $\tanh(\cdot/T) = 1$ for $T = 0$ and then equation (\ref{meanfieldtprimezeroT0}) for $\Delta_{\AF}$ at temperature $T= 0$ is included as a special case of (\ref{mAFeq}).

\subsection{Brief description of main results}
The main results of the paper are presented in the form of eight theorems:
\begin{enumerate}[$-$]
\item Theorems \ref{2Ddensityth} and \ref{3Ddensityth} provide asymptotic expansions to all orders of the two-dimensional density of states $N_0(\epsilon)$ and of the three-dimensional density of states $N_{t_z}(\epsilon)$, respectively, as $\epsilon \to 0$. While $N_0(\epsilon)$ has a logarithmic singularity at $\epsilon = 0$, the function $N_{t_z}(\epsilon)$ is regular at $\epsilon = 0$ for any $t_z > 0$. This is a first indication of the singular nature of the limit $t_z \to 0$. The logarithmic singularity is known as a van Hove singularity in the physics literature and is a consequence of the fact that $\varepsilon(\mathbf{k})$ has saddle points on the fermi surface $\varepsilon(\mathbf{k}) =0$ when $n = 2$ (at $\mathbf{k} = (\pi, 0)$ and $\mathbf{k} = (0,\pi)$). No saddle points are present on this fermi surface when $t_z > 0$.

\item Theorem \ref{3Dtzdensityth} provides an asymptotic expansion of $N_{t_z}(0)$ as $t_z \to 0$, showing in particular that $N_{t_z}(0)$ has a logarithmic singularity as $t_z \to 0$. 

\item Theorems \ref{2DNeelth} and \ref{3DNeelth} provide the leading asymptotic behavior of the N\'eel temperature $T_N(U,t_z)$ as $U$ tends to zero for $t_z = 0$ and $t_z > 0$, respectively. Both $T_N(U,0)$ and $T_N(U,t_z)$ with $t_z > 0$ are exponentially small in the limit $U \to 0$, but the exponential decay is much faster if $t_z > 0$.

\item Theorems \ref{2Dmeanfieldth} and \ref{3Dmeanfieldth} provide asymptotic expansions of the antiferromagnetic order parameter $\Delta_{\AF}(U, t_z, T)$ as $U \to 0$ for $t_z = 0$ and $t_z > 0$, respectively. In fact, rather than stating the asymptotics for $\Delta_{\AF}$, we give the asymptotics for the quotient
\begin{align}\label{hatmdef}
\hat{m}(U, t_z, T) := \frac{\Delta_{\AF}(U, t_z, T)}{T_N(U, t_z)}, 
\end{align}
which we sometimes call the gap ratio. 
It is natural to consider $\hat{m}$ because the coefficients of the expansions of $\hat{m}$ depend on $T$ and $t_z$ only via the reduced temperature $\frac{T}{T_N(U,t_z)}$.

\item Theorem \ref{3Dmeanfieldimprovedth} computes an exponentially small subleading term in the asymptotics of $\hat{m}(U, t_z, T)$, thereby providing a refinement of Theorem \ref{3Dmeanfieldth}.

\end{enumerate}

\subsection{Universality}
The asymptotic formulas for $\hat{m}$ obtained in Theorems \ref{2Dmeanfieldth}--\ref{3Dmeanfieldimprovedth} involve a function, $f_\BCS$, well-known from Bardeen--Cooper--Schrieffer (BCS) theory for superconductivity (see e.g.\ \cite{L2006}). The function $f_{\BCS}(y)$ can be defined for $y \in [0,1]$ as the unique solution of the equation $J(f_{\BCS}(y), y) =0$, where
\begin{align}\label{Jdef} 
J(x,y) := \int_0^\infty \bigg(\frac{\tanh(\frac{\sqrt{x^2 + \epsilon^2}}{2y})}{\sqrt{x^2 +  \epsilon^2 }} 
- \frac{\tanh(\frac{\epsilon}{2})}{\epsilon} \bigg) d\epsilon,
\end{align}
and $\tanh(\cdot/y) = 1$ for $y = 0$ by convention; it is a decreasing function of $y$ that satisfies $f_{\BCS}(0) = \pi e^{-\gamma} \approx 1.764$ and $f_{\BCS}(1) = 0$, where $\gamma \approx 0.5772$ is Euler's gamma constant, see Figure \ref{fBCSc1fig} (see \cite[Appendix~A]{LT2023} for further details about $f_{\BCS}(y)$). In BCS theory, $f_\BCS$ describes the famous universal behavior of the gap ratio. More precisely, the ratio of the temperature-dependent superconducting gap in BCS theory, $\Delta(T)$, to the superconducting critical temperature, $T_c$, satisfies
\begin{align}\label{fracDeltaTcfBCS}
\frac{\Delta(T)}{T_c} \approx f_\BCS(T/T_c)
\end{align}
up to correction terms that are exponentially small in the BCS coupling parameter, see \cite{LT2023} for a proof. The independence of the right-hand side of (\ref{fracDeltaTcfBCS}) on model parameters is a reflection of the universality of BCS theory. 
In a completely analogous way, Theorem \ref{3Dmeanfieldth} shows that the function $\hat{m}(U, t_z, T)$ in (\ref{hatmdef}) satisfies
\begin{align}\label{hatmfBCS}
\hat{m}(U, t_z, T) \approx f_{\BCS}(T/T_N)
\end{align}
up to correction terms that are exponentially small in the coupling parameter $U$ of the Hubbard model. 
The independence of the right-hand side of (\ref{hatmfBCS}) on model parameters is a sign of universality in the Hubbard model, at least in the Hartree--Fock approximation. It is therefore interesting to note that equation (\ref{hatmfBCS}) is valid {\it only} when $t_z > 0$. For the two-dimensional Hubbard model (i.e., for $t_z = 0$), the universal prediction (\ref{hatmfBCS}) does not hold, rather it is replaced by an expansion of the form
$$\hat{m}(U, 0, T) = f_{\BCS}(T/T_N) + c_1(T/T_N) \sqrt{U} + O(U) \qquad \text{as $U \downarrow 0$}$$
for a special function $c_1(y)$ of $y\in[0,1]$ defined in  \eqref{c1def}, see Figure~\ref{fBCSc1fig}. 
The $U$-dependent term $c_1(T/T_N) \sqrt{U}$ on the right-hand side spoils the universality that was present in (\ref{hatmfBCS}).
Our results therefore show that certain universality features of the three-dimensional Hubbard model are lost in the limit $t_z \to 0$ in which the three-dimensional model reduces to the two-dimensional model.

\begin{figure}
\bigskip
\bigskip
\begin{center}
\begin{overpic}[width=.46\textwidth]{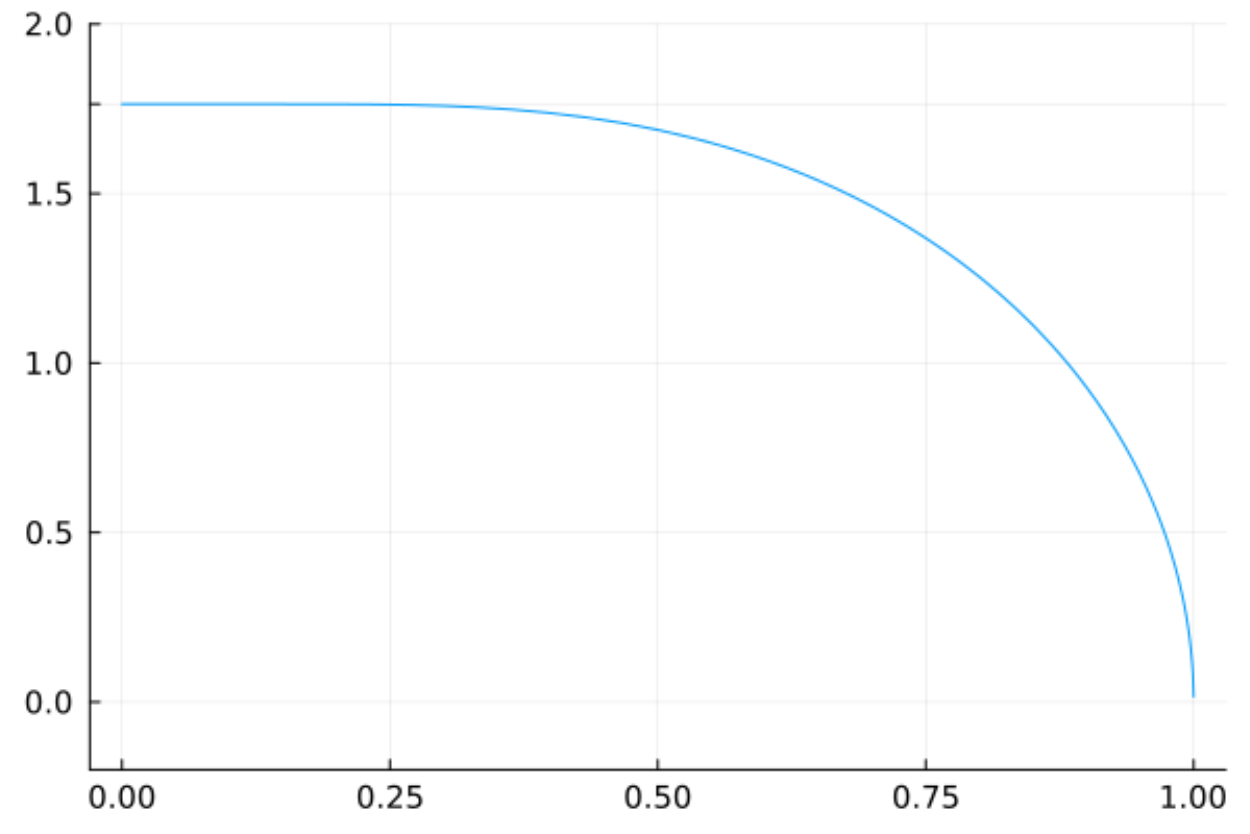}
      \put(1,69){\footnotesize $f_{\BCS}(y)$}
      \put(-4,57.5){\footnotesize $\pi e^{-\gamma}$}
      \put(100.5,3.7){\footnotesize $y$}
    \end{overpic}
   \hspace{0.2cm}
\begin{overpic}[width=.46\textwidth]{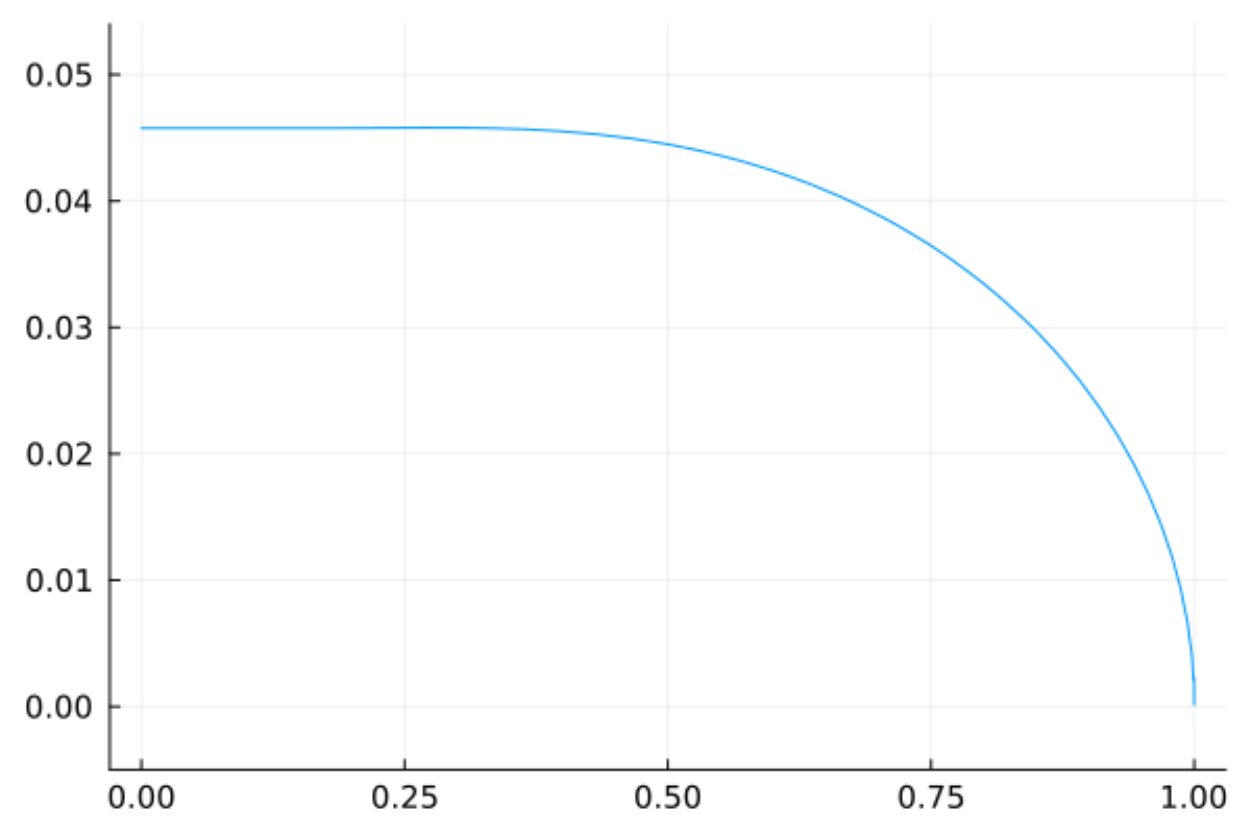}
      \put(4,68){\footnotesize $c_1(y)$}
      \put(100.5,3.7){\footnotesize $y$}
    \end{overpic}
     \begin{figuretext}\label{fBCSc1fig}
       The graphs of the function $f_{\BCS}(y)$ (left) and $c_1(y)$ (right). Note that $c_1(y) \approx \alpha_0 f_{\BCS}(y)$ where $\alpha_0 = c_1(0)/f_{\BCS}(0) \approx 0.02594$ and the error of this approximation, $c_1(y) - \alpha_0 f_{BCS}(y)$, is in the range between 0 and 0.001.
\end{figuretext}
     \end{center}
\end{figure}

\subsection{Related other works} Mathematical results about the Hubbard model in dimensions $n\geq 2$ have a long tradition, with famous examples including works by  Lieb, Lieb and Mattis, and Nagaokar; see e.g.\ \cite{T1998} for review. Our work owes much to previous mathematical results about Hartree--Fock theory for the Hubbard model \cite{BLS1994,BP1996,BP1997}. 
We were also inspired by previous mathematical work on BCS theory; see e.g.\ \cite{FHNS2007,HS2008,HL2022,HL2023}. 

We mention interesting work on non-interacting lattice fermion models with the tight-binding band relation $\varepsilon(\mathbf{k})$ in \eqref{epsilondef} \cite{DJ2001}. Using methods from special function theory, the authors of \cite{DJ2001} obtain an expansion that allows them to compute the density of states $N_{t_z}(\epsilon)$; our expansion is different.   

There is a large amount of work on Hartree--Fock theory on the Hubbard model in the literature which can be traced from the reviews \cite{ABKR2022,QSACG2022}; most of this work is numerical. 

The anisotropic 3D Hubbard model has received recent attention in the physics literature due to the possibility to realize it in cold atom systems; see e.g.\ \cite{I2020} and references therein.

\subsection{Organization of the paper}
The main results are stated in Section \ref{mainsec}. In Section \ref{derivationsec}, we review some aspects of Hartree--Fock theory for the Hubbard model and explain how it leads to the antiferromagnetic mean-field equations (\ref{meanfieldtprimezero}) and (\ref{meanfieldtprimezeroT0}); we also include proofs of the fact that $T_N$ and $\Delta_{\AF}$
are well-defined by these equations. In Sections \ref{densityofstatessec}--\ref{hatmsec}, we study asymptotic properties of the density of states $N_{t_z}(\epsilon)$, of the N\'eel temperature $T_N$, and of the gap ratio $\hat{m}$, respectively. In particular, Sections \ref{densityofstatessec}--\ref{hatmsec} contain proofs of all the main results. 

\subsection{Notation}
Throughout the paper $C >0$ and $c > 0$ will denote generic constants that may change within a calculation.
We write $x \downarrow a$ and $x \uparrow a$ to denote the limit as $x$ tends to $a \in \R$ from above and from below, respectively. 
In order to include the cases of positive and zero temperatures in the same formulas, we adopt the convention that $\tanh(\cdot/T) \equiv 1$ for $T = 0$ and $\tanh(\cdot/y) \equiv 1$ for $y = 0$ (for example, this allows us to include (\ref{meanfieldtprimezeroT0}) as a special case of (\ref{meanfieldtprimezero})).

\section{Main results}\label{mainsec}

\subsection{Asymptotic formulas for the density of states}

Our first main result, whose proof is given in Section \ref{2Ddensitysubsec}, provides an asymptotic expansion to all orders of the two-dimensional density of states $N_0(\epsilon)$ as $\epsilon \to 0$. 

\begin{theorem}[Behavior of $N_0(\epsilon)$ as $\epsilon \to 0$]\label{2Ddensityth}
The two-dimensional density of states $N_0(\epsilon)$ given by (\ref{2Ddensityofstates}) obeys the following asymptotic expansion to all orders as $\epsilon \downarrow 0$:
\begin{align}\label{N0expansion}
N_0(\epsilon)
   \sim -\frac{1}{\pi^2 (4-\epsilon)} \sum_{k,l=0}^\infty A_{k,l}\Big(-\frac{\epsilon}{4 - \epsilon}, \frac{4}{4-\epsilon}\Big).
\end{align}  
The functions $A_{k,l}$ in (\ref{N0expansion}) are defined in terms of the Gamma function $\Gamma(z)$ and the digamma function $\psi(z) = \frac{\Gamma'(z)}{\Gamma(z)}$ by
\begin{align}\nonumber
A_{k,l}(w_1, w_2&)
=  w_1^k (w_2-1)^l \frac{\Gamma \left(k+\frac{1}{2}\right)  \Gamma (k+l+1)}{  \Gamma (k+1)^2 \Gamma \left(\frac{1}{2}-l\right)
   \Gamma (l+1)^2} \bigg\{\ln(-w_1)+ \ln (w_2-1)
   	\\ \label{Akldef}
&  + 2\big(\psi(k+l+1)- \psi(k+1) -\psi(l+1)\big) +\psi\left(k+\frac{1}{2}\right)+\psi \left(\frac{1}{2}-l\right) \bigg\},
 \end{align}
where the principal branch is used for the logarithms. 
In particular, keeping the first several terms in (\ref{N0expansion}), we have
\begin{align}  \nonumber
N_0(\epsilon) 
= &\; \frac{\ln(\frac{16}{\epsilon})}{2 \pi^2} + \frac{\epsilon^2 (\ln(\frac{16}{\epsilon})-1)}{128 \pi^2}
+ \frac{3 \epsilon^4 (6 \ln(\frac{16}{\epsilon})-7)}{2^{16} \pi^2}
	\\ \nonumber
& +\frac{5}{3}\frac{\epsilon^6 (30 \ln(\frac{16}{\epsilon})-37)}{2^{22} \pi^2}
+ \frac{35}{3} \frac{\epsilon^8 (420 \ln(\frac{16}{\epsilon})-533)}{2^{33} \pi^2}
	\\ \label{N0expansionfirstfew}
&+ \frac{63}{5} \frac{\epsilon^{10} (1260\ln(\frac{16}{\epsilon})-1627)}{2^{39} \pi^2} + O\Big(\epsilon^{12} \ln{\tfrac{1}{\epsilon}}\Big) \qquad \text{as $\epsilon \downarrow 0$.}
\end{align} 
The expansions (\ref{N0expansion}) and (\ref{N0expansionfirstfew}) can be differentiated termwise with respect to $\epsilon$ any finite number of times.
\end{theorem}

\begin{figure}
\bigskip
\bigskip
\begin{center}
\begin{overpic}[width=.46\textwidth]{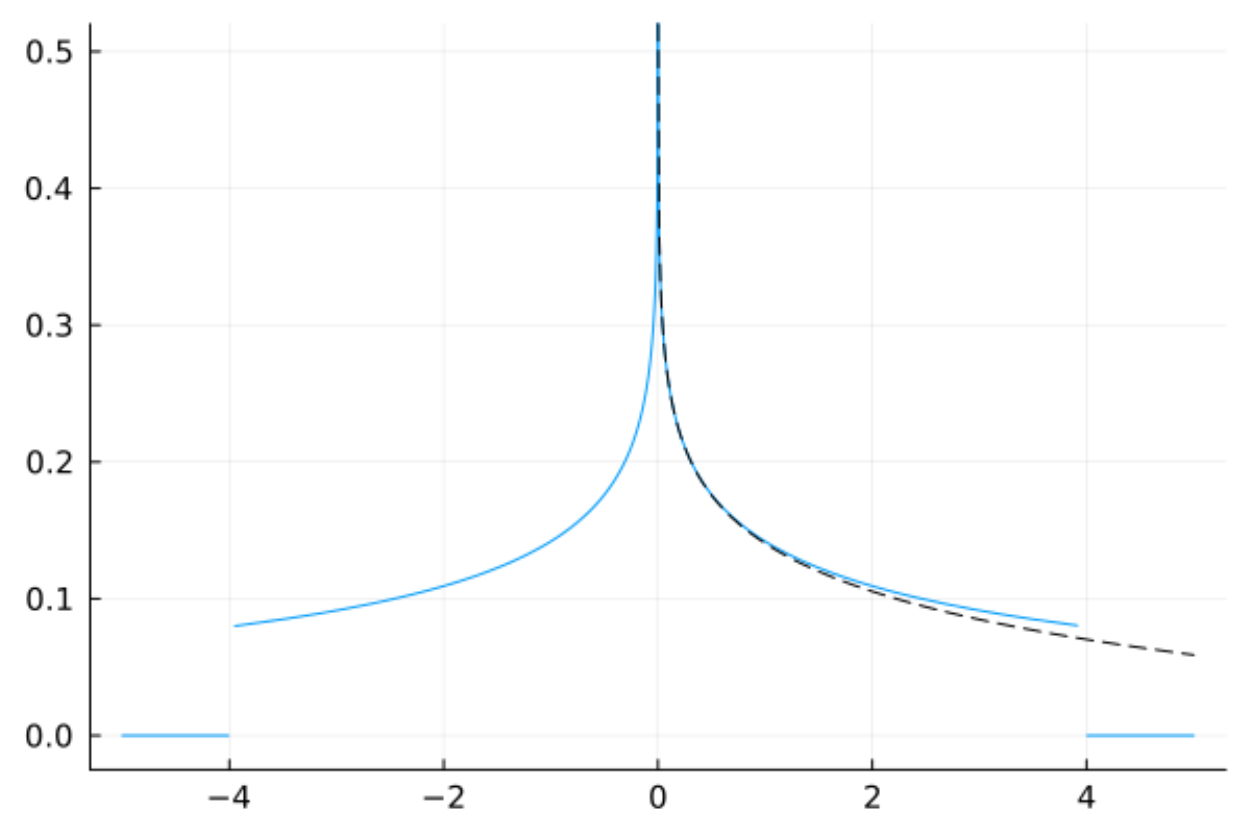}
      \put(2,68){\footnotesize $N_0(\epsilon)$}
      \put(101,3.7){\footnotesize $\epsilon$}
    \end{overpic}
   \hspace{0.5cm}
\begin{overpic}[width=.46\textwidth]{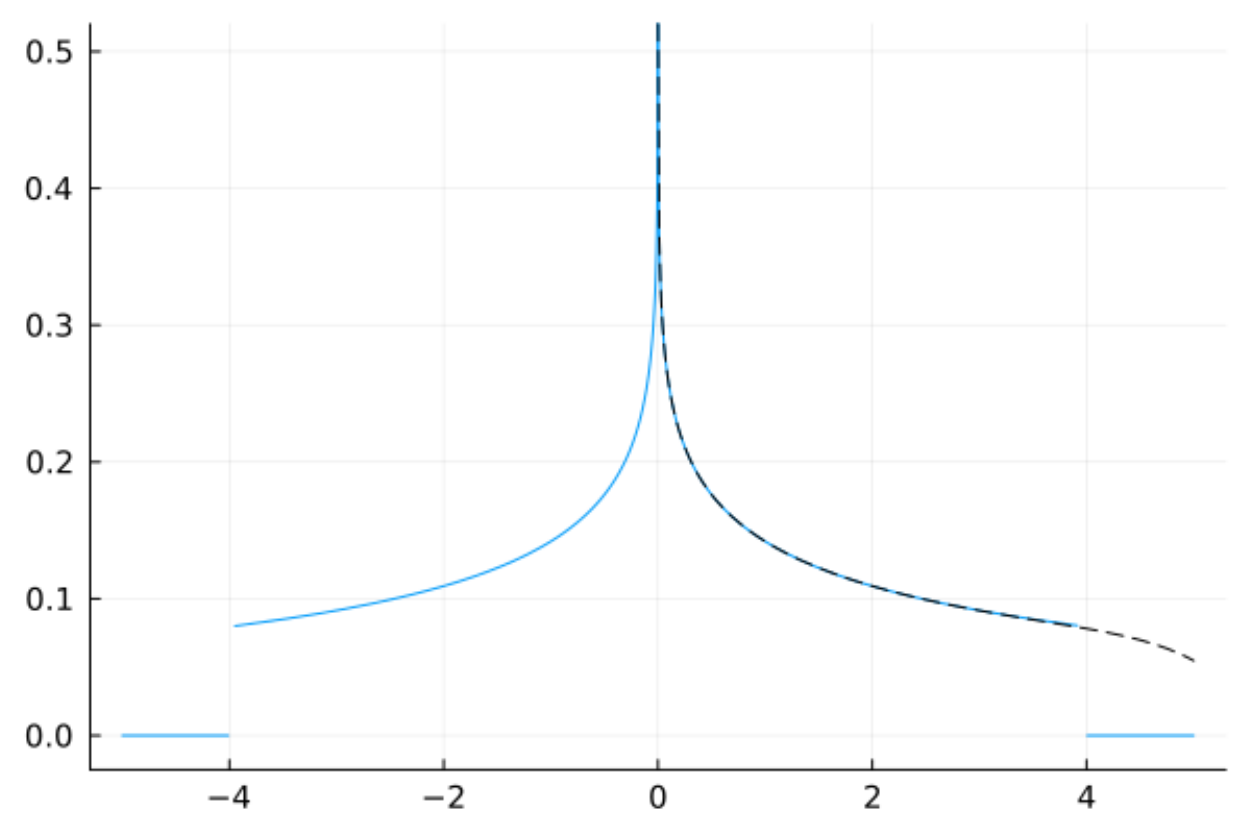}
      \put(2,68){\footnotesize $N_0(\epsilon)$}
      \put(101,3.7){\footnotesize $\epsilon$}
    \end{overpic}
     \begin{figuretext}[Illustration of Theorem \ref{2Ddensityth}]\label{N0fig}
       Left: The function $N_0(\epsilon)$ (solid blue) and the approximation $\frac{\ln(16/\epsilon)}{2 \pi^2}$ (dashed black).
       Right: The function $N_0(\epsilon)$ (solid blue) and the asymptotic approximation of (\ref{N0expansionfirstfew}) including terms up to order $O(\epsilon^{10})$ (dashed black).
      \end{figuretext}
     \end{center}
\end{figure}

\begin{figure}
\bigskip
\begin{center}
\begin{overpic}[width=.46\textwidth]{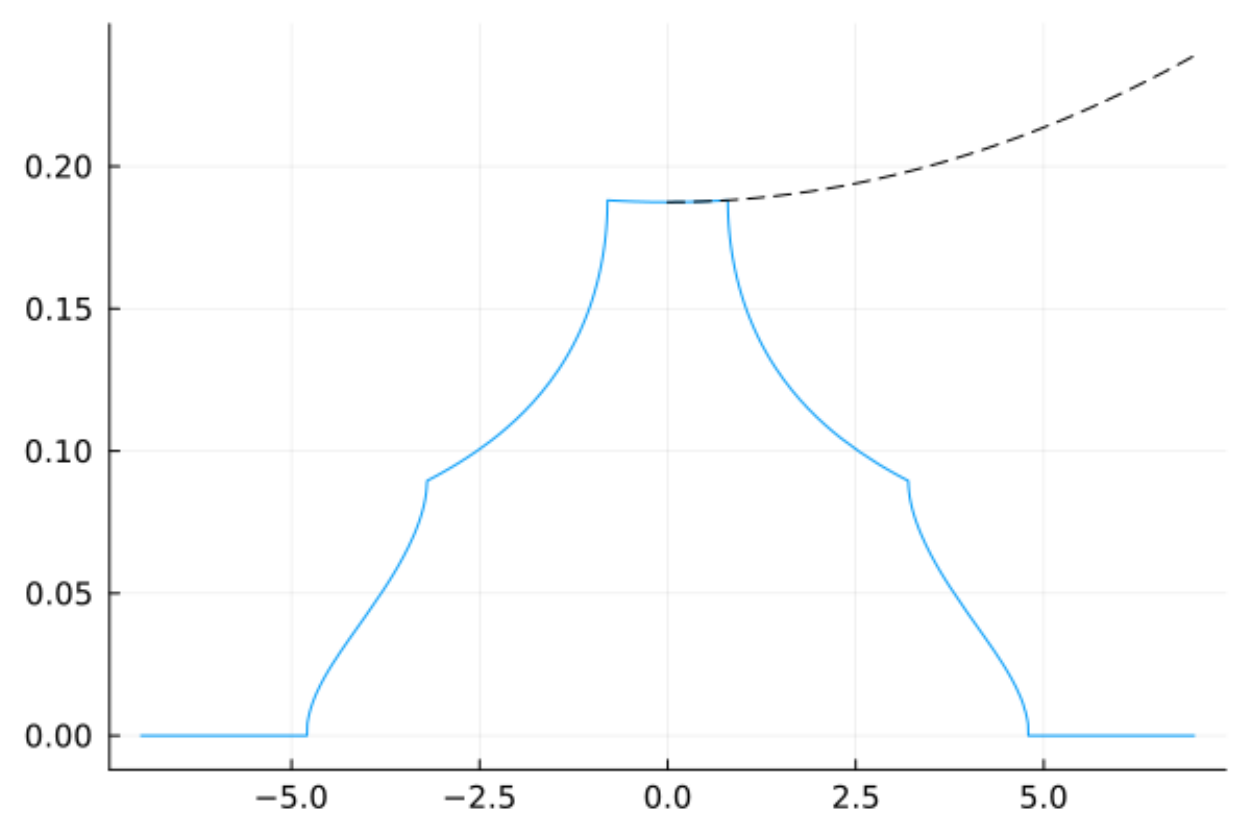}
      \put(2,68){\footnotesize $N_{t_z}(\epsilon)$}
      \put(101,3.7){\footnotesize $\epsilon$}
    \end{overpic}
   \hspace{0.5cm}
\begin{overpic}[width=.46\textwidth]{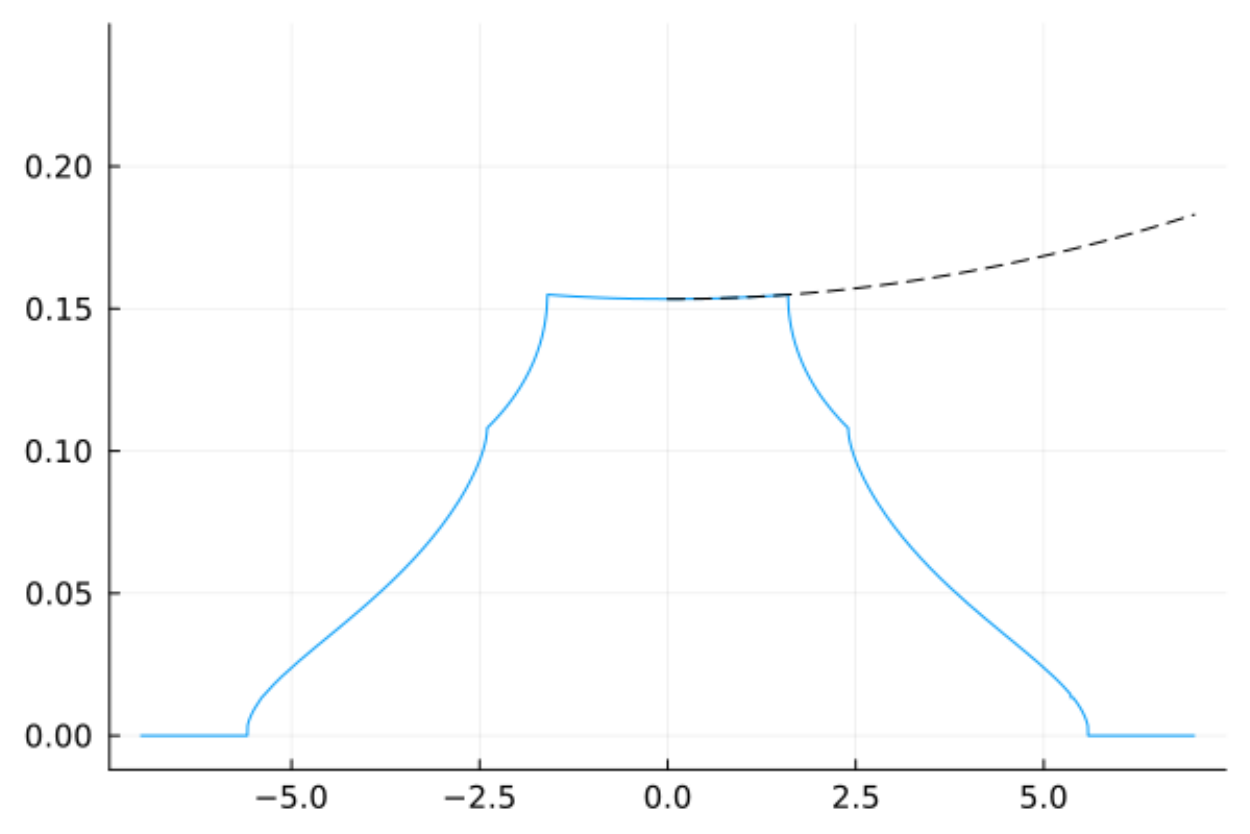}
      \put(2,68){\footnotesize $N_{t_z}(\epsilon)$}
      \put(101,3.7){\footnotesize $\epsilon$}
    \end{overpic}
     \begin{figuretext}[Illustration of Theorem \ref{3Ddensityth}]\label{N3Dleadingtwotz}
       The function $N_{t_z}(\epsilon)$ (solid blue) and the approximation $N_{t_z}(0) + \frac{N_{t_z}''(0)}{2} \epsilon^2$ (dashed black) for $t_z = 0.4$ (left) and $t_z = 0.8$ (right).
      \end{figuretext}
     \end{center}
\end{figure}

Theorem \ref{2Ddensityth} shows that the two-dimensional density of states $N_{0}(\epsilon)$ has a logarithmic singularity at $\epsilon = 0$, see Figure \ref{N0fig}. The next theorem shows that turning on an ever so slight hopping in the $z$-direction makes this singularity disappear, see Figure \ref{N3Dleadingtwotz}. The proof is presented in Section \ref{3Ddensitysubsec}.

 \begin{theorem}[Behavior of $N_{t_z}(\epsilon)$ as $\epsilon \to 0$ for $t_z > 0$]\label{3Ddensityth}
For each $t_z \in (0,2)$, the three-dimensional density of states $N_{t_z}(\epsilon)$ given by (\ref{3Ddensityofstates}) is an even function of $\epsilon$ that is real analytic for $|\epsilon| < \min(2t_z, 4 - 2t_z)$, so that 
\begin{align}\label{Ntzexpansion}
N_{t_z}(\epsilon) = \sum_{j=0}^\infty \frac{N_{t_z}^{(2j)}(0)}{(2j)!}\epsilon^{2j} \qquad \text{for $|\epsilon| < \min(2t_z, 4 - 2t_z)$. }
\end{align}  
Furthermore, the coefficients are given by
\begin{align*}
& N_{t_z}(0) = \int_{-2}^2 \frac{N_0(t_z u)}{\sqrt{4 - u^2}} \frac{du}{\pi},
	\\
& N_{t_z}^{(2j)}(0) = \re \bigg(\int_{\gamma_{-2,2}} \frac{\tilde{N}_0^{(2j)}(t_z u)}{\sqrt{4 - u^2}} \frac{du}{\pi}
\bigg) \qquad \text{for $j = 1,2,\dots$},
\end{align*}
where $\tilde{N}_0(\epsilon)$ is the unique analytic continuation of $N_0(\epsilon)$ from $\epsilon \in (0,4)$ to $\{\epsilon \in \C \,|\, \im \epsilon > 0\}$, and $\gamma_{-2,2}$ is the clockwise semicircle in the upper half-plane starting at $-2$ and ending at $2$. 
\end{theorem}

Our next result provides an asymptotic expansion of the three-dimensional density of states $N_{t_z}(\epsilon)$ evaluated at $\epsilon = 0$ as $t_z \to 0$, see Figure \ref{N3Dateps0fig}. The proof is presented in Section \ref{3Dtzdensitysubsec}.

\begin{theorem}[Behavior of $N_{t_z}(0)$ as $t_z \to 0$]\label{3Dtzdensityth}
As $t_z \downarrow 0$, we have
\begin{align} \nonumber
N_{t_z}(0) = 
& \frac{\ln(\frac{16}{t_z})}{2 \pi^2}
+\frac{t_z^2 (2 \ln (\frac{16}{t_z})-3)}{128 \pi^2}
+ \frac{27 t_z^4 (4 \ln(\frac{16}{t_z})-7)}{2^{16} \pi^2}
	\\ \nonumber
&  +\frac{25 t_z^6 (20 \ln(\frac{16}{t_z})-37)}{2^{21}  \pi^2}
  +\frac{1225 t_z^8 (280 \ln(\frac{16}{t_z})-533)}{2^{33} \pi^2}
	\\ \label{Ntzat0expansion}
 &   +\frac{11907}{5} \frac{t_z^{10} (840 \ln  (\frac{16}{t_z})-1627)}{2^{38} \pi^2}
   + O\Big(t_z^{12} \ln{\tfrac{1}{t_z}}\Big),
\end{align}  
and this expansion can be differentiated termwise with respect to $t_z$ any finite number of times.
\end{theorem}

\begin{figure}
\bigskip
\bigskip
\begin{center}
\begin{overpic}[width=.46\textwidth]{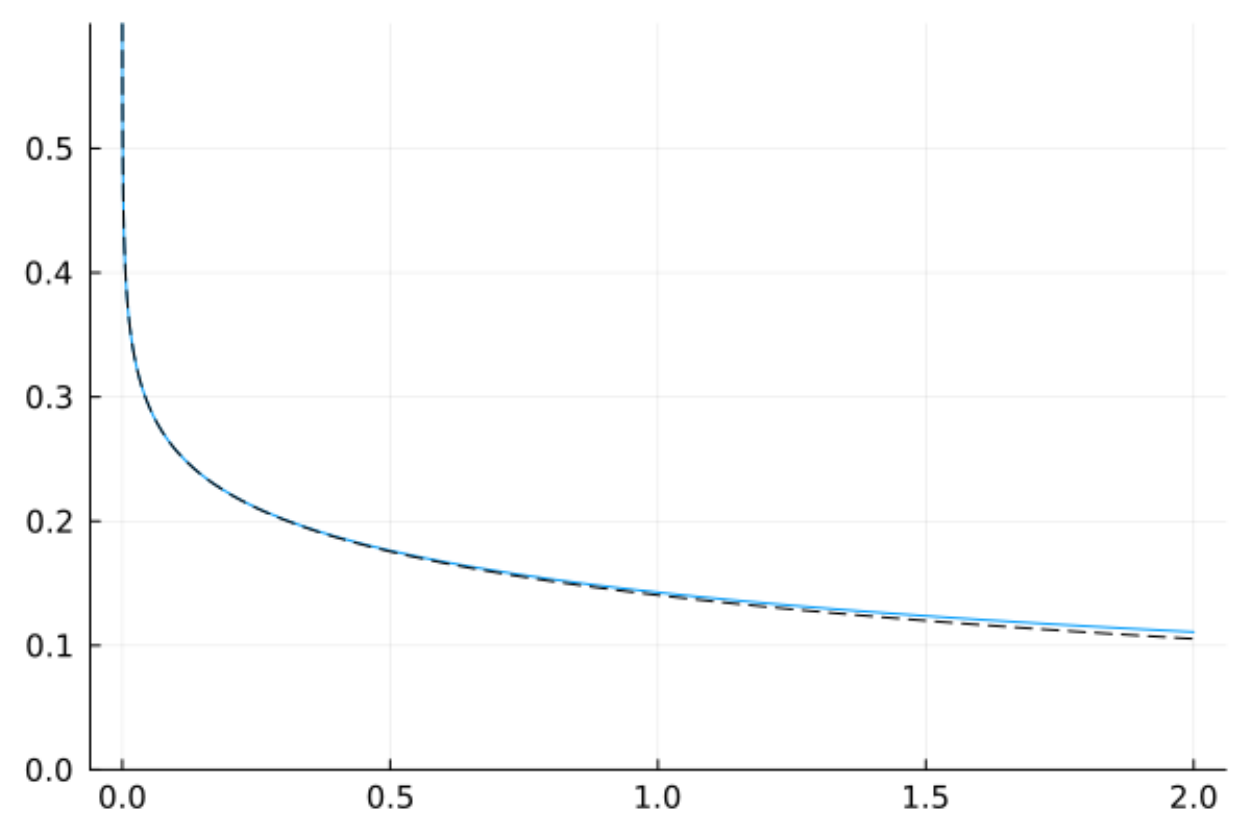}
      \put(2,68){\footnotesize $N_{t_z}(0)$}
      \put(101,3.7){\footnotesize $t_z$}
    \end{overpic}
   \hspace{0.5cm}
\begin{overpic}[width=.46\textwidth]{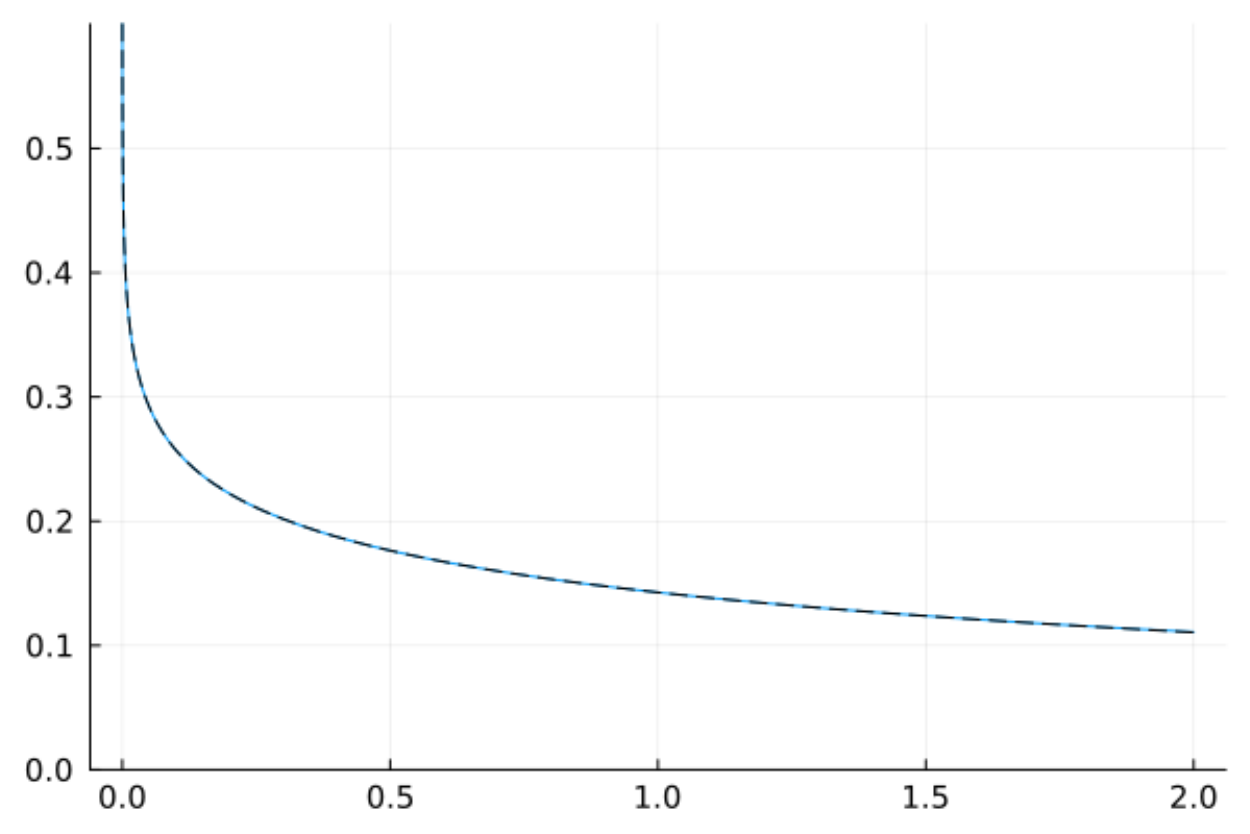}
      \put(2,68){\footnotesize $N_{t_z}(0)$}
      \put(101,3.7){\footnotesize $t_z$}
    \end{overpic}
     \begin{figuretext}[Illustration of Theorem \ref{3Dtzdensityth}]\label{N3Dateps0fig}
       Left: The function $t_z \mapsto N_{t_z}(0)$ (solid blue) and the approximation $\frac{\ln(16/t_z)}{2 \pi^2}$ (dashed black). 
       Right: The function $t_z \mapsto N_{t_z}(0)$ (solid blue) and the asymptotic approximation of (\ref{Ntzat0expansion}) including terms up to order $O(t_z^{10})$ (dashed black). 
       
\end{figuretext}
     \end{center}
\end{figure}

\subsection{The small $U$ behavior of the N\'eel temperature}
Recall that $T_N(U,t_z)$ reduces to the N\'eel temperature for the 2D Hubbard model when $t_z = 0$.
The following theorem gives the small $U$ behavior of $T_N(U,0)$, see Figure \ref{TNfig}. 

\begin{theorem}[Small $U$ behavior of the 2D N\'eel temperature]\label{2DNeelth}
The N\'eel temperature $T_N(U, 0)$ satisfies
\begin{align}\label{TNexpansion2D}
T_N(U,0)
= \frac{32}{\pi e^{-\gamma}}
e^{-\sqrt{ \frac{4 \pi^2}{U}  + a_1}}\Big(1 + O\Big(e^{-\frac{4\pi}{\sqrt{U}}}\Big)\Big) \qquad \text{as $U \downarrow 0$},
\end{align}
where $\gamma \approx 0.5772$ is Euler's gamma constant and the constant $a_1 \approx 0.3260$ is defined by
\begin{align} \nonumber
a_1 = & -4 \pi ^2 \int_{0}^{4} \bigg(N_0(\epsilon) - \frac{\ln(\frac{16}{\epsilon})}{2 \pi^2}\bigg)
\frac{1}{\epsilon}  d\epsilon
	\\ \label{a1def}
& - \int_{0}^\infty \frac{(\ln x)^2}{\cosh^2{x}} dx + (2\ln2)^2 +(\gamma+2\ln2-\ln\pi)^2.
\end{align} 
\end{theorem}
\begin{proof}
See Section \ref{2DNeelsubsec}.
\end{proof}

\begin{figure}
\bigskip
\begin{center}
\begin{overpic}[width=.46\textwidth]{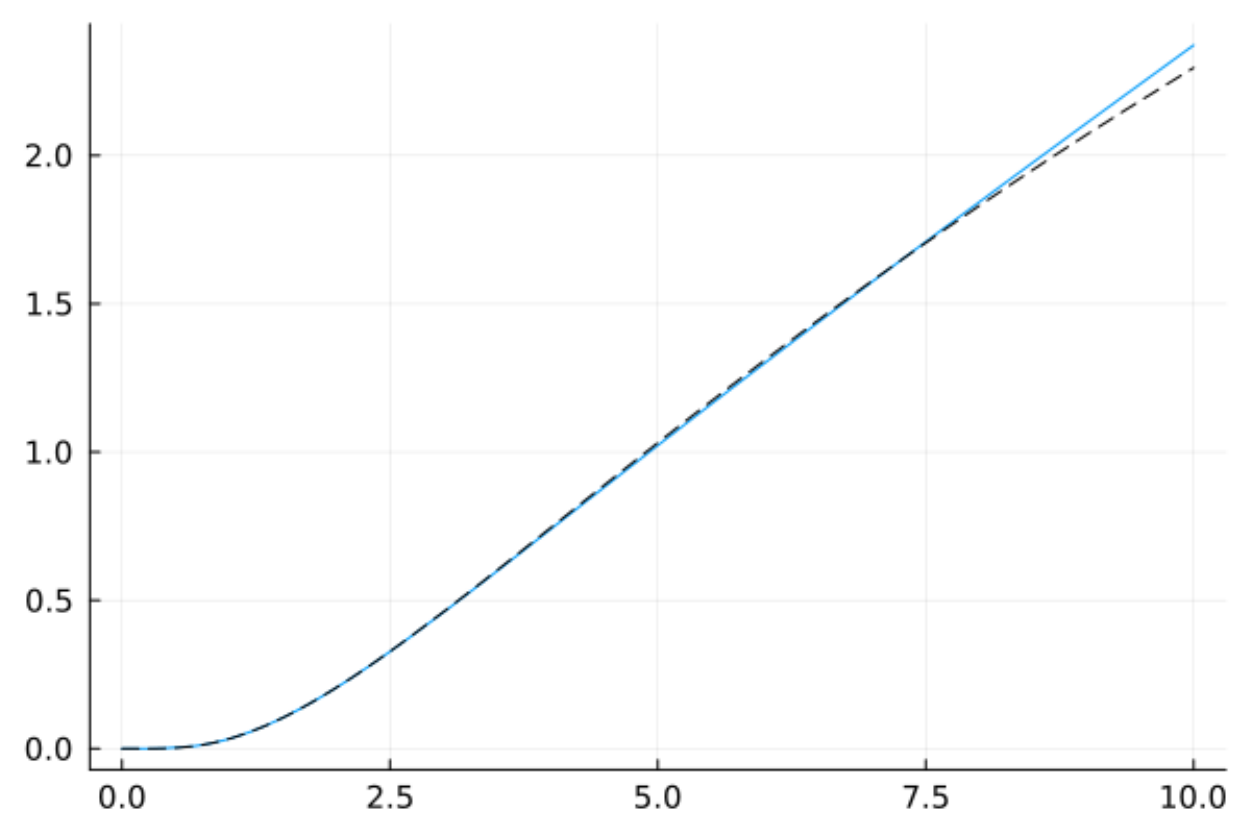}
      \put(0,68){\footnotesize $T_N(U,0)$}
      \put(101,3.7){\footnotesize $U$}
    \end{overpic}
   \hspace{0.5cm}
\begin{overpic}[width=.46\textwidth]{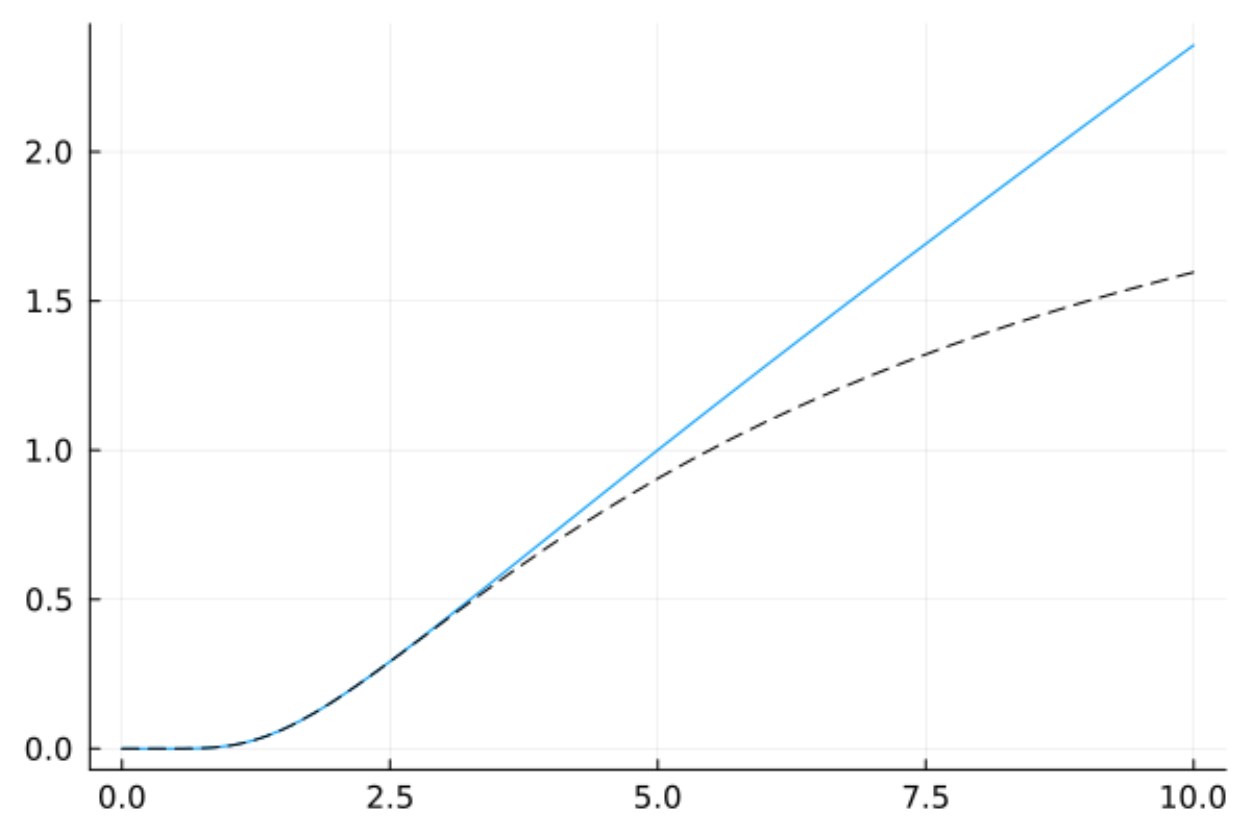}
      \put(0,68){\footnotesize $T_N(U,t_z)$}
      \put(101,3.7){\footnotesize $U$}
    \end{overpic}
     \begin{figuretext}[Illustration of Theorem \ref{2DNeelth} and Theorem \ref{3DNeelth}]\label{TNfig}
       Left: The 2D N\'eel temperature $T_N(U,0)$ (solid blue) and the approximation $\frac{32}{\pi e^{-\gamma}} \exp(-\sqrt{4 \pi^2/U  + a_1})$ (dashed black) of Theorem \ref{2DNeelth}. 
       Right: The 3D N\'eel temperature $T_N(U,t_z)$ (solid blue) and the approximation $\frac{8+4t_z }{\pi e^{-\gamma}}\exp(-N_{t_z}(0)^{-1}(1/U - b_0(t_z)))$ (dashed black) of Theorem \ref{3DNeelth} for $t_z = 1/2$. 
\end{figuretext}
     \end{center}
\end{figure}

Our next theorem gives the small $U$ behavior of the N\'eel temperature for the 3D Hubbard model in the case when $t_z > 0$, see Figure \ref{TNfig}. An important ingredient in the formula is the function $N_{t_z}(0)$, which is shown in Figure \ref{N3Dateps0fig}.

\begin{theorem}[Small $U$ behavior of the 3D N\'eel temperature]\label{3DNeelth}
The N\'eel temperature $T_N(U, t_z)$ satisfies
\begin{align}\label{TNexpansion3D}
T_N(U, t_z) = \frac{8+4t_z }{\pi e^{-\gamma}}e^{- \frac{1}{N_{t_z}(0)}\big(\frac{1}{U} - b_0(t_z)\big)}
\Big(1 + O(e^{-\frac{2}{N_{t_z}(0) U}})\Big)
 \qquad \text{as $U \downarrow 0$}
\end{align}
uniformly for $t_z$ in compact subsets of $(0, 2)$, where $\gamma$ is Euler's gamma constant and
\begin{align}\label{b0tzdef}
& b_0(t_z) = \int_{0}^{4+ 2t_z} \big(N_{t_z}(\epsilon) - N_{t_z}(0)\big)
\frac{1}{\epsilon}  d\epsilon.
\end{align}
\end{theorem}
\begin{proof}
See Section \ref{3DNeelsubsec}.
\end{proof}

\subsection{The small $U$ behavior of the mean field}
Let $\hat{m}(U, t_z, T)$ be the function defined in (\ref{hatmdef}).
Our next theorem gives the small $U$ behavior of $\hat{m}$ for the 2D Hubbard model, i.e., in the case when $t_z = 0$, see Figure \ref{2Dmhatfig}. 

\begin{theorem}[Small $U$ behavior of 2D mean field]\label{2Dmeanfieldth}
The function $\hat{m}(U, 0, T)$ satisfies
\begin{align}\label{2Dhatmexpansion}
\hat{m}(U, 0, T) = f_{\BCS}\bigg(\frac{T}{T_N(U, 0)}\bigg) + c_1\bigg(\frac{T}{T_N(U, 0)}\bigg) \sqrt{U} + O(U) \qquad \text{as $U \downarrow 0$}
\end{align}
uniformly for all $T \geq 0$ such that $\frac{T}{T_N(U, 0)}$ remains in a compact subset of $[0,1)$, where
\begin{align}\label{c1def} 
c_1(y) =   (f_{\BCS}(y) - y f_{\BCS}'(y))
\int_{0}^\infty \frac{\ln(\frac{16}{\epsilon})}{2 \pi}
\bigg(\frac{\tanh(\frac{\sqrt{f_{\BCS}(y)^2 + \epsilon^2}}{2y})}{\sqrt{f_{\BCS}(y)^2 + \epsilon^2}}  
- \frac{\tanh(\frac{\epsilon}{2})}{\epsilon} \bigg) d\epsilon. 
\end{align}
In particular, for $T= 0$,
\begin{align*}
\hat{m}(U, 0, 0) = \pi e^{-\gamma} + c_1(0) \sqrt{U} + O(U) \qquad \text{as $U \downarrow 0$,}
\end{align*}
where $\gamma \approx 0.5772$ is Euler's gamma constant and
\begin{align}\label{c1at0}
c_1(0) = &\;  \pi e^{-\gamma}
\int_{0}^\infty\frac{\ln(\frac{16}{\epsilon})}{2 \pi}  
\bigg(\frac{1}{\sqrt{\pi^2 e^{-2\gamma} + \epsilon^2}}  
- \frac{\tanh(\frac{\epsilon}{2})}{\epsilon} \bigg) d\epsilon
 \approx 0.04576.
\end{align}
\end{theorem}
\begin{proof}
See Section \ref{2Dmeanfieldsubsec}. 
\end{proof}

\begin{figure}
\bigskip
\begin{center}
\begin{overpic}[width=.46\textwidth]{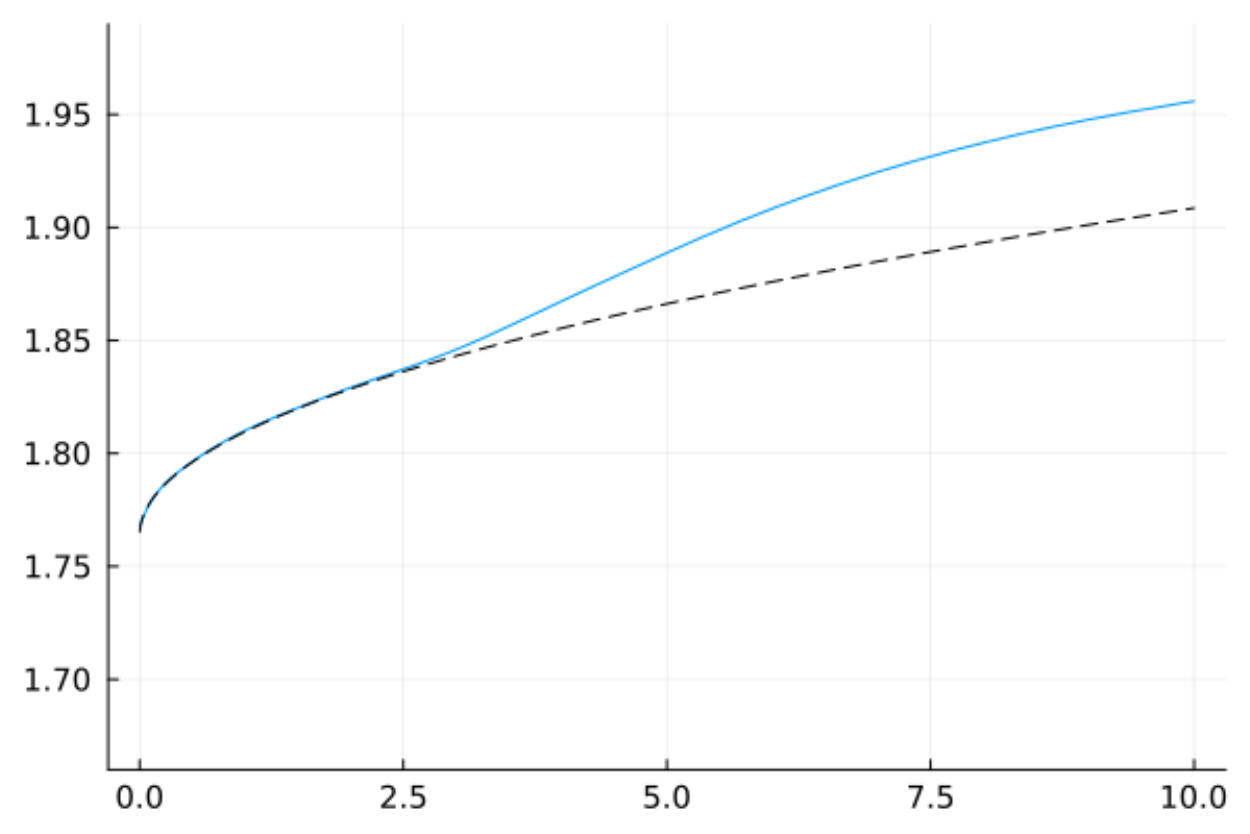}
      \put(-2,68){\footnotesize $\hat{m}(U, 0, yT_N)$}
      \put(101,3.7){\footnotesize $U$}
    \end{overpic}
   \hspace{0.5cm}
\begin{overpic}[width=.46\textwidth]{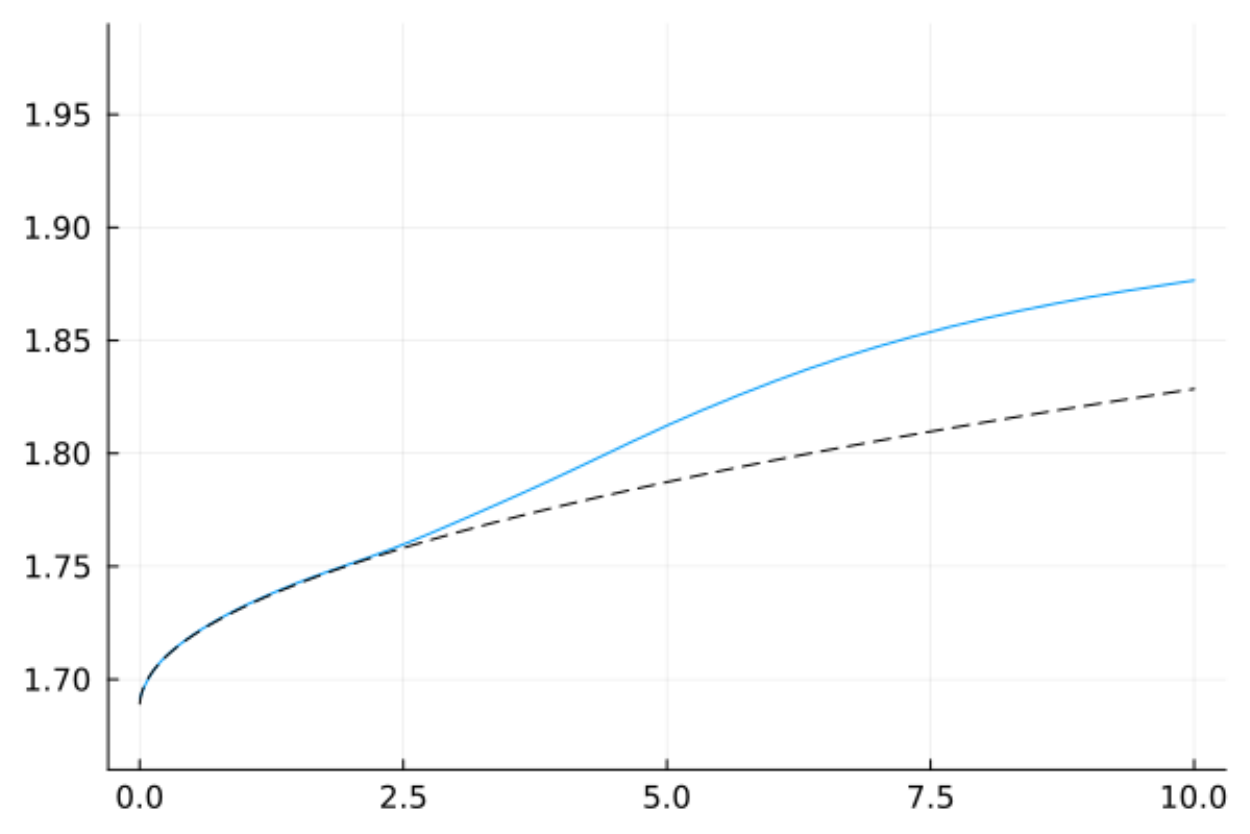}
      \put(-2,68){\footnotesize $\hat{m}(U, 0, yT_N)$}
      \put(101,3.7){\footnotesize $U$}
    \end{overpic}
     \begin{figuretext}[Illustration of Theorem \ref{2Dmeanfieldth}]\label{2Dmhatfig}
       The function $U \mapsto \hat{m}(U, 0, yT_N(U,0))$ (solid blue) and the approximation $f_{\BCS}(y) + c_1(y) \sqrt{U}$ (dashed black) of Theorem \ref{2Dmeanfieldth} for $y = 0$ (left) and $y = 1/2$ (right). 
\end{figuretext}
     \end{center}
\end{figure}

Our next theorem gives the small $U$ behavior of $\hat{m}$ in the case of $t_z > 0$, see Figure \ref{3Dmhatfig}. 

\begin{theorem}[Small $U$ behavior of 3D mean field]\label{3Dmeanfieldth}
The function $\hat{m}(U, t_z, T)$ satisfies
\begin{align*}
\hat{m}(U, t_z, T) = f_{\BCS}\bigg(\frac{T}{T_N(U, t_z)}\bigg) + O\Big(U^{-1} e^{- \frac{2}{N_{t_z}(0) U}} \Big) \qquad \text{as $U \downarrow 0$}
\end{align*}
uniformly for $t_z$ in compact subsets of $(0, 2)$ and for $T \geq 0$ such that $\frac{T}{T_N(U, t_z)}$ remains in a compact subset of $[0,1)$.
\end{theorem}
\begin{proof}
See Section \ref{3Dmeanfieldsubsec}. 
\end{proof}

\begin{figure}
\bigskip
\begin{center}
\begin{overpic}[width=.46\textwidth]{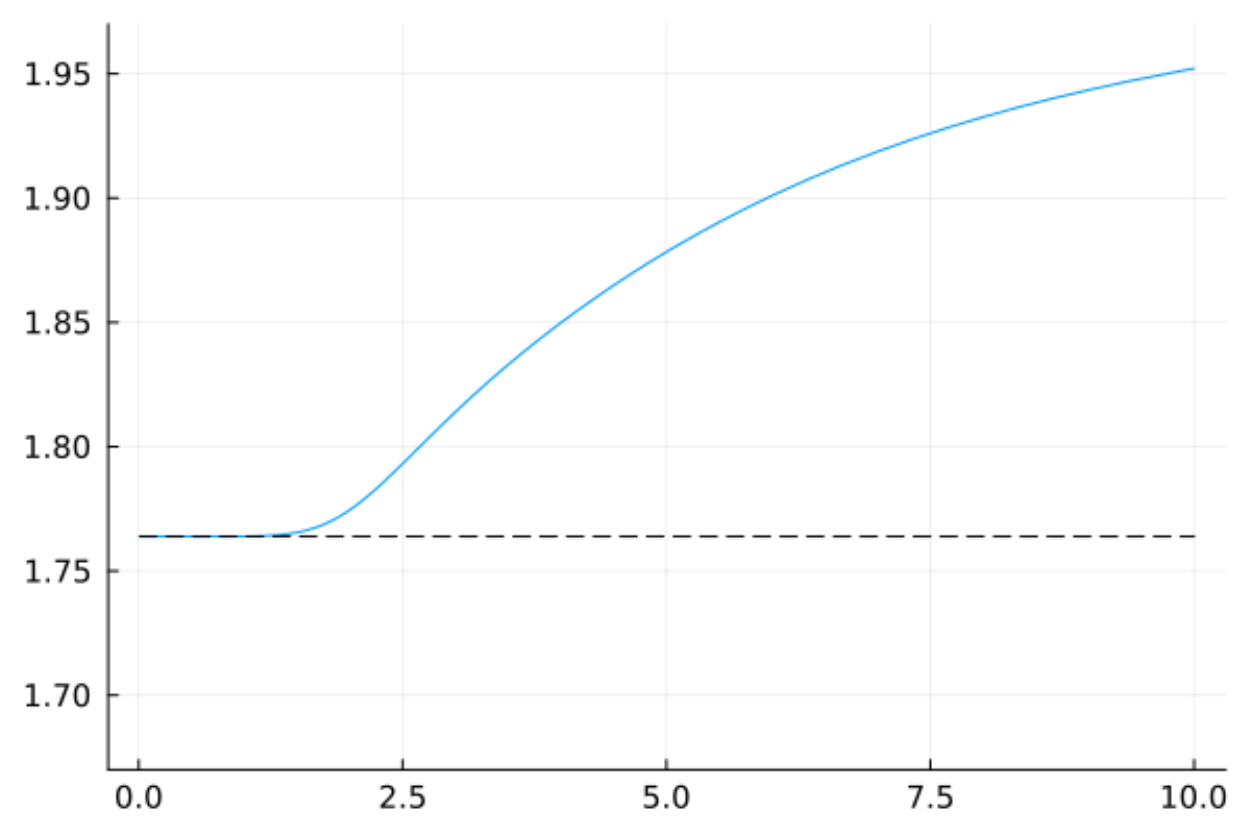}
      \put(-2,68){\footnotesize $\hat{m}(U, t_z, yT_N)$}
      \put(101,3.7){\footnotesize $U$}
    \end{overpic}
   \hspace{0.5cm}
\begin{overpic}[width=.46\textwidth]{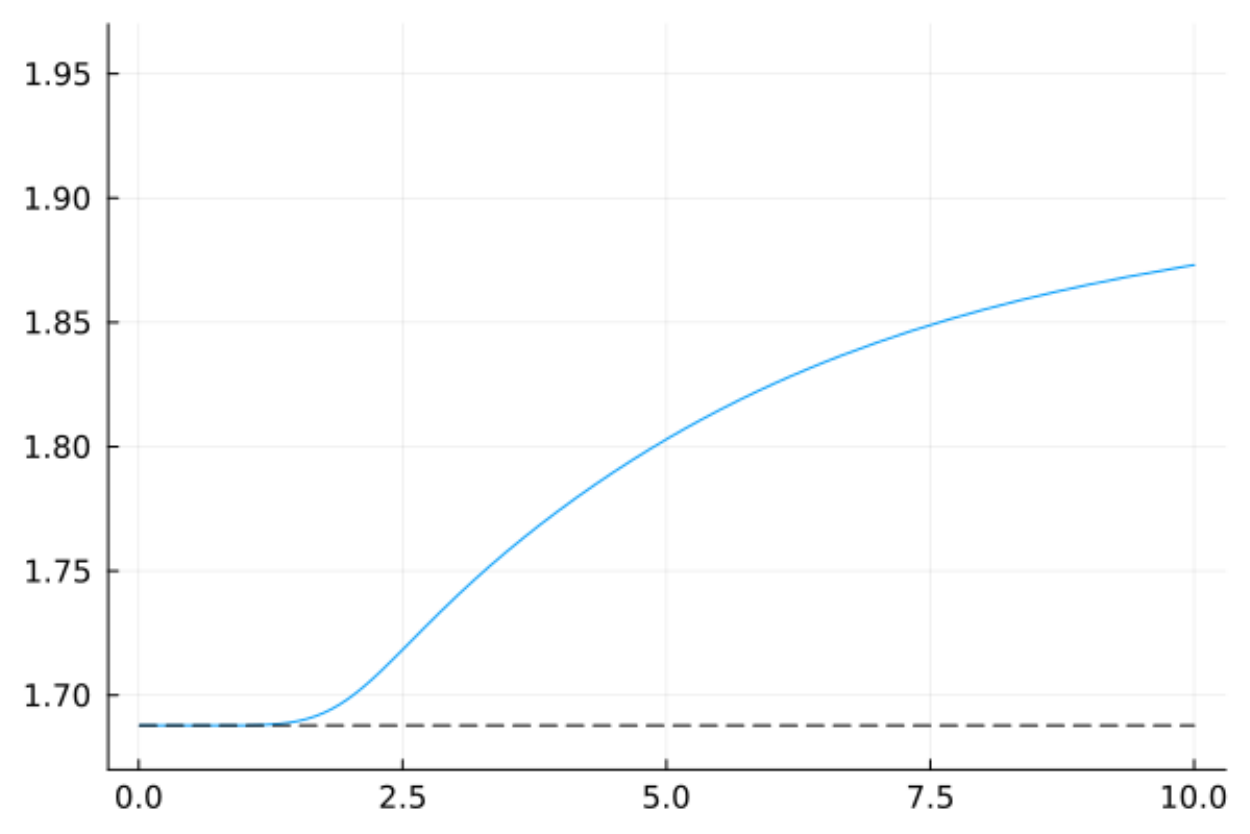}
      \put(-2,68){\footnotesize $\hat{m}(U, t_z, yT_N)$}
      \put(101,3.7){\footnotesize $U$}
    \end{overpic}
     \begin{figuretext}[Illustration of Theorem \ref{3Dmeanfieldth}]\label{3Dmhatfig}
       The function $U \mapsto \hat{m}(U, t_z, yT_N(U,t_z))$ (solid blue) and the approximation $f_{\BCS}(y)$ (dashed black) of Theorem \ref{3Dmeanfieldth} for $y = 0$ (left) and $y = 1/2$ (right) in the case of $t_z = 1/2$. 
\end{figuretext}
     \end{center}
\end{figure}

\begin{figure}
\bigskip
\begin{center}
\begin{overpic}[width=.46\textwidth]{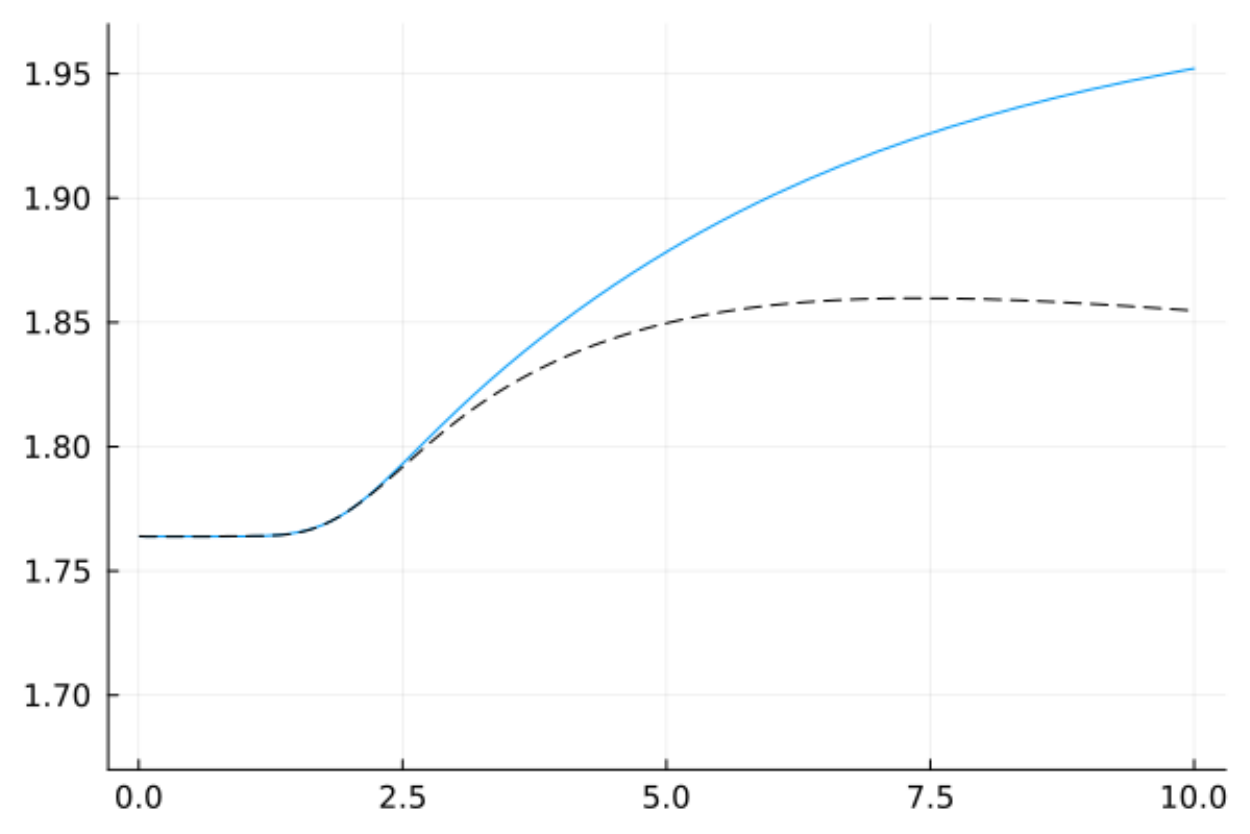}
      \put(-2,68){\footnotesize $\hat{m}(U, t_z, yT_N)$}
      \put(101,3.7){\footnotesize $U$}
    \end{overpic}
   \hspace{0.5cm}
\begin{overpic}[width=.46\textwidth]{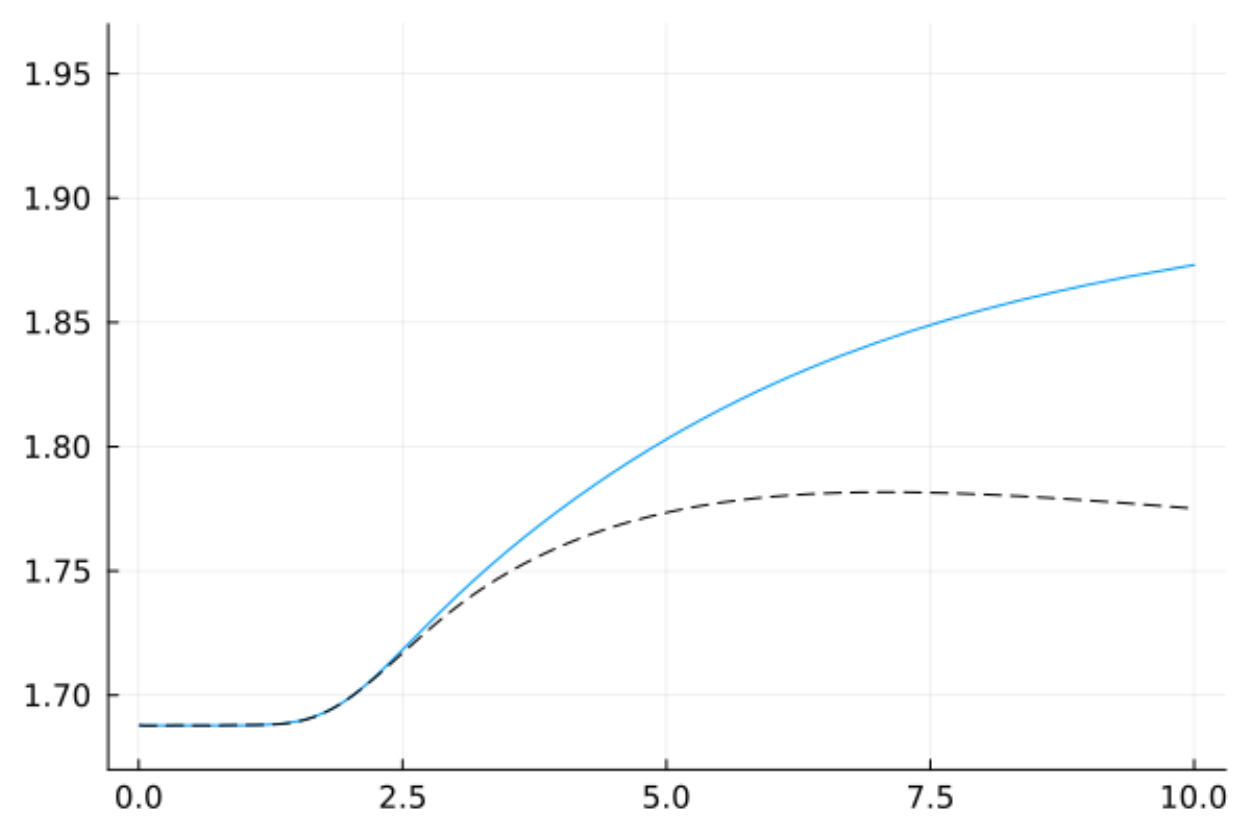}
      \put(-2,68){\footnotesize $\hat{m}(U, t_z, yT_N)$}
      \put(101,3.7){\footnotesize $U$}
    \end{overpic}
     \begin{figuretext}[Illustration of Theorem \ref{3Dmeanfieldimprovedth}]\label{3Dmhatimprovedfig}
       The function $U \mapsto \hat{m}(U, t_z, yT_N(U,t_z))$ (solid blue) and the approximation $f_{\BCS}(y) + g(U,t_z,y)$ (dashed black) of Theorem \ref{3Dmeanfieldth} for $y = 0$ (left) and $y = 1/2$ (right) in the case of $t_z = 1/2$. 
\end{figuretext}
     \end{center}
\end{figure}

\subsection{An exponentially small subleading term}
Theorems \ref{2Dmeanfieldth} and \ref{3Dmeanfieldth} provide approximations of $\hat{m}$ as $U \to 0$ in two and three dimensions. At the expense of introducing more complicated formulas, it is possible to include higher order terms in these approximations. To illustrate this point, we derive in this section an exponentially small subleading term in the expansion of $\hat{m}(U, t_z, T)$, thereby improving the formula of Theorem \ref{3Dmeanfieldth}.

\begin{theorem}[Small $U$ behavior of 3D mean field, improved version]\label{3Dmeanfieldimprovedth}
The function $\hat{m}(U, t_z, T)$ satisfies (see Figure \ref{3Dmhatimprovedfig})
\begin{align*}
\hat{m}(U, t_z, T) = f_{\BCS}\bigg(\frac{T}{T_N(U, t_z)}\bigg) 
+ g\bigg(U, t_z, \frac{T}{T_N(U, t_z)}\bigg)
+ O\Big(U^{-2} e^{- \frac{4}{N_{t_z}(0) U}} \Big) \qquad \text{as $U \downarrow 0$}
\end{align*}
uniformly for $t_z$ in compact subsets of $(0, 2)$ and for $T \geq 0$ such that $\frac{T}{T_N(U, t_z)}$ remains in a compact subset of $[0,1)$, where the subleading term is given by
\begin{align*}
g(U, t_z, y) = \big(f_{\BCS}(y) - y f_{\BCS}'(y)\big)
\int_0^\infty & \frac{N_{t_z}(T_{N}(U,t_z) \epsilon)  - N_{t_z}(0)}{N_{t_z}(0)} 
	\\
& \times \bigg(\frac{\tanh(\frac{\sqrt{f_{\BCS}(y)^2 + \epsilon^2}}{2y})}{\sqrt{f_{\BCS}(y)^2 + \epsilon^2}}  -
\frac{\tanh(\frac{\epsilon}{2})}{ \epsilon}\bigg) d \epsilon.
\end{align*}
In particular, for $T = 0$,
\begin{align*}
\hat{m}(U, t_z, 0) = \pi e^{-\gamma} 
+ g(U, t_z, 0)
+ O\Big(U^{-2} e^{- \frac{4}{N_{t_z}(0) U}} \Big) \qquad \text{as $U \downarrow 0$}
\end{align*}
uniformly for $t_z$ in compact subsets of $(0, 2)$, where
\begin{align*}
g(U, t_z, 0) = \pi e^{-\gamma} 
\int_0^\infty & \frac{N_{t_z}(T_{N}(U,t_z) \epsilon)  - N_{t_z}(0)}{N_{t_z}(0)} 
\bigg(\frac{1}{\sqrt{\pi^2 e^{-2\gamma} + \epsilon^2}}  -
\frac{\tanh(\frac{\epsilon}{2})}{ \epsilon}\bigg) d \epsilon.
\end{align*}
\end{theorem}
\begin{proof}
See Section \ref{3Dmeanfieldimprovedsubsec}. 
\end{proof}

\section{Antiferromagnetic mean-field equation}\label{derivationsec}

\subsection{The Hubbard model}\label{sec:Hubbard} 
The Hubbard model describes itinerant electrons with spin $\sigma \in \{\uparrow, \downarrow\}$ on a lattice $\Lambda_L$. In this paper, $\Lambda_L$  will be a square lattice or a cubic lattice, i.e.,
$$\Lambda_L = \bigg\{(x_1, \dots, x_n) \in \Z^n \,\bigg|\, -\frac{L}{2} \leq x_j < \frac{L}{2} \;\; \text{for $j =1,\dots, n$}\bigg\}$$
for $n = 2$ or $n = 3$; we use periodic boundary conditions, i.e., we identity $\Lambda_L$ with the torus $(\Z/L\Z)^n$.
A single particle has to be at one of the $L^n$ lattice sites and can have either spin up or spin down; thus the single-particle Hilbert space $\mathcal{H}$ is isomorphic to $\C^{2L^n}$. The {\it fermion Fock space $\mathcal{F}$} is by definition the Grassmann algebra 
$$\mathcal{F} = \bigwedge \mathcal{H} = \bigwedge\nolimits^{0} \mathcal{H} \oplus \bigwedge\nolimits^1 \mathcal{H} \oplus \cdots \oplus \bigwedge\nolimits^{\dim \mathcal{H}} \mathcal{H}$$
and is therefore isomorphic to $\C^{4^{L^n}}$.
The Hubbard model is defined by the Hamiltonian operator $H:\mathcal{F} \to \mathcal{F}$ given by
\begin{align}\label{hubbardhamiltonian}
H =  \sum_{\mathbf{x}, \mathbf{y} \in \Lambda_{L}} \sum_{\sigma = \uparrow, \downarrow} t_{\mathbf{x}\mathbf{y}} c_{\mathbf{x} \sigma}^\dagger c_{\mathbf{y} \sigma} + U \sum_{\mathbf{x} \in \Lambda_L} n_{\mathbf{x} \uparrow} n_{\mathbf{x} \downarrow}
- \frac{U}{2} \mathcal{N} + \frac{U}{4}L^n,
\end{align}
where the coupling parameter $U > 0$ determines the strength of the on-site repulsion, $t_{\mathbf{x}\mathbf{y}}$ are the hopping matrix elements, $c_{\mathbf{x} \sigma}^\dagger$ is the operator on $\mathcal{F}$ which creates an electron of spin $\sigma \in \{\uparrow, \downarrow\}$ at site $\mathbf{x} \in \Lambda_L$, $c_{\mathbf{x} \sigma}$ is the corresponding annihilation operator, $n_{\mathbf{x} \sigma} := c_{\mathbf{x} \sigma}^\dagger c_{\mathbf{x} \sigma}$ is the density operator, and $\mathcal{N} := \sum_{\mathbf{x}, \sigma} n_{\mathbf{x} \sigma}$ is the number operator. 
The last two terms in (\ref{hubbardhamiltonian}) can be removed by shifting the ground state energy and the chemical potential $\mu$, but it is convenient to include these terms because they make particle-hole symmetry manifest and ensure that half-filling corresponds to $\mu = 0$. Another way to write \eqref{hubbardhamiltonian} is 
$$
H =  \sum_{\mathbf{x}, \mathbf{y} \in \Lambda} \sum_{\sigma = \uparrow, \downarrow} t_{\mathbf{x}\mathbf{y}} c_{\mathbf{x} \sigma}^\dagger c_{\mathbf{y} \sigma} + U \sum_{\mathbf{x} \in \Lambda_L} (n_{\mathbf{x}\uparrow} -\tfrac12)(n_{\mathbf{x} \downarrow}-\tfrac12).$$

In this paper, we consider the case of a nearest neighbor hopping matrix given in two dimensions by 
\begin{align}\label{txy2D}
\text{$n = 2$}: \qquad
t_{\mathbf{x}\mathbf{y}} 
= \begin{cases} -t & \text{if $\mathbf{x}$ and $\mathbf{y}$ are nearest neighbors},
	\\
0 & \text{otherwise}.
\end{cases}
\end{align}
In three dimensions, we allow for a different hopping constant in the $z$-direction, so that 
\begin{align}\label{txy3D}
\text{$n = 3$}: \qquad
t_{\mathbf{x}\mathbf{y}} 
= \begin{cases} 
-t & \text{if $\mathbf{x}$ and $\mathbf{y}$ are nearest neighbors and $x_3 = y_3$},
	\\
-t_z & \text{if $\mathbf{x}$ and $\mathbf{y}$ are nearest neighbors and $x_3 = y_3 \pm 1$},
	\\
0 & \text{otherwise}.
\end{cases}
\end{align}

Instead of using the basis $\{|\mathbf{x} \sigma\rangle\}_{\mathbf{x} \in \Lambda_L, \sigma \in \{\uparrow, \downarrow\}}$ for the single-particle Hilbert space $\mathcal{H} \cong \C^{2L^2}$, it is often convenient to use the basis $\{|\mathbf{k} \sigma\rangle\}_{\mathbf{k} \in \Lambda_L^*, \sigma \in \{ \uparrow, \downarrow\}}$ defined by
\begin{align}\label{kbasisdef}
\langle \mathbf{x}\sigma |\mathbf{k} \sigma' \rangle = \frac{1}{L^{n/2}} e^{i\mathbf{k} \cdot \mathbf{x}} \delta_{\sigma \sigma'}
\end{align}
with the dual lattice $\Lambda_L^*$ given by
$$\Lambda_L^* = \bigg\{(k_1, \dots, k_n) \in \frac{2\pi}{L}\Z^n \,\bigg|\, -\pi \leq k_j < \pi\;\; \text{for $j =1,\dots, n$}\bigg\}.$$
Under the change of basis from $\{|\mathbf{x}\sigma\rangle\}$ to $\{|\mathbf{k}\sigma\rangle\}$, the creation and annihilation operators transform as
\begin{align*}
c_{\mathbf{x} \sigma}^\dagger = \sum_{\mathbf{k}} c_{\mathbf{k} \sigma}^\dagger \langle \mathbf{k}\sigma |\mathbf{x} \sigma \rangle
= \frac{1}{L^{n/2}}  \sum_{\mathbf{k}} c_{\mathbf{k} \sigma}^\dagger e^{-i\mathbf{k} \cdot \mathbf{x}}, \quad
c_{\mathbf{x} \sigma} = \sum_{\mathbf{k}}  \langle \mathbf{x}\sigma |\mathbf{k} \sigma \rangle c_{\mathbf{k} \sigma}
= \frac{1}{L^{n/2}}  \sum_{\mathbf{k}} c_{\mathbf{k} \sigma} e^{i\mathbf{k} \cdot \mathbf{x}},
\end{align*}
and using these relations we arrive at a Fourier space representation for the Hubbard Hamiltonian:
$$H = H_0 + H_{\text{int}} - \frac{U}{2} \mathcal{N} + \frac{U}{4}L^n,$$
where
\begin{align}\label{HubbardH0}
& H_0 := 
\sum_{\mathbf{k} \in \Lambda_L^*}  \sum_{\sigma = \uparrow, \downarrow}   \varepsilon(\mathbf{k}) c_{\mathbf{k} \sigma}^\dagger c_{\mathbf{k} \sigma}, 
	\\
& H_{\text{int}} := \frac{U}{L^n}  \sum_{\mathbf{k}^{(1)}, \mathbf{k}^{(2)}, \mathbf{k}^{(3)}, \mathbf{k}^{(4)} \in \Lambda_L^*} 
c_{\mathbf{k}^{(1)} \uparrow}^\dagger 
 c_{\mathbf{k}^{(3)} \downarrow}^\dagger 
c_{\mathbf{k}^{(4)} \downarrow}
c_{\mathbf{k}^{(2)} \uparrow} 
 \sum_{\mathbf{n} \in \Z^n} \delta_{\mathbf{k}^{(1)} - \mathbf{k}^{(2)} + \mathbf{k}^{(3)} - \mathbf{k}^{(4)}, 2\pi\mathbf{n}},
\end{align}
with $\varepsilon(\mathbf{k})$ given by (\ref{epsilondef}).

\subsection{Hartree--Fock theory for the Hubbard model}
We review some relevant aspects of Hartree--Fock theory for the Hubbard model from \cite{BLS1994, BP1996}.

An {\it observable} $A$ is a self-adjoint linear map $A:\mathcal{F} \to \mathcal{F}$. 
A {\it state} (also known as a {\it density matrix}) is a linear map $\rho:\mathcal{F} \to \mathcal{F}$ such that $\rho \geq 0$ and $\tr \rho = 1$. The {\it expectation value} of an observable $A$ in the state $\rho$ is $\langle A \rangle_\rho := \tr(\rho A)$. 
We consider the grand canonical ensemble where the system is in contact with a heat and particle reservoir. In this case, the most probable state of the system for a given inverse temperature $\beta  =1/T \in (0, +\infty]$ and chemical potential $\mu \in \R$ is the state that minimizes the {\it grand canonical potential} defined by
$$\Omega_{\qm}(\rho) = \tr[(H - \mu \mathcal{N})\rho] + \frac{1}{\beta} \tr[\rho \ln \rho].$$
For each $\beta$ and $\mu$, there is a unique state $\rho^{(\beta,\mu)}$ at which $\Omega_{\qm}(\rho)$ is minimized, i.e.,
$$\Omega_{\qm}(\rho^{(\beta,\mu)}) = \min_{\rho \geq 0, \; \tr(\rho) = 1} \Omega_{\qm}(\rho);$$
this follows from the finite dimensionality of $\mathcal{F}$ and the strict convexity of $\Omega_{\qm}$.
The unique minimizer $\rho^{(\beta,\mu)}$ is known as the {\it Gibbs state} and is given in terms of the grand canonical partition function $Z := \tr[ e^{-\beta(H - \mu \mathcal{N})}]$ by
$$\rho^{(\beta,\mu)} = \frac{e^{-\beta(H - \mu \mathcal{N})}}{Z}.$$

It is exceedingly difficult to determine the equilibrium state $\rho^{(\beta,\mu)}$ in general. However, important progress can be made within the so-called Hartree--Fock approximation. In this approximation, one does not minimize grand canonical potential over all states $\rho$. Instead one minimizes $\Omega_{\qm}(\rho)$ over the set of quasi-free states, where a state $\rho$ is {\it quasi-free} if all correlation functions can be computed with the help of Wick's theorem in the sense that
$$\langle e_1e_2\cdots e_{2N-1} \rangle_\rho = 0$$
and
$$\langle e_1e_2\cdots e_{2N} \rangle_\rho = \sum_\pi' \sgn(\pi) \langle e_{\pi(1)\pi(2)}\rangle_\rho \cdots \langle e_{\pi(2N-1)\pi(2N)}\rangle_\rho$$
whenever $e_1, e_2, \dots, e_{2N}$ are operators such that each $e_j$ is either a creation or an annihilation operator, and the sum runs over all permutations $\pi$ such that $\pi(1) < \pi(3) < \cdots < \pi(2N-1)$ and $\pi(2j-1) < \pi(2j)$ for $1 \leq j \leq N$. 
In particular, if $c_j^\dagger$ and $c_j$ are the creation and annihilation operators with respect to a basis $\{|f_j \rangle\}$ for the single particle Hilbert space $\mathcal{H}$, it holds that
$$\langle c_j^\dagger c_k^\dagger c_l c_m \rangle_\rho
= \langle c_j^\dagger c_k^\dagger  \rangle_\rho \langle c_l c_m \rangle_\rho
- \langle c_j^\dagger c_l \rangle_\rho \langle  c_k^\dagger c_m \rangle_\rho
+ \langle c_j^\dagger c_m \rangle_\rho \langle c_k^\dagger c_l \rangle_\rho.$$

Clearly, any quasi-free state is uniquely determined by the values of terms of the form $\langle c_j^\dagger c_k^\dagger \rangle_\rho = \overline{\langle c_k c_j  \rangle_\rho}$ and $\langle c_j^\dagger c_k \rangle_\rho$. 
In fact, the repulsive Hubbard Hamiltonian (\ref{hubbardhamiltonian}) has a positive definite interaction term, and therefore it holds that if $\rho$ is a quasi-free state that minimizes $\Omega_{\qm}(\rho)$ within the set of quasi-free states, then $\rho$ is {\it particle-number conserving} in the sense that $\langle c_k^\dagger c_l^\dagger \rangle_\rho = 0$ for all $k$ and $l$, see  \cite[Theorem 2.11]{BLS1994}. It is therefore enough to consider particle-number conserving quasi-free states, and such states are in one-to-one correspondence with the set of {\it one-particle density matrices}, which are defined as linear maps $\gamma: \mathcal{H} \to \mathcal{H}$ such that $0 \leq \gamma \leq 1$; the correspondence is given explicitly by $\rho \mapsto \gamma$ where $\gamma$ is defined by $\langle f_j | \gamma f_k\rangle = \langle c_j^\dagger c_k \rangle_\rho$, see \cite[Lemma 1]{BP1996}. 

The {\it HF grand canonical potential} $\Omega$ is defined by
$$\Omega(\gamma) = \Omega_{\qm}(\rho_\gamma),$$
where $\rho_\gamma$ is the quasi-free state corresponding to the one-particle density matrix $\gamma$.
For fixed $\beta$ and $\mu$, any one-particle density matrix $\gamma^{(\beta, \mu)}$ that minimizes $\Omega(\gamma)$ is known as a {\it HF Gibbs state}. Since $\mathcal{H}$ is finite-dimensional, the set of one-particle density matrices is compact; thus there exists at least one HF Gibbs state, but it is not necessarily unique. We let
$$\Omega^{(\beta, \mu)} := \Omega(\gamma^{(\beta, \mu)}) = \min_{0 \leq \gamma \leq 1} \Omega (\gamma)$$
denote the minimum value of the HF grand canonical potential for any given $\beta$ and $\mu$. 

The next lemma shows that the expectation value of a one-body operator in a quasi-free state $\rho_\gamma$ can be expressed in terms of $\gamma$ alone. 

\begin{lemma}\label{rhogammalemma}
Let $\rho_\gamma$ be the quasi-free state corresponding to a one-particle density matrix $\gamma$.
If $L:\mathcal{H}\to \mathcal{H}$ is a self-adjoint linear map, and $\hat{L} = \sum_{j,k} \langle f_j |L | f_k\rangle c_j^\dagger c_k:\mathcal{F} \to \mathcal{F}$ is the second quantization of $L$, then 
$$\langle \hat{L} \rangle_{\rho_\gamma} = \tr(\gamma L).$$
\end{lemma}
\begin{proof}
Since $L$ is self-adjoint, there exists a basis $\{| \lambda \rangle\}$ of eigenvectors of $L$ for $\mathcal{H}$. In this basis, $L$ is diagonal and $\hat{L} = \sum_{\lambda} \langle \lambda | L | \lambda \rangle c_\lambda^\dagger c_\lambda$.
Thus, recalling that $\langle \lambda | \gamma | \lambda \rangle = \langle c_\lambda^\dagger c_\lambda \rangle_{\rho_\gamma}$, 
\begin{align*}
  \tr_\mathcal{H} (\gamma L )
  &=   \sum_{\lambda} \langle \lambda | \gamma | \lambda \rangle \langle \lambda |L| \lambda \rangle
  =   \sum_{\lambda} \langle c_\lambda^\dagger c_\lambda \rangle_{\rho_{\gamma}}\langle \lambda |L | \lambda \rangle
   =   \sum_{\lambda} \tr_{\mathcal{F}}\big[ \rho_{\gamma} c_\lambda^\dagger c_\lambda\big] \langle \lambda |L | \lambda \rangle
	\\
&   =   \tr_{\mathcal{F}}\bigg[ \rho_{\gamma} \sum_\lambda \langle \lambda |L | \lambda \rangle c_\lambda^\dagger c_\lambda\bigg] 
  =   \tr_{\mathcal{F}}[ \rho_{\gamma} \hat{L}] = \langle \hat{L} \rangle_{\rho_\gamma},
\end{align*}
which is the desired conclusion.
\end{proof}

The Gibbs state $\rho^{(\beta,\mu)}$ is a linear map on the $4^{L^n}$-dimensional Fock space $\mathcal{F}$, whereas a HF Gibbs state $\gamma^{(\beta, \mu)}$ is a linear map on the $2L^n$-dimensional  single-particle Hilbert space $\mathcal{H}$. The problem of determining a HF Gibbs state therefore involves minimization over far fewer parameters than the problem of determining the true Gibbs state. It turns out that it is possible to reduce the number of independent parameters in the minimization problem even further (from $4L^{2n}$ to $4L^n$ real parameters) by applying a result of \cite{BP1996}. To this end, define the compact sets $K$ and $B$ by
\begin{align*}
& K = \{d:\Lambda_L \to \R \,|\; \text{$|d(\mathbf{x})| \leq 1$ for all $\mathbf{x} \in \Lambda_L$}\},
	\\
& B = \{\vec{m}:\Lambda_L \to \R^3 \,|\; \text{$|\vec{m}(\mathbf{x})| \leq 1$ for all $\mathbf{x} \in \Lambda_L$}\}.
\end{align*}
Let $\vec{\sigma} = (\sigma_1, \sigma_2, \sigma_3)$ be the vector of Pauli matrices, let $h$ be the linear map on $\mathcal{H}$ given in the basis $\{|\mathbf{x} \sigma\rangle\}$ by the $2L^2 \times 2L^2$ matrix 
\begin{align}\label{hdef}
h_{\mathbf{x}\tau, \mathbf{y}\lambda} = t_{\mathbf{x}\mathbf{y}}  \delta_{\tau \lambda} - \mu\delta_{\mathbf{x}\mathbf{y}} \delta_{\tau \lambda} + \frac{U}{2}  \delta_{\mathbf{x}\mathbf{y}} \big(d(\mathbf{x})\delta_{\tau \lambda} - \vec{m}(\mathbf{x}) \cdot \vec{\sigma}_{\tau \lambda}\big),
\end{align}
and consider the functional $\mathcal{G}:K \times B \to \R$ defined by
\begin{align}\label{calGdef}
\mathcal{G}(d, \vec{m}) = -\frac{1}{\beta}\tr\bigg[\ln \cosh\bigg(\frac{\beta h}{2}\bigg)\bigg] + \frac{U}{4} \sum_{\mathbf{x} \in \Lambda_L} \big(|\vec{m}(\mathbf{x})|^2 - d(\mathbf{x})^2\big)
- L^n(\beta^{-1}2\ln{2} + \mu).
\end{align}
Then, by \cite[Theorem 1]{BP1996}, for any $0 < \beta < \infty$ and $\mu \in \R$, 
\begin{align}\label{OmegaG}
\Omega^{(\beta, \mu)} = \min_{\vec{m} \in B} \max_{d \in K} \mathcal{G}(d, \vec{m}),
\end{align}
and, if $d^{(\beta, \mu)}, \vec{m}^{(\beta, \mu)}$ is an extremizer of the right-hand side of (\ref{OmegaG}) in the sense that
$$\mathcal{G}(d^{(\beta, \mu)}, \vec{m}^{(\beta, \mu)}) = \min_{\vec{m} \in B} \max_{d \in K} \mathcal{G}(d, \vec{m})$$
and $h^{(\beta, \mu)}$ is the associated linear map (defined by (\ref{hdef}) with $d, \vec{m}$ replaced by $d^{(\beta, \mu)}, \vec{m}^{(\beta, \mu)}$), then: 
\begin{enumerate}[$(i)$]
\item $\gamma^{(\beta, \mu)} := \big(1 + e^{\beta h^{(\beta, \mu)}}\big)^{-1}$ is a HF Gibbs state, i.e., $\Omega(\gamma^{(\beta, \mu)}) = \Omega^{(\beta, \mu)}$.

\item For any $\mathbf{x} \in \Lambda_L$,  
\begin{subequations}\label{dvecm}
\begin{align}\label{dvecma}
& d^{(\beta, \mu)}(\mathbf{x}) = \tr(\gamma^{(\beta, \mu)} (1_{\mathbf{x}} \otimes I)) - 1, 
	\\\label{dvecmb}
& \vec{m}^{(\beta, \mu)}(\mathbf{x}) = \tr(\gamma^{(\beta, \mu)}(1_{\mathbf{x}} \otimes \vec{\sigma})),
\end{align}
\end{subequations}
where $1_{\mathbf{x}} = |\mathbf{x}\rangle\langle \mathbf{x}|$ is the projection onto $\mathbf{x}$ and $I$ is the identity map in spin space. 
\end{enumerate}
Thus, instead of minimizing $\Omega$, one can search for extremizers of $\mathcal{G}$. 

The relations (\ref{dvecm}) show that $d^{(\beta, \mu)}(\mathbf{x}) \in [-1,1]$ and $\vec{m}^{(\beta, \mu)}(\mathbf{x})$ are the expected electron density and the expected spin at $\mathbf{x}$ in the HF Gibbs state $\gamma^{(\beta, \mu)}$, where the normalization is such that $d^{(\beta, \mu)}(\mathbf{x}) = 0$ corresponds to one electron per site (i.e., to half-filling). Indeed, the expected electron density and the expected spin at $\mathbf{x}$ in a state $\rho:\mathcal{F} \to \mathcal{F}$ are given by $\langle n_{\mathbf{x}}\rangle_\rho$ and $\langle \vec{S}_{\mathbf{x}}\rangle_\rho$, respectively, where $n_{\mathbf{x}} := \sum_{\sigma} c_{\mathbf{x} \sigma}^\dagger c_{\mathbf{x} \sigma}$ and $\vec{S}_{\mathbf{x}} := \sum_{\lambda, \tau} c_{\mathbf{x}\lambda}^\dagger \vec{\sigma}_{\lambda \tau} c_{\mathbf{x}\tau}$ are the density and spin operators at $\mathbf{x}$. 
But $n_{\mathbf{x}}$ is the second quantization of $1_{\mathbf{x}} \otimes I$, and $\vec{S}_{\mathbf{x}}$ is the second quantization of $1_{\mathbf{x}} \otimes \vec{\sigma}$, so Lemma \ref{rhogammalemma} and (\ref{dvecm}) imply that 
\begin{align*}
& d^{(\beta, \mu)}(\mathbf{x}) = \langle n_{\mathbf{x}}\rangle_{\rho_{\gamma^{(\beta, \mu)} }} - 1, 
\qquad
 \vec{m}^{(\beta, \mu)}(\mathbf{x}) = \langle \vec{S}_{\mathbf{x}} \rangle_{\rho_{\gamma^{(\beta, \mu)} }}.
\end{align*}

\subsection{The case of a vanishing chemical potential}
If $\mu = 0$, the following results were obtained for any $0 < \beta \leq \infty$ in \cite{BLS1994} (see also \cite{BP1996}): 

\begin{enumerate}[$(i)$]

\item Any extremizer $(d^{(\beta, 0)}, \vec{m}^{(\beta, 0)})$ of $\mathcal{G}(d, \vec{m})$ satisfies $d^{(\beta, 0)}(\mathbf{x}) = 0$ for all $\mathbf{x}$. 

\item $\Omega^{(\beta, 0)} = \min_{\vec{m} \in B} \mathcal{G}(0, \vec{m})$, where
\begin{align}\label{calGdzerodef}
 \mathcal{G}(0, \vec{m}) &= 
\begin{cases}
-\frac{1}{\beta}\tr\big[\ln \cosh\big(\frac{\beta h}{2}\big)\big] + \frac{U}{4} \sum_{\mathbf{x} \in \Lambda_L} |\vec{m}(\mathbf{x})|^2 
- L^n \beta^{-1}2\ln{2}, & 0 < \beta < \infty,
	\\
-\frac{1}{2}\tr |h| + \frac{U}{4} \sum_{\mathbf{x} \in \Lambda_L} |\vec{m}(\mathbf{x})|^2, & \beta = \infty,
\end{cases}
	\\
 h_{\mathbf{x}\tau, \mathbf{y}\lambda} & = t_{\mathbf{x}\mathbf{y}}  \delta_{\tau \lambda} - \frac{U}{2}  \delta_{\mathbf{x}\mathbf{y}} \vec{m}(\mathbf{x}) \cdot \vec{\sigma}_{\tau \lambda}.
\end{align}

\item If $\vec{m}^{(\beta, 0)}$ is a minimizer of $\mathcal{G}(0, \vec{m})$, then there is a fixed unit vector $\vec{e} \in \R^3$ such that 
\begin{align}\label{vecmAF}
\vec{m}^{(\beta, 0)}(\mathbf{x}) = |\vec{m}^{(\beta, 0)}(\mathbf{x})| (-1)^{x_1 + \cdots + x_n} \vec{e}.
\end{align}

\end{enumerate}
The fact that $d^{(\beta, 0)}(\mathbf{x}) = 0$ for all $\mathbf{x}$ shows that the system is at half-filling when $\mu = 0$.

\subsection{Derivation of equations (\ref{meanfieldtprimezero}) and (\ref{meanfieldtprimezeroT0})}
In this paper, we are interested in the case of uniform antiferromagnetic order at half-filling. We therefore take $\mu = 0$ and assume that $|\vec{m}^{(\beta, 0)}(\mathbf{x})|$ is independent of $\mathbf{x}$. By (\ref{vecmAF}), this implies that $\vec{m}$ has the form
\begin{align}\label{vecmdef}
\vec{m}(\mathbf{x}) = m_{\AF} (-1)^{x_1 + \cdots + x_n}  \vec{e},
\end{align}
where $m_{\AF} \in [0,1]$ is a constant and $\vec{e} \in \R^3$ is a fixed unit vector. 
In this case, the expression for the functional $\mathcal{G}(d, \vec{m})$ in (\ref{calGdef}) can be simplified as follows. In terms of the basis $\{|\mathbf{k} \sigma\rangle\}$ defined in (\ref{kbasisdef}), the linear map $h$ is given by
$$h_{\mathbf{k}\tau, \mathbf{k}'\lambda} = 
\varepsilon(\mathbf{k}) \delta_{\mathbf{k}\mathbf{k}'}\delta_{\tau\lambda}
- \Delta_{\AF} \delta_{\mathbf{k},\mathbf{k}'+\mathbf{Q}} \vec{e} \cdot \vec{\sigma}_{\tau \lambda},
$$
where $\varepsilon(\mathbf{k})$ is the function in (\ref{epsilondef}), $\mathbf{Q} \in \R^n$ is the vector $\mathbf{Q} := (\pi, \dots, \pi)$, and $\mathbf{k}'+\mathbf{Q}$ is computed modulo $2\pi\Z^n$, so that $\mathbf{k}'+\mathbf{Q} \in \Lambda_L^*$.
We observe that $h$ is diagonal in the basis $\{|\mathbf{k} \sigma\rangle\}$ except that for each $\mathbf{k}$, the four elements $\mathbf{k}\uparrow$, $\mathbf{k} + \mathbf{Q} \uparrow$, $\mathbf{k}\downarrow$, and $\mathbf{k}+ \mathbf{Q} \downarrow$ are coupled via the $4 \times 4$ matrix
\begin{align*}
& \begin{pmatrix} 
 h_{\mathbf{k}\uparrow , \mathbf{k}\uparrow} & 
 h_{\mathbf{k}\uparrow , \mathbf{k}+\mathbf{Q}\uparrow} & 
 h_{\mathbf{k}\uparrow , \mathbf{k}\downarrow} & 
 h_{\mathbf{k}\uparrow , \mathbf{k}+ \mathbf{Q}\downarrow} 
 	\\
h_{\mathbf{k}+\mathbf{Q}\uparrow, \mathbf{k}\uparrow} & 
h_{\mathbf{k}+\mathbf{Q}\uparrow, \mathbf{k}+\mathbf{Q}\uparrow} & 
h_{\mathbf{k}+\mathbf{Q}\uparrow, \mathbf{k}\downarrow} & 
h_{\mathbf{k}+\mathbf{Q}\uparrow, \mathbf{k}+\mathbf{Q}\downarrow} & 
 	\\
h_{\mathbf{k}\downarrow, \mathbf{k}\uparrow} & 
h_{\mathbf{k}\downarrow, \mathbf{k}+\mathbf{Q}\uparrow} & 
h_{\mathbf{k}\downarrow, \mathbf{k}\downarrow} & 
h_{\mathbf{k}\downarrow, \mathbf{k}+\mathbf{Q}\downarrow} & 
 	\\
h_{\mathbf{k}+\mathbf{Q}\downarrow, \mathbf{k}\uparrow} & 
h_{\mathbf{k}+\mathbf{Q}\downarrow, \mathbf{k}+\mathbf{Q}\uparrow} & 
h_{\mathbf{k}+\mathbf{Q}\downarrow, \mathbf{k}\downarrow} & 
h_{\mathbf{k}+\mathbf{Q}\downarrow, \mathbf{k}+\mathbf{Q}\downarrow} 
 \end{pmatrix}
 =
 	\\
 & \begin{pmatrix} \varepsilon(\mathbf{k}) & 
   -\Delta_{\AF}(\vec{e} \cdot \vec{\sigma})_{\uparrow \uparrow}  & 
0  & 
   -\Delta_{\AF}(\vec{e} \cdot \vec{\sigma})_{\uparrow \downarrow}
   	\\
-\Delta_{\AF}(\vec{e} \cdot \vec{\sigma})_{\uparrow \uparrow} &
\varepsilon(\mathbf{k}+\mathbf{Q})  & 
-\Delta_{\AF}(\vec{e} \cdot \vec{\sigma})_{\uparrow \downarrow} & 
0
	\\
0 & 
-\Delta_{\AF}(\vec{e} \cdot \vec{\sigma})_{\downarrow \uparrow}  & 
 \varepsilon(\mathbf{k}) & 
 -\Delta_{\AF}(\vec{e} \cdot \vec{\sigma})_{\downarrow \downarrow} 
 	\\
-\Delta_{\AF}(\vec{e} \cdot \vec{\sigma})_{\downarrow \uparrow}  &
0 &
-\Delta_{\AF}(\vec{e} \cdot \vec{\sigma})_{\downarrow \downarrow} & 
\varepsilon(\mathbf{k}+\mathbf{Q}) 
 \end{pmatrix}.
\end{align*}
Using that $\varepsilon(\mathbf{k}+\mathbf{Q}) = -\varepsilon(\mathbf{k})$ and $|\vec{e}|^2 = 1$, we find that this $4 \times 4$ matrix has two double eigenvalues given by
$$E_\pm := \pm \sqrt{\Delta_{\AF}^2 + \varepsilon(\mathbf{k})^2}.$$
Hence
\begin{align*}
\tr\Big[\ln \cosh\Big(\frac{\beta h}{2}\Big)\Big] 
& = 2 \sum_{\mathbf{k} \in \Lambda_{L, \half}^*} \Big\{\ln \cosh\Big(\frac{\beta E_+}{2}\Big) + \ln \cosh\Big(\frac{\beta E_-}{2}\Big)\Big\}
	\\
& = 2 \sum_{\mathbf{k} \in \Lambda_L^*} \ln \cosh\Big(\frac{\beta E_+}{2}\Big),
\end{align*}
where 
$$\Lambda_{L, \half}^* := \bigg\{(k_1, \dots, k_n) \in \frac{2\pi}{L}\Z^n \,\bigg|\, -\pi \leq k_j < \pi\;\; \text{for $j =1,\dots, n-1$ and $-\pi \leq k_n < 0$}\bigg\}$$
includes exactly one of the points $\mathbf{k}$ and $\mathbf{k} + \mathbf{Q}$ for each $\mathbf{k} \in \Lambda_L^*$. Similarly, 
$$\tr |h| = \sum_{\mathbf{k} \in \Lambda_{L, \half}^*} 4E_+
= \sum_{\mathbf{k} \in \Lambda_{L}^*} 2E_+.$$
Consequently, we can write (\ref{calGdzerodef}) as (recall that $m_{\AF} = 2\Delta_{\AF}/U$)
\begin{align}\label{calGmAF}
\mathcal{G}(0, \vec{m}) = \begin{cases}
- 2 \sum_{\mathbf{k} \in \Lambda_L^*} \frac{1}{\beta} \ln\big( 2 \cosh\big(\frac{\beta E_+}{2}\big)\big)
+ \frac{L^n\Delta_{\AF}^2}{U}, & 0 < \beta < \infty,
	\\
-\sum_{\mathbf{k} \in \Lambda_{L}^*} E_+ + \frac{L^n \Delta_{\AF}^2}{U}, & \beta = \infty.
\end{cases}
\end{align}
A calculation gives 
\begin{align}\label{dcalGdm}
\frac{\partial \mathcal{G}(0, \vec{m})}{\partial \Delta_{\AF}}
=  \begin{cases}
L^n\big(- \frac{\Delta_{\AF}}{L^n} \sum_{\mathbf{k} \in \Lambda_{L}^*} \frac{\tanh(\frac{\sqrt{\Delta_{\AF}^2 + \varepsilon(\mathbf{k})^2}}{2T})}{\sqrt{\Delta_{\AF}^2 + \varepsilon(\mathbf{k})^2}} 
+ \frac{2\Delta_{\AF}}{U}\big), & 0 < \beta < \infty,
	\\
L^n\big(-\frac{\Delta_{\AF}}{L^n} \sum_{\mathbf{k} \in \Lambda_{L}^*} \frac{1}{\sqrt{\Delta_{\AF}^2 + \varepsilon(\mathbf{k})^2}} + \frac{2\Delta_{\AF}}{U}\big), & \beta = \infty.
\end{cases}
\end{align}
Equations (\ref{meanfieldtprimezero}) and (\ref{meanfieldtprimezeroT0}) follow by taking the thermodynamic limit $L\to \infty$ of the critical point equation $\partial \mathcal{G}(0, \vec{m})/\partial m_{\AF}= 0$ using the rule $\frac{1}{L^n} \sum_{\mathbf{k} \in \Lambda_{L}^*} \to \int_{[-\pi,\pi]^n}  \frac{d\mathbf{k}}{(2\pi)^n}$.

\subsection{Existence of solutions of the mean-field equation}\label{existencesubsec}
We first show that the N\'eel temperature $T_N$ is well-defined by (\ref{TNeelequation}).

\begin{lemma}\label{TNlemma}
For any $U > 0$ and $t_z \geq 0$, there is a unique $T_N = T_N(U, t_z) > 0$ satisfying (\ref{TNeelequation}). Moreover, $T_N(U,t_z) \downarrow 0$ uniformly for $t_z \geq 0$ as $U \downarrow 0$.
\end{lemma}
\begin{proof}
Fix $U > 0$. For each $t_z \geq 0$, 
\begin{align}\label{ftzTdef}
f_{t_z}(T) := \int_0^\infty N_{t_z}(\epsilon) 
\frac{\tanh(\frac{\epsilon}{2T})}{\epsilon} d\epsilon
\end{align}
is a continuous function of $T \in (0, +\infty)$ that tends to $+\infty$ as $T \to 0$ and to $0$ as $T \to +\infty$. Hence the equation $f_{t_z}(T_N) = 1/U$, which is equivalent to (\ref{TNeelequationNtz}) and hence also to (\ref{TNeelequation}), has at least one solution $T$ in $(0, +\infty)$. Since 
\begin{align}\label{ftzprimeTdef}
f_{t_z}'(T) = -\int_0^\infty N_{t_z}(\epsilon) 
\frac{1}{2T^2 \cosh^2(\frac{\epsilon}{2T})} d\epsilon < 0 \quad \text{for all $T > 0$,}
\end{align}
the solution is unique. Let $T_0 > 0$. Since $\tanh{x} \leq x$ for $x \in [0,1]$ and $\tanh{x} \leq 1$ for $x \geq 1$, we have
$$f_{t_z}(T_0) \leq \frac{1}{2T_0} \int_0^{2T_0} N_{t_z}(\epsilon) 
d\epsilon
+ \int_{2T_0}^\infty N_{t_z}(\epsilon) 
\frac{d\epsilon}{\epsilon} 
\leq \frac{1}{T_0} \int_0^{\infty} N_{t_z}(\epsilon) 
d\epsilon = \frac{1}{2T_0}$$
for all $t_z \geq 0$. Consequently, $f_{t_z}(T) \leq 1/(2T_0)$ for all $T \geq T_0$ and all $t_z \geq 0$.
It follows that if $U < 2T_0$, then $T_N(U, t_z) < T_0$ for all $t_z \geq 0$. Since $T_0 > 0$ was arbitrary, we conclude that $T_N(U,t_z) \downarrow 0$ uniformly for $t_z \geq 0$ as $U \downarrow 0$.
\end{proof}

The next lemma establishes the existence of a unique strictly positive solution $\Delta_{AF}$ of (\ref{meanfieldtprimezero})--(\ref{meanfieldtprimezeroT0}) for temperatures below $T_N$.
Note that, by \eqref{Deltadef}, $0\leq \Delta_{\AF}\leq U/2$ corresponds to $0\leq m_{\AF}\leq 1$; the latter is to be expected by the physics interpretation of $m_{\AF}$ as the magnitude of the spin expectation value.

\begin{lemma}\label{mAFlemma}
For any $U > 0$ and $t_z \geq 0$, equations (\ref{meanfieldtprimezero})--(\ref{meanfieldtprimezeroT0}) have a unique strictly positive solution $\Delta_{\AF} = \Delta_{\AF}(U,t_z,T) > 0$ whenever $T \in [0, T_N(U,t_z))$,  whereas if $T \geq T_N(U,t_z)$ the only solution is $\Delta_{\AF} = \Delta_{\AF}(U,t_z,T) = 0$. Moreover, $\Delta_{\AF}(U,t_z,T) \in [0,U/2)$ for all $U > 0$, $t_z \geq 0$, and $T \geq 0$.
\end{lemma}
\begin{proof}
Fix $U > 0$ and $t_z \geq 0$. Recalling the reformulation (\ref{mAFeq}) of equations (\ref{meanfieldtprimezero})--(\ref{meanfieldtprimezeroT0}), we see that it is enough to show that the equation $F_T(\Delta) = 1/U$ has a unique solution $\Delta > 0$ whenever $T \in [0, T_N(U,t_z))$ and that no such solution exists if $T \geq T_N(U,t_z)$, where
\begin{align}\label{Fmdef}
F_T(\Delta) := \int_0^\infty N_{t_z}(\epsilon)  
\frac{\tanh(\frac{\sqrt{\Delta^2 + \epsilon^2}}{2T})}{\sqrt{\Delta^2 + \epsilon^2}}  d\epsilon.
\end{align}
The function $F_T$ is continuous $[0,+\infty) \to [0,+\infty)$ and obeys $F_T(\Delta) \to 0$ as $\Delta \to +\infty$. Moreover, $F_T(0)$ equals the right-hand side of (\ref{ftzTdef}), which is a decreasing function of $T \in [0,+\infty)$ by (\ref{ftzprimeTdef}). 
The N\'eel temperature $T_N = T_N(U,t_z)$ is defined by the equation $F_{T_N}(0) = 1/U$.
We infer that if $T < T_N$, then $F_T(0) > 1/U$, and so the equation $F_T(\Delta) = 1/U$ has at least one solution $\Delta$ in $(0, +\infty)$. Since
$$F_T'(\Delta) =  \int_0^\infty N_{t_z}(\epsilon)  
\frac{\Delta}{(\Delta^2 + \epsilon^2)^{3/2}} \frac{x - \sinh(x)\cosh(x)}{\cosh^2(x)}\bigg|_{x = \frac{\sqrt{\Delta^2 + \epsilon^2}}{2T}} d\epsilon$$
and $x - \sinh(x)\cosh(x) < 0$ for all $x > 0$, we have $F_T'(\Delta) < 0$ for $\Delta\in (0, +\infty)$ so the solution is unique.  It only remains to show that $\Delta_{\AF}(U,t_z,T) < U/2$. This follows because if $\Delta \geq U/2$, then
$$F_T(\Delta) \leq \int_0^\infty N_{t_z}(\epsilon)  
\frac{1}{\sqrt{\Delta^2 + \epsilon^2}}  d\epsilon
< \frac{2}{U} \int_0^\infty N_{t_z}(\epsilon)  d\epsilon
= \frac{1}{U},$$
so no $\Delta \geq U/2$ can fulfill the equation $F_T(\Delta) = 1/U$.
\end{proof}

Lemma \ref{mAFlemma} implies that $\Delta_{\AF} = 0$ is the only solution of (\ref{meanfieldtprimezero}) if $T \geq T_N$. On the other hand, for any $T \in [0, T_N)$, there are two distinct solutions in $[0,U/2]$, namely $0$ and $\Delta_{\AF}(U,t_z,T) \in (0, U/2)$. The next lemma shows that the state corresponding to $\Delta_{\AF}(U,t_z,T)$ is favored over the state corresponding to $0$, and it therefore justifies discarding the zero solution when writing (\ref{mAFeq}).
The statement of the lemma involves the thermodynamic limit
\begin{align*}
G_T(\Delta) :=&\; \lim_{L\to \infty} L^{-3}\mathcal{G}(0, (2\Delta/U)(-1)^{x_1 + x_2 + x_3}  \vec{e})
	\\
= &\; \begin{cases}
- 2 T \int_{[-\pi,\pi]^3}  \ln\big( 2 \cosh\big(\frac{\sqrt{\Delta^2 + \varepsilon(\mathbf{k})^2}}{2 T}\big)\big) \frac{d\mathbf{k}}{(2\pi)^3}
+ \frac{\Delta^2}{U}, & 0 < T < \infty,
	\\
-\int_{[-\pi,\pi]^3} \sqrt{\Delta^2 + \varepsilon(\mathbf{k})^2} \frac{d\mathbf{k}}{(2\pi)^3} + \frac{\Delta^2}{U}, & T = 0.
\end{cases}
\end{align*}
 of $\mathcal{G}$ restricted to antiferromagnetic vector fields of the form (\ref{vecmdef}) in dimension $n = 3$.

\begin{lemma}\label{Gminlemma}
Let $U > 0$ and $t_z \geq 0$. If $T \in [0, T_N(U,t_z))$, then $\Delta = \Delta_{\AF}(U,t_z,T) > 0$ is the unique point at which $G_T$ attains its minimum on $[0,U/2]$.
If $T \geq T_N(U,t_z)$, then $\Delta = 0$ is the unique point at which $G_T$ attains its minimum on $[0,U/2]$.
\end{lemma}
\begin{proof}
Fix $U > 0$ and $t_z \geq 0$. We have, for any $T \geq 0$,
$$G_T'(\Delta) = 2\Delta \bigg(\frac{1}{U} - F_T(\Delta)\bigg),
$$
where $F_T(\Delta)$ is the function in (\ref{Fmdef}). 
Suppose first that $T \in [0, T_N(U,t_z))$. In this case, the proof of Lemma \ref{mAFlemma} shows that $F_T(\Delta) > 1/U$ for $\Delta \in [0, \Delta_{\AF}(U,t_z,T))$ and $F_T(\Delta) < 1/U$ for $\Delta \in (\Delta_{\AF}(U,t_z,T), U/2]$. Hence $\Delta_{\AF}(U,t_z,T) \in (0,U/2)$ is the unique critical point of $G_T$ on $(0,U/2)$ and $G_T$ is decreasing (resp. increasing) to the left (resp. right) of $\Delta_{\AF}(U,t_z,T)$, so $\Delta_{\AF}(U,t_z,T)$ is indeed the minimizer of $G_T$. 

If $T \geq T_N(U,t_z)$, then we instead have $G_T'(0) = 0$ and $G_T'(\Delta) > 0$ for $\Delta \in (0,U/2)$, so $0$ is the unique minimizer of $G_T$.
\end{proof}

\section{Asymptotics for the density of states}\label{densityofstatessec}
In this section, we derive asymptotic properties of the density of states function $N_{t_z}(\epsilon)$ defined in (\ref{3Ddensityofstates}). In particular, we prove Theorems \ref{2Ddensityth}, \ref{3Ddensityth}, and \ref{3Dtzdensityth}.

Making the change of variables $u_j = 2\cos k_j$ for $j = 1,2,3$, 
we can write (\ref{3Ddensityofstates}) as
\begin{align*}
N_{t_z}(\epsilon)  
& = 8\int_0^\pi \int_0^\pi \int_0^\pi \delta\big(2 (\cos{k_1} + \cos{k_2}) + 2 t_z \cos{k_3} + \epsilon\big) \frac{dk_1 dk_2 dk_3}{(2\pi)^3}
	\\
& = \int_{-2}^2 \int_{-2}^2 \int_{-2}^2 \frac{\delta\big(u_1 + u_2 + t_z u_3 + \epsilon\big)}{\sqrt{4 - u_1^2} \sqrt{4 - u_2^2} \sqrt{4 - u_3^2}} \frac{du_1 du_2 du_3}{\pi^3}.
\end{align*}
For $t_z = 0$, this reduces to
\begin{align}\label{N0epsilon}
N_{0}(\epsilon)  
& = \int_{-2}^2 \int_{-2}^2 \frac{\delta\big(u_1 + u_2 + \epsilon\big)}{\sqrt{4 - u_1^2} \sqrt{4 - u_2^2} } \frac{du_1 du_2}{\pi^2}
\end{align}
because $\int_{-2}^2 \frac{du_3}{\sqrt{4 - u_3^2}} = \pi$. In particular, we can express $N_{t_z}(\epsilon)$ in terms of $N_0(\epsilon)$:
\begin{align}\label{N3DN2D}
N_{t_z}(\epsilon)  
& = \int_{-2}^2 \frac{N_0\big(t_z u + \epsilon\big)}{\sqrt{4 - u^2}} \frac{du}{\pi}.
\end{align}
From (\ref{N0epsilon}), we infer that $N_0(\epsilon)$ is an even function of $\epsilon \in \R \setminus \{0, \pm 4\}$ that vanishes identically for $\epsilon \in (-\infty, -4) \cup (4, +\infty)$. Performing the integral with respect to $u_2$ in (\ref{N0epsilon}), we obtain
\begin{align}\label{N0epsilon2}
N_0(\epsilon)
   = \int_{-2}^{2- \epsilon} \frac{1}{\sqrt{4 - u^2}} \frac{1}{\sqrt{4 - (u + \epsilon)^2}} \frac{du}{\pi^2} \qquad \text{for $\epsilon \in (0,4)$}.
\end{align}   
We see that $N_0(\epsilon)$ is continuous for $\epsilon \in \R \setminus \{0, \pm 4\}$ with a singularity at $\epsilon = 0$. 

\subsection{Proof of Theorem \ref{2Ddensityth}}\label{2Ddensitysubsec}
Our goal is to find the behavior of the integral in (\ref{N0epsilon2}) as $\epsilon \downarrow 0$. 
The difficulty lies in the fact that the singularities of the integrand merge with the endpoints of integration as $\epsilon \downarrow 0$. This means that a naive approach based on substituting the expansion
$$ \frac{1}{\sqrt{4 - u^2}} \frac{1}{\sqrt{4 - (u + \epsilon)^2}} = \frac{1}{4-u^2} + \frac{u}{(4 - u^2)^2}\epsilon + \frac{\left(u^2+2\right)}{\left(4-u^2\right)^3}\epsilon^2 + \cdots$$
into (\ref{N0epsilon2}) is unfruitful, because the resulting terms are not integrable at $u = \pm 2$.
We therefore instead proceed as follows.

Let $F(w_1, w_2) = F(a,b,c,d; w_1, w_2)$ be defined for $w_1,w_2 \in \C \setminus [0, \infty)$ by the Pochhammer integral
\begin{align}\label{Fdef}
F(w_1, w_2) = \int_A^{(0+,1+,0-,1-)} v^a (v-w_1)^b (v-w_2)^{c} (1-v)^{d} dv, 
\end{align}
where $A \in (0,1)$ is a base point, $a,b,c,d$ are real exponents, and $w_1, w_2$ are assumed to lie outside the contour. More precisely, the Pochhammer contour in (\ref{Fdef}) begins at the base point $A$, encircles $0$ once in the positive (counterclockwise) direction, returns to $A$, encircles $1$ once in the positive direction, returns to $A$, encircles $0$ once in the negative direction, returns to $A$, and finally encircles $1$ once in the negative direction before returning to $A$.
In order to make $F$ single-valued, we restrict the domain of definition in (\ref{Fdef}) to $w_1,w_2 \in \C \setminus [0, \infty)$. Note that $F$ can be analytically continued to a multiple-valued analytic function of $w_1,w_2 \in \C \setminus \{0, 1\}$. Asymptotic expansions of $F(w_1, w_2)$ to all orders as one or both of the points $w_1$ and $w_2$ approach $0$ or $1$ were obtained in \cite{LV2019}.

To prove Theorem \ref{2Ddensityth}, we will use the following result from \cite{LV2019}. 
For $c \in \C$, we define $\rho_c(w)$ and $H(a,d)$ by
\begin{align}\label{rhodef}
\rho_c(w) = \begin{cases} e^{-i\pi c}, & \im w \geq 0, \\
e^{i\pi c}, & \im w < 0,
\end{cases}
\end{align}
and
\begin{align}\label{Hdef}
  H(a,d) = \int_A^{(0+,1+,0-,1-)} v^a (1-v)^{d} dv.
\end{align}

\begin{theorem}\cite[Theorem 2.4]{LV2019} \label{mainth4}
Let $a,b,c,d \in \R \setminus \Z$ be such that $a+b,c+d \notin\Z$. 
Then $F$ satisfies the following asymptotic expansion to all orders as $w_1 \to 0$ and $w_2 \to 1$ with $w_1,w_2 \in \C \setminus [0, \infty)$:
\begin{align}\nonumber
F(w_1, w_2) \sim &\; \sum_{k=0}^\infty\sum_{l=0}^\infty \frac{ \rho_c(w_2)}{k!\, l!}\bigg\{ D_{kl}^{(1)} w_1^k (w_2- 1)^l 
+ D_{kl}^{(2)} w_1^k (w_2-1)^{c+d+1+ l}
	\\ \label{Fasymptotics}
& + \rho_{a+b}(w_1) D_{kl}^{(3)} w_1^{a+b+1+ k}(w_2-1)^l\bigg\}
\end{align}
where the coefficients $D_{kl}^{(j)} \equiv D_{kl}^{(j)}(a,b,c,d)$, $j = 1,2,3$, are given by
\begin{align*}
& D_{kl}^{(1)} = \frac{(e^{2\pi i a} -1)(e^{2\pi i d} -1)}{(e^{2\pi i(a+b)} - 1)(e^{2\pi i(c+d)} - 1)} \frac{\Gamma(b+1)\Gamma(c+1)(-1)^k}{\Gamma(b+1-k)\Gamma(c+1-l)} H(a+b-k, c+d-l),
	\\
& D_{kl}^{(2)} = \frac{(e^{2\pi i a} -1) e^{\pi i d}}{1 - e^{2\pi i (c+d)}} \frac{\Gamma(b+1)\Gamma(a+b+1-k)(-1)^k}{\Gamma(b+1-k)\Gamma(a+b+1-k-l)}H(c, d+l),
	\\
& D_{kl}^{(3)} = \frac{(e^{2\pi i d} -1) e^{\pi i a}}{e^{2\pi i (a+b)} - 1} \frac{\Gamma(c+1)\Gamma(-c-d+k-l)}{\Gamma(c+1-l)\Gamma(-c-d+l)}H(a+k, b),
\end{align*} 
and principal branches are used for all powers.
\end{theorem}

In order to use Theorem \ref{mainth4}, we let $a = b = c = d = -1/2$, assume that $\epsilon \in (0,4)$, and write (\ref{N0epsilon2}) as
\begin{align*}
N_0(\epsilon)
   = \frac{1}{\pi^2} \int_{-2}^{2- \epsilon} (u-(-2))^a (u-(-2-\epsilon))^b (2-u)^c (2-\epsilon - u)^d  du.
\end{align*}   
The linear fractional transformation
$$u \mapsto v = \frac{u+2}{4-\epsilon}$$
maps $(-2,2-\epsilon,\infty)$ to $(0,1,\infty)$, and changing variables from $u$ to $v$, we obtain
\begin{align}\label{N0w1w2}
N_0(\epsilon)
   = \frac{1}{\pi^2} (4-\epsilon)^{a+b+c+d+1} \int_{0}^{1}  v^a (v-w_1)^b (w_2 - v)^c   (1-v)^d dv \qquad \text{for $\epsilon \in (0,4)$},
\end{align}   
where principal branches are used and
\begin{align}\label{w1w2epsilon}
w_1 := -\frac{\epsilon}{4 - \epsilon}, \qquad w_2 := \frac{4}{4-\epsilon} + i0.
\end{align}
We have added an infinitesimal positive imaginary part to $w_2$ to ensure that $w_1,w_2 \in \C \setminus [0, \infty)$. 
As $\epsilon \downarrow 0$, we have $w_1 \uparrow 0$ and $w_2 \downarrow 1$. 
Since $w_2$ has an infinitesimal positive imaginary part and $v \in [0,1]$, we have
$$(w_2 -v)^c = e^{c \ln|w_2 -v| + ci\arg(w_2 - v)} = e^{c \ln|w_2 -v| + ci\arg(v - w_2) + c\pi i}
= e^{c\pi i}(v- w_2)^c.$$
We therefore find
\begin{align}\label{N0epsilonintabcd}
N_0(\epsilon)
   = -\frac{i}{\pi^2 (4-\epsilon)} \int_{0}^{1}  v^a (v-w_1)^b (v - w_2)^c  (1-v)^d dv.
\end{align}   

On the other hand, since $a =-1/2 > -1$ and $d =-1/2> -1$, we can collapse the contour in (\ref{Fdef}) onto $[0,1]$, which gives
\begin{align*}
F(w_1, w_2) & = (-1 + e^{2\pi i a} - e^{2\pi i (a+d)} + e^{2\pi i d}) \int_0^1 v^a (v-w_1)^b (v-w_2)^{c} (1-v)^{d} dv
	\\
& = -4 \int_0^1 v^a (v-w_1)^b (v-w_2)^{c} (1-v)^{d} dv.
\end{align*}
Comparing with (\ref{N0epsilonintabcd}), we conclude that
\begin{align}\label{N0F}
N_0(\epsilon)
   = \frac{i}{4\pi^2 (4-\epsilon)} F(w_1, w_2).
\end{align}  

Even though (\ref{N0F}) relates $N_0$ to $F$, we cannot immediately apply Theorem \ref{mainth4} because the assumption $a+b,c+d \notin\Z$ is not fulfilled. However, we can set $a = -\frac{1}{2} + \delta$, $b = -\frac{1}{2}$, $c = -\frac{1}{2} + \delta$, $d = -\frac{1}{2}$ in Theorem \ref{mainth4}, and then let $\delta \downarrow 0$.\footnote{It is not clear from the statement of Theorem \ref{mainth4} that the asymptotic formula for $F$ remains valid in the limit $\delta \to 0$. However, a look at the proof in \cite{LV2019} shows that this limit may be taken term-wise in (\ref{Fasymptotics}).} 
In this limit, the poles at $\delta = 0$ introduced by the factors $e^{2\pi i(a+b)} - 1$ and $e^{2\pi i(c+d)} - 1$ in the coefficients $D_{kl}^{(j)}$ are cancelled by zeros originating from the poles of the Gamma functions. Using the identity
$$H(a,d) = (-1 + e^{2\pi i a} - e^{2\pi i (a+d)} + e^{2\pi i d}) \frac{\Gamma (a+1) \Gamma (d+1)}{\Gamma (a+d+2)}$$
and taking $\delta \to 0$ yields the following asymptotic expansion of $F(w_1, w_2)$ to all orders as $w_1 \uparrow 0$ and $w_2 \downarrow 1$:
\begin{align}\label{FAkl}
F(w_1, w_2) \sim 4i \sum_{k,l=0}^\infty A_{k,l}(w_1, w_2),
\end{align}
where the functions $A_{k,l}$ are given by (\ref{Akldef}).
Using this formula in (\ref{N0F}) with $w_1$ and $w_2$ given by (\ref{w1w2epsilon}), we obtain the all-order expansion of $N_0(\epsilon)$ stated in (\ref{N0expansion}), and a direct calculation then gives (\ref{N0expansionfirstfew}). 

It follows from the proof of Theorem \ref{mainth4} presented in \cite{LV2019} that the asymptotic formula (\ref{Fasymptotics}) can be differentiated term-by-term with respect to $w_1$ and $w_2$ any finite number of times. Repeating the above arguments, we see that the same is true of (\ref{FAkl}). It follows that the formulas (\ref{N0expansion}) and (\ref{N0expansionfirstfew}) can be differentiated term-by-term with respect to $\epsilon$ any finite number of times. 
This completes the proof of Theorem \ref{2Ddensityth}.

\subsection{Proof of Theorem \ref{3Ddensityth}}\label{3Ddensitysubsec}
Let $t_z \in (0,2)$. By Theorem \ref{2Ddensityth}, $N_0$ has a logarithmic singularity at $\epsilon = 0$. Hence, it is clear from (\ref{N3DN2D}) that $N_{t_z}(\epsilon)$ has a continuous extension to $\epsilon = 0$ such that
\begin{align}\label{Ntz0}
N_{t_z}(0) = \int_{-2}^2 \frac{N_0(t_z u)}{\sqrt{4 - u^2}} \frac{du}{\pi}.
\end{align}
However, for $j \geq 1$, the $j$th derivative $N_0^{(j)}(t_z u)$ has a more severe singularity at $u = 0$, so it is not possible to differentiate inside the integral in (\ref{N3DN2D}) as it stands. In what follows, we therefore show that the contour of integration in (\ref{N3DN2D}) can be deformed away from the singularity at $\epsilon = 0$. 

The function $F$ in (\ref{Fdef}) can be analytically continued to a multiple-valued analytic function of $w_1,w_2 \in \C \setminus \{0, 1\}$. Hence, by (\ref{N0F}), $N_0(\epsilon)$ can be analytically continued from $0 < \epsilon < 4$ to a multiple-valued analytic function of $\epsilon \in \C \setminus \{0, 4\}$ (i.e., $N_0$ has an analytic continuation to the universal cover of $\C \setminus \{0, 4\}$). We let $\tilde{N}_0(\epsilon)$ be the (single-valued) analytic continuation of $N_0(\epsilon)$ from $\epsilon \in (0,4)$ to $\epsilon \in \C \setminus \big((-\infty, 0] \cup [4, +\infty)\big)$. 
For $\epsilon \in \R$, let $\tilde{N}_{0+}(\epsilon) = \lim_{h \downarrow 0}\tilde{N}_{0}(\epsilon + i h)$ be the boundary values of $\tilde{N}_{0}(\epsilon)$ from the upper half-plane. 
For $\epsilon \in (0,4)$, we have $\tilde{N}_{0+}(\epsilon) = N_0(\epsilon)$. 
For $\epsilon \in (-4,0)$, $\tilde{N}_{0+}(\epsilon)$ does not coincide with $N_0(\epsilon)$, but we have the following lemma.

\begin{lemma}\label{N0N0tildelemma}
 For $\epsilon \in (-4,0)$, it holds that $N_0(\epsilon) = \re \tilde{N}_{0+}(\epsilon)$.
\end{lemma}
\begin{proof}
Fix $r \in (0,4)$ and let $\epsilon(\alpha) = re^{i \alpha}$. 
According to (\ref{N0w1w2}), we have
\begin{align*}
\tilde{N}_0(\epsilon(0))
   = \frac{1}{\pi^2} (4-\epsilon(0))^{a+b+c+d+1} \int_{0}^{1}  v^a (v-w_1(0))^b (w_2(0) - v)^c   (1-v)^d dv,
\end{align*}   
where $a=b=c=d=-1/2$, $w_1(\alpha) := -\frac{\epsilon(\alpha)}{4 - \epsilon(\alpha)}$, and $w_2(\alpha) := \frac{4}{4-\epsilon(\alpha)} + i0$.
As $\alpha$ increases from $0$ to $\pi$, $\epsilon(\alpha)$ traces out a counterclockwise semicircle in the upper half-plane from $r$ to $-r$, $w_1(\alpha)$ moves below $0$ from $w_1(0) = - \frac{r}{4-r} < 0$ to $w_1(\pi) = \frac{r}{4 + r} > 0$, and $w_2(\alpha)$ moves above $1$ from $w_2(0) = \frac{4}{4-r}  + i0$ to $w_2(\pi) = \frac{4}{4 + r} < 1$. Hence, 
\begin{align*}
\tilde{N}_{0+}(-r) =\lim_{\alpha \uparrow \pi} \tilde{N}_0(\epsilon(\alpha))
= &\; \frac{(4 - \epsilon(\pi))^{a+b+c+d+1} }{\pi^2} 
	\\
&\times
\bigg(e^{\pi i b} \int_{0}^{w_1(\pi)} v^a (w_1(\pi) - v)^b (w_2(\pi) - v)^c   (1-v)^d dv
	\\
&\qquad +\int_{w_1(\pi)}^{w_2(\pi)}  v^a (v-w_1(\pi))^b (w_2(\pi) - v)^c   (1-v)^d dv
	\\
&\qquad + e^{\pi i c} \int_{w_2(\pi)}^{1}  v^a (v-w_1(\pi))^b  (v-w_2(\pi))^c   (1-v)^d dv
\bigg).
 \end{align*} 
Since $e^{\pi i b} = e^{\pi i c} = -i$, only the middle integral contributes to the real part of $\tilde{N}_{0+}(-r)$, i.e.,
\begin{align*}
\re \tilde{N}_{0+}(-r) 
= \frac{(4 - \epsilon(\pi))^{a+b+c+d+1}}{\pi^2}  
\int_{w_1(\pi)}^{w_2(\pi)}  v^a (v-w_1(\pi))^b (w_2(\pi) - v)^c   (1-v)^d dv.
\end{align*} 
Performing the change of variables $u = \frac{v-w_1(\pi)}{w_2(\pi) - w_1(\pi)}$, we obtain
\begin{align*}
\re \tilde{N}_{0+}(-r) 
= &\; \frac{(4 - \epsilon(\pi))^{a+b+c+d+1} }{\pi^2}  (w_2(\pi) - w_1(\pi))^{a+b+c+d+1}
	\\
&\times \int_0^1 \Big(u + \frac{w_1(\pi)}{w_2(\pi) - w_1(\pi)}\Big)^a u^b  (1-u)^c  \Big(\frac{1-w_1(\pi)}{w_2(\pi) - w_1(\pi)} - u\Big)^d du.
\end{align*} 
Using that $\epsilon(\pi) = -r$, $w_1(\pi) = \frac{r}{4 + r}$, and $w_2(\pi) = \frac{4}{4 + r}$, we conclude that \begin{align*}
\re \tilde{N}_{0+}(-r) 
= &\; \frac{(4 - r)^{a+b+c+d+1} }{\pi^2} 
 \int_0^1 \Big(u - \frac{-r}{4-r}\Big)^a u^b  (1-u)^c  \Big(\frac{4}{4-r} - u\Big)^d du.
\end{align*} 
In view of (\ref{N0w1w2}) and the fact that $a=b=c=d=-1/2$, the right-hand side equals $N_0(r)$. Since $r \in (0, 4)$ was arbitrary, this completes the proof. 
\end{proof}

%

By virtue of Lemma \ref{N0N0tildelemma} and (\ref{N3DN2D}), we have, for $\epsilon \in (-4 + 2t_z, 4 - 2t_z)$,
\begin{align}
N_{t_z}(\epsilon)  
 = \int_{-2}^2 \frac{\re \tilde{N}_{0+}\big(t_z u + \epsilon\big)}{\sqrt{4 - u^2}} \frac{du}{\pi}
 = \re \int_{-2}^2 \frac{\tilde{N}_{0+}\big(t_z u + \epsilon\big)}{\sqrt{4 - u^2}} \frac{du}{\pi}.
\end{align}
The contour in the right-most integral can be deformed into the upper half-plane. This gives
\begin{align}\label{Ntzepsilonre}
N_{t_z}(\epsilon)  
 = \re \int_{\gamma_{-2,2}} \frac{\tilde{N}_{0}\big(t_z u + \epsilon\big)}{\sqrt{4 - u^2}} \frac{du}{\pi}.
\end{align}
where $\gamma_{-2,2}$ is any contour starting at $-2$ and ending at $2$ which (apart from the endpoints) lies in $\{\im u > 0\}$. For definiteness, we let $\gamma_{-2,2}$ be the clockwise semicircle of radius $2$ starting at $-2$ and ending at $2$. 

For $t_z \in (0,2)$, $t_zu + \epsilon$ stays away from $\{0,4\}$ for all $u\in \gamma_{-2,2}$ provided that $\epsilon \in \C$ satisfies $|\epsilon| < \min(2t_z, 4 - 2t_z)$. Hence, the map $\epsilon \mapsto \int_{\gamma_{-2,2}} \frac{\tilde{N}_{0}(t_z u + \epsilon)}{\sqrt{4 - u^2}} \frac{du}{\pi}$ is analytic in the open disk $\{\epsilon \in \C \,|\, |\epsilon| < \min(2t_z, 4 - 2t_z)\}$. Theorem \ref{3Ddensityth} now follows from (\ref{Ntz0}) and (\ref{Ntzepsilonre}).

\subsection{Proof of Theorem \ref{3Dtzdensityth}}\label{3Dtzdensitysubsec}
Setting $\epsilon = 0$ in (\ref{N3DN2D}) and then using that $N_0$ is an even function, we obtain
\begin{align}\label{Ntzat0}
N_{t_z}(0) = \int_{-2}^2 \frac{N_0\big(t_z u\big)}{\sqrt{4 - u^2}} \frac{du}{\pi}
= 2\int_{0}^2 \frac{N_0\big(t_z u\big)}{\sqrt{4 - u^2}} \frac{du}{\pi}.
\end{align}
Substituting in the asymptotic formula (\ref{N0expansionfirstfew}) for $N_0$, we obtain, as $t_z \downarrow 0$,
\begin{align*}
N_{t_z}(0)
= &\; 2\int_{0}^2 \frac{1}{\sqrt{4 - u^2}} \bigg( \frac{\ln(\frac{16}{\epsilon})}{2 \pi^2} + \frac{\epsilon^2 (\ln(\frac{16}{\epsilon})-1)}{128 \pi^2}
+ \frac{3 \epsilon^4 (6 \ln(\frac{16}{\epsilon})-7)}{2^{16} \pi^2}
	\\ \nonumber
& +\frac{5}{3}\frac{\epsilon^6 (30 \ln(\frac{16}{\epsilon})-37)}{2^{22} \pi^2}
+ \frac{35}{3} \frac{\epsilon^8 (420 \ln(\frac{16}{\epsilon})-533)}{2^{33} \pi^2}
	\\ 
&+ \frac{63}{5} \frac{\epsilon^{10} (1260\ln(\frac{16}{\epsilon})-1627)}{2^{39} \pi^2} 
\bigg)\bigg|_{\epsilon = t_z u} \frac{du}{\pi}
+ E(t_z),
\end{align*}
where the error term $E(t_z)$ satisfies
$$|E(t_z)| \leq 2\int_{0}^2 \frac{C(ut_z)^{12}}{\sqrt{4 - u^2}}\bigg(1 + \ln\Big(\frac{1}{ut_z}\Big)\bigg) \frac{du}{\pi}
\leq C t_z^{12} \ln\frac{1}{t_z}$$
for some constant $C> 0$ independent of $0 < t_z \leq 1$.
Using the identities
\begin{subequations}\label{integralidentities}
\begin{align}
& 2\int_{0}^2 \frac{u^j}{\sqrt{4 - u^2}} \frac{du}{\pi} = \frac{2^j\Gamma(\frac{j+1}{2})}{\sqrt{\pi} \Gamma(\frac{j}{2} +1)}, \qquad j = 0,1, 2, \dots,
	\\
&  2\int_{0}^2 \frac{u^j \ln{u}}{\sqrt{4 - u^2}} \frac{du}{\pi} = \frac{2^{j-1} \Gamma (\frac{j+1}{2}) \big(2\ln(2) -\psi(\frac{j}{2}+1)+\psi(\frac{j+1}{2})\big)}{\sqrt{\pi }  \Gamma (\frac{j}{2}+1)}, \quad j = 0,1, 2, \dots,
\end{align}
\end{subequations}
long but straightforward calculations yield the expansion in (\ref{Ntzat0expansion}). 

The claim about termwise differentiation of (\ref{Ntzat0expansion}) follows from the above argument and the analogous claim in Theorem \ref{2Ddensityth}.
This completes the proof of Theorem \ref{3Dtzdensityth}. 

\begin{remark}
  By substituting the expansion (\ref{N0expansion}) for $N_0$ into (\ref{Ntzat0}) and using (\ref{integralidentities}), one can obtain an expansion to all orders of $N_{t_z}(0)$ as $t_z \downarrow 0$. 
\end{remark}


\section{Asymptotics for the N\'eel temperature}\label{Neelsec}
In this section, we study the asymptotic behavior of the N\'eel temperature $T_N(U, t_z)$ as $U \downarrow 0$. Our goal is to prove Theorems \ref{2DNeelth} and \ref{3DNeelth}.

\subsection{Proof of Theorem \ref{2DNeelth}}\label{2DNeelsubsec}
Suppose $t_z=0$. 
From (\ref{N0epsilon}), we see that $N_0(\epsilon) = 0$ for $\epsilon \geq 4$. 
Thus, equation (\ref{TNeelequationNtz}) for $T_N = T_N(U,0)$ can be written as 
\begin{align}\label{TNeelequationN0}
\frac{1}{U} = \int_0^4 N_{0}(\epsilon) 
\frac{\tanh(\frac{\epsilon}{2T_N})}{\epsilon} d\epsilon.
\end{align}
Adding and subtracting terms, we write (\ref{TNeelequationN0}) as
\begin{align}\label{TceqI1234}
\frac{1}{U} - a_0 = &\; I_1(T_N) + \frac{\ln(\frac{8}{T_N})}{2\pi^2} I_2(T_N)  + \frac{1}{2\pi^2}I_3(T_N) ,
\end{align}
where the constant $a_0$ is given by
\begin{align*}
& a_0 = \int_{0}^{4} \bigg(N_0(\epsilon) - \frac{\ln(\frac{16}{\epsilon})}{2 \pi^2}\bigg)
\frac{1}{\epsilon}  d\epsilon \approx 0.007013,
\end{align*}
and the functions $I_j(T)$, $j =1, 2, 3$, are defined by
\begin{align}
& I_1(T) = \int_{0}^{4} \bigg(N_0(\epsilon) - \frac{\ln(\frac{16}{\epsilon})}{2 \pi^2} \bigg)\frac{\tanh(\frac{\epsilon}{2T}) - 1}{\epsilon}  d\epsilon,
	\\ \label{I2def}
& I_2(T) = \int_{0}^{4} \frac{\tanh(\frac{\epsilon}{2T})}{\epsilon}  d\epsilon,
	\\
& I_3(T) =  \int_{0}^{4} \ln\bigg(\frac{2T}{\epsilon}\bigg) \frac{\tanh(\frac{\epsilon}{2T})}{\epsilon}  d\epsilon.
\end{align}

The next three lemmas describe the behavior of $I_1(T)$, $I_2(T)$, and $I_3(T)$ for small $T$.

\begin{lemma}\label{I1lemma}
There exists a constant $C > 0$ such that
$$|I_1(T)| \leq  C T^2 \ln\Big(\frac{2}{T}\Big)  \qquad \text{for all $T \in (0,1)$.}$$
\end{lemma}
\begin{proof}
Using (\ref{N0expansionfirstfew}), we can estimate
\begin{align*}
|I_1(T)| \leq C \int_{0}^{4} \epsilon^2 (1 + |\ln \epsilon|) \frac{|\tanh(\frac{\epsilon}{2T}) - 1|}{\epsilon}  d\epsilon
\end{align*}
for $T \in (0, 1)$.
Using that $|\tanh(a) - 1 | \leq 2 e^{-2a}$ for $a \geq 0$ and then changing variables to $x = \frac{2\epsilon}{T}$, we find, again for $T \in (0,1)$,
\begin{align*}
|I_1(T)| 
& \leq C \int_{0}^{4} \epsilon (1 + |\ln \epsilon|)  e^{-\frac{\epsilon}{T}} d\epsilon
\leq C T^2 \int_{0}^{\frac{8}{T}} x \Big(1 + |\ln x| + \ln \frac{2}{T}\Big) e^{-\frac{x}{2}} dx
	\\
& \leq C T^2 \int_{0}^{\infty} x (1 + |\ln x|) e^{-\frac{x}{2}} dx
+ C T^2 \ln\Big(\frac{2}{T}\Big) \int_{0}^{\infty} x e^{-\frac{x}{2}} dx
 \leq C T^2  \ln\Big(\frac{2}{T}\Big),
\end{align*}
which is the desired conclusion.
\end{proof}

\begin{lemma}\label{I2lemma}
There exists a constant $C > 0$ such that
$$\bigg|I_2(T) - \ln\Big(\frac{2}{T}\Big) + \ln \Big(\frac{\pi }{4} \Big)  - \gamma\bigg|
 \leq  C \ln\Big(\frac{2}{T}\Big) e^{-\frac{4}{T}} \qquad \text{for all $T \in (0,1)$.}$$
\end{lemma}
\begin{proof}
Making the change of variables $x = \frac{\epsilon}{2T}$ in the expression for $I_2(T)$, we find
\begin{align*}
I_2(T) & = \int_{0}^{\frac{2}{T}}  \frac{\tanh{x}}{x}dx
 =  \ln(x) \tanh{x} \Big|_{0}^{\frac{2}{T}}
- \int_{0}^{\frac{2}{T}} \frac{\ln{x}}{\cosh^2{x}}dx
	\\
& =  \ln\Big(\frac{2}{T}\Big) \tanh\Big(\frac{2}{T}\Big)
- \int_{0}^\infty \frac{\ln{x}}{\cosh^2{x}} dx
+ \int_{\frac{2}{T}}^\infty  \frac{\ln{x}}{\cosh^2{x}}dx
\end{align*}
for $T \in (0,1)$.
Since $\int_{0}^\infty \frac{\ln{x}}{\cosh^2{x}} dx = \ln \left(\frac{\pi }{4}\right)-\gamma$, it follows that
\begin{align*}
& I_2(T) - \ln\Big(\frac{2}{T}\Big) + \ln \Big(\frac{\pi }{4}\Big)-\gamma
= R_1\Big(\frac{2}{T}\Big),
\end{align*}
where
\begin{align}\label{R1def}
R_1(a) :=  \ln(a) (\tanh(a) -1)
+ \int_{a}^\infty  \frac{\ln{x}}{\cosh^2{x}}dx.
\end{align}
Since, for $a \geq 2$,
\begin{align}\nonumber
|R_1(a)| 
& \leq \ln(a) |\tanh(a) -1|
+ \int_{a}^\infty  \frac{\ln{x}}{\cosh^2{x}}dx
\leq \ln(a) 2 e^{-2a}
+ C \int_{a}^\infty e^{-2x} \ln(x) dx
	\\ \label{R1logtanh}
& \leq \ln(a) 2 e^{-2a}
+ C e^{-2a} \ln(a) + C\int_{a}^\infty \frac{e^{-2x}}{x}dx
\leq C \ln(a) e^{-2a},
\end{align}
the desired conclusion follows.
\end{proof}

\begin{lemma}\label{I3lemma}
There exists a constant $C > 0$ such that
$$\bigg|I_3(T) + \frac{(\ln \frac{2}{T})^2}{2}  - \int_{0}^\infty \frac{(\ln x)^2}{2\cosh^2{x}} dx\bigg|
 \leq  C\Big( \ln\frac{2}{T}\Big)^2 e^{-\frac{4}{T}} \qquad \text{for all $T \in (0,1)$.}$$
\end{lemma}
\begin{proof}
Making the change of variables $x = \frac{\epsilon}{2T}$ in the expression for $I_3(T)$, we find
\begin{align*}
I_3(T) & 
= \int_{0}^{\frac{2}{T}} \ln\bigg(\frac{1}{x}\bigg) \frac{\tanh{x}}{x}dx
= - \frac{(\ln x)^2}{2} \tanh{x} \bigg|_{0}^{\frac{2}{T}}
+ \int_{0}^{\frac{2}{T}} \frac{(\ln x)^2}{2\cosh^2{x}} dx
	\\
& = -\frac{(\ln\frac{2}{T})^2}{2} \tanh\bigg(\frac{2}{T}\bigg)
+ \int_{0}^\infty \frac{(\ln x)^2}{2\cosh^2{x}} dx
- \int_{\frac{2}{T}}^\infty \frac{(\ln x)^2}{2\cosh^2{x}} dx.
\end{align*}
It follows that
\begin{align*}
& I_3(T) + \frac{(\ln\frac{2}{T})^2}{2} - \int_{0}^\infty \frac{(\ln x)^2}{2\cosh^2{x}} dx
= R_2\Big(\frac{2}{T}\Big),
\end{align*}
where
$$R_2(a) :=  \frac{(\ln a)^2}{2} (1 - \tanh(a)) - \int_{a}^\infty \frac{(\ln x)^2}{2\cosh^2{x}} dx.$$
Since, for $a \geq 2$,
\begin{align*}
|R_2(a)| & \leq \frac{(\ln a)^2}{2} (1 - \tanh(a)) + C \int_{a}^\infty (\ln x)^2 e^{-2x} dx
	\\
& \leq (\ln a)^2 e^{-2a} + C (\ln a)^2 e^{-2a} + C \int_{a}^\infty \frac{2\ln x}{x} e^{-2x} dx
 \leq  C (\ln a)^2 e^{-2a},
\end{align*}
the desired conclusion follows.
\end{proof}

The above lemmas suggest that we rewrite the equation (\ref{TceqI1234}) for $T_N$ as follows:
\begin{align}\nonumber
& \frac{1}{U} - a_0
+ \frac{\ln(\frac{8}{T_N})}{2\pi^2}\bigg(- \ln\Big(\frac{2}{T_N}\Big) + \ln \Big(\frac{\pi }{4} \Big)  - \gamma\bigg)
 + \frac{1}{2\pi^2}\bigg(\frac{(\ln \frac{2}{T_N})^2}{2}  - \int_{0}^\infty \frac{(\ln x)^2}{2\cosh^2{x}} dx\bigg)
	\\ \nonumber
= &\; I_1(T_N) + \frac{\ln(\frac{8}{T_N})}{2\pi^2}\bigg(I_2(T_N)
- \ln\Big(\frac{2}{T_N}\Big) + \ln \Big(\frac{\pi }{4} \Big)  - \gamma\bigg)  
	\\ \label{TNlogsinserted}
& + \frac{1}{2\pi^2}\bigg(I_3(T_N) + \frac{(\ln \frac{2}{T_N})^2}{2}  - \int_{0}^\infty \frac{(\ln x)^2}{2\cosh^2{x}} dx\bigg).
\end{align}
Applying Lemma \ref{I1lemma}, Lemma \ref{I2lemma}, and Lemma \ref{I3lemma}, we see that the absolute value of the right-hand side of (\ref{TNlogsinserted}) is bounded above by
\begin{align*}
&  C \bigg( T_N^2 \ln\Big(\frac{2}{T_N}\Big) + \ln\Big(\frac{8}{T_N}\Big) \ln\Big(\frac{2}{T_N}\Big) e^{-\frac{4}{T_N}}
+ \Big( \ln\frac{2}{T_N}\Big)^2 e^{-\frac{4}{T_N}}\bigg)
\leq C T_N^2 \ln\Big(\frac{2}{T_N}\Big)
\end{align*}
whenever $T_N \in (0,1)$. It follows that
\begin{align}\nonumber
 \frac{1}{U}  - a_0
&+ \frac{\ln(\frac{8}{T_N})}{2\pi^2}\bigg(- \ln\Big(\frac{2}{T_N}\Big) + \ln \Big(\frac{\pi }{4} \Big)  - \gamma\bigg)
 	\\ \label{logTNR3}
& + \frac{1}{2\pi^2}\bigg(\frac{(\ln \frac{2}{T_N})^2}{2}  - \int_{0}^\infty \frac{(\ln x)^2}{2\cosh^2{x}} dx\bigg)
  = R_3(T_N),
 \end{align}
 where the function $R_3$ satisfies
\begin{align}\label{R3bound}
|R_3(T)| \leq C T^2 \ln\Big(\frac{2}{T}\Big)\qquad \text{for all $T \in (0,1)$.}
\end{align}
Solving (\ref{logTNR3}) for $\ln T_N$, we obtain
\begin{align}\label{logTNsolvedfor}
\ln T_N(U,0)
=\gamma + \ln\Big(\frac{32}{\pi}\Big) 
-\sqrt{4 \pi ^2 \left(\frac{1}{U}-R_3(T_N(U,0))\right) + a_1},
\end{align}
where $a_1 \approx 0.3260$ is the constant defined in (\ref{a1def}).
We know from Lemma \ref{TNlemma} that $T_N \downarrow 0$ as  $U \downarrow 0$, so for all sufficiently small $U$ we have $T_N \in (0,1)$. Thus (\ref{R3bound}) shows that $R_3(T_N(U,0)) = O(1)$ as $U \downarrow 0$. But then (\ref{logTNsolvedfor}) implies that $T_N(U,0) = O(e^{-\frac{2\pi}{\sqrt{U}}})$ as $U \downarrow 0$. Substituting this back into (\ref{R3bound}), we infer that $R_3(T_N(U,0)) = O(e^{-\frac{4\pi}{\sqrt{U}}}/\sqrt{U})$ as $U \downarrow 0$. Inserting this in (\ref{logTNsolvedfor}), we arrive at (\ref{TNexpansion2D}). This completes the proof of Theorem \ref{2DNeelth}.

\subsection{Proof of Theorem \ref{3DNeelth}}\label{3DNeelsubsec}
Let $K$ be a compact subset of $(0,2)$ and suppose that $t_z \in K$.
Adding and subtracting terms, we can write equation (\ref{TNeelequationNtz}) for $T_N = T_N(U,t_z)$ as
\begin{align}\label{Tneeleq2}
\frac{1}{U} - b_0(t_z) = &\; J_1(t_z, T_N) + N_{t_z}(0) J_2(t_z, T_N),
\end{align}
where $b_0(t_z)$ is the function in (\ref{b0tzdef}), and $J_1$ and $J_2$ are defined by
\begin{align*}
& J_1(t_z, T) = \int_{0}^{4+2t_z} \big(N_{t_z}(\epsilon) - N_{t_z}(0) \big)\frac{\tanh(\frac{\epsilon}{2T}) - 1}{\epsilon}  d\epsilon
	\\
& J_2(t_z, T) = \int_{0}^{4+2t_z} \frac{\tanh(\frac{\epsilon}{2T})}{\epsilon}  d\epsilon.
\end{align*}

The following two lemmas describe the behavior of $J_1(t_z, T)$ and $J_2(t_z, T)$ for small $T$.

\begin{lemma}\label{J1lemma}
There exists a constant $C > 0$ such that
$$|J_1(t_z, T)| \leq  C T^2  \qquad \text{for all $t_z \in K$ and $T \in (0,1)$.}$$
\end{lemma}
\begin{proof}
From Theorem \ref{3Ddensityth}, we have
\begin{align*}
N_{t_z}(\epsilon)  
& = N_{t_z}(0) + O(\epsilon^2) \qquad \text{as $\epsilon \to 0$}
\end{align*}
uniformly for $t_z \in K$, where $N_{t_z}(0) = \int_{-2}^2 \frac{N_0(t_z u)}{\sqrt{4 - u^2}} \frac{du}{\pi}$.
Thus we can estimate
\begin{align*}
|J_1(t_z, T)| \leq C \int_{0}^{\infty} \epsilon^2 \frac{|\tanh(\frac{\epsilon}{2T}) - 1|}{\epsilon}  d\epsilon
\end{align*}
for $t_z \in K$ and $T \in (0, 1)$.
Using that $|\tanh(a) - 1 | \leq 2 e^{-2a}$ for $a \geq 0$ and then changing variables to $x = \frac{2\epsilon}{T}$, we find, again for $t_z \in K$ and $T \in (0,1)$,
\begin{align*}
|J_1(T)| 
& \leq C \int_{0}^\infty \epsilon e^{-\frac{\epsilon}{T}} d\epsilon
\leq C T^2 \int_{0}^\infty x e^{-\frac{x}{2}} dx \leq C T^2,
\end{align*}
which is the desired conclusion.
\end{proof}

\begin{lemma}\label{J2lemma}
There exists a constant $C > 0$ such that
$$\bigg|J_2(t_z, T) - \ln\Big(\frac{2+t_z}{T}\Big) + \ln \Big(\frac{\pi }{4} \Big)  - \gamma\bigg|
 \leq  C \ln\Big(\frac{2+t_z}{T}\Big) e^{-\frac{4+2t_z}{T}} \;\; \text{for all $t_z \in K$ and $T \in (0,1)$.}$$
\end{lemma}
\begin{proof}
Making the change of variables $x = \frac{\epsilon}{2T}$ in the expression for $J_2(t_z,T)$, we find
\begin{align*}
J_2(t_z,T) & = \int_{0}^{\frac{2+t_z}{T}}  \frac{\tanh{x}}{x}dx
 =  \ln(x) \tanh{x} \Big|_{0}^{\frac{2+t_z}{T}}
- \int_{0}^{\frac{2+t_z}{T}} \frac{\ln{x}}{\cosh^2{x}}dx
	\\
& =  \ln\Big(\frac{2+t_z}{T}\Big) \tanh\Big(\frac{2+t_z}{T}\Big)
- \int_{0}^\infty \frac{\ln{x}}{\cosh^2{x}} dx
+ \int_{\frac{2+t_z}{T}}^\infty  \frac{\ln{x}}{\cosh^2{x}}dx
\end{align*}
for $t_z \in K$ and $T \in (0,1)$.
Since $\int_{0}^\infty \frac{\ln{x}}{\cosh^2{x}} dx = \ln \left(\frac{\pi }{4}\right)-\gamma$, it follows that
\begin{align*}
& J_2(t_z,T) - \ln\Big(\frac{2+t_z}{T}\Big) + \ln \Big(\frac{\pi }{4}\Big)-\gamma
= R_1\Big(\frac{2+t_z}{T}\Big),
\end{align*}
where $R_1$ is the function in (\ref{R1def}). The desired conclusion now follows from (\ref{R1logtanh}).
\end{proof}

We rewrite the equation (\ref{Tneeleq2}) for $T_N$ as follows:
\begin{align}\nonumber
& \frac{1}{U} - b_0(t_z)
+ N_{t_z}(0)\bigg(- \ln\Big(\frac{2+t_z}{T_N}\Big) + \ln \Big(\frac{\pi }{4} \Big)  - \gamma\bigg)
	\\ \nonumber
= &\; J_1(t_z, T_N) + N_{t_z}(0)\bigg(J_2(t_z, T_N)
- \ln\Big(\frac{2+t_z}{T_N}\Big) + \ln \Big(\frac{\pi }{4} \Big)  - \gamma\bigg).
\end{align}
Applying Lemma \ref{J1lemma} and Lemma \ref{J2lemma}, we see that the absolute value of the right-hand side is bounded above by
\begin{align*}
&  C \bigg( T_N^2 + \ln\Big(\frac{2+t_z}{T_N}\Big) e^{-\frac{4+2t_z}{T_N}} \bigg)
\leq C T_N^2
\end{align*}
for $t_z \in K$ and $T_N \in (0,1)$. It follows that
\begin{align}\label{logTNR4}
& \frac{1}{U} - b_0(t_z)
+ N_{t_z}(0)\bigg(- \ln\Big(\frac{2+t_z}{T_N}\Big) + \ln \Big(\frac{\pi }{4} \Big)  - \gamma\bigg) = R_4(t_z, T_N),
 \end{align}
 where the function $R_4$ satisfies
\begin{align}\label{R4bound}
|R_4(t_z, T)| \leq C T^2 \qquad \text{for all $t_z \in K$ and $T \in (0,1)$.}
\end{align}
Solving (\ref{logTNR4}) for $\ln T_N$, we obtain
\begin{align}\label{logTNtzsolvedfor}
\ln T_N(U,t_z)
= \gamma + \ln\Big(\frac{8+4t_z}{\pi}\Big) - \frac{1}{N_{t_z}(0)}\bigg(\frac{1}{U} - b_0(t_z)\bigg) + \frac{R_4(t_z, T_N(U,t_z))}{N_{t_z}(0)}.
\end{align}
By Lemma \ref{TNlemma}, $T_N \downarrow 0$ as  $U \downarrow 0$ uniformly for $t_z \geq 0$, so for all sufficiently small $U$ we have $T_N \in (0,1)$. Thus (\ref{R4bound}) shows that $R_4(t_z, T_N(U,t_z)) = O(1)$ as $U \downarrow 0$ uniformly for $t_z \in K$. But then (\ref{logTNtzsolvedfor}) implies that $T_N(U,t_z) = O(e^{-\frac{1}{N_{t_z}(0) U}})$ as $U \downarrow 0$ uniformly for $t_z \in K$. Substituting this back into (\ref{R4bound}), we infer that $R_4(t_z, T_N(U,t_z)) = O(e^{-\frac{2}{N_{t_z}(0) U}})$ as $U \downarrow 0$ uniformly for $t_z \in K$. Inserting this in (\ref{logTNtzsolvedfor}), we arrive at (\ref{TNexpansion3D}). This completes the proof of Theorem \ref{3DNeelth}.

\section{Asymptotics for $\hat{m}$}\label{hatmsec}
In this section, we study the asymptotic behavior of the mean-field ratio $\hat{m}$ defined in (\ref{hatmdef}) as $U \downarrow 0$. Taking the difference of (\ref{mAFeq}) and (\ref{TNeelequationNtz}), and using that $N_{t_z}(\epsilon) = 0$ for $\epsilon > 4 + 2t_z$ as a consequence of (\ref{3Ddensityofstates}), we obtain
$$0 = \int_0^{4 + 2t_z} N_{t_z}(\epsilon)  
\bigg(\frac{\tanh(\frac{\sqrt{\Delta_{\AF}^2 + \epsilon^2}}{2T})}{\sqrt{\Delta_{\AF}^2 + \epsilon^2}}  -
\frac{\tanh(\frac{\epsilon}{2T_N})}{\epsilon} \bigg)d\epsilon,
$$
where $\Delta_{\AF} = \Delta_{\AF}(U, t_z, T)$ and $T_N = T_N(U,t_z)$. Performing the change of variables $\epsilon \to T_N\epsilon$, we get
\begin{align}\label{mAFeq3}
\int_0^{\frac{4 + 2t_z}{T_N}} N_{t_z}(T_N\epsilon)  
\bigg(\frac{\tanh(\frac{\sqrt{\hat{m}^2 + \epsilon^2}}{2T/T_N})}{\sqrt{\hat{m}^2 + \epsilon^2}}  -
\frac{\tanh(\frac{\epsilon}{2})}{\epsilon} \bigg) d\epsilon 
= 0,
\end{align}
where $\hat{m} = \hat{m}(U, t_z, T)$ is given by (\ref{hatmdef}). We will prove Theorems \ref{2Dmeanfieldth}--\ref{3Dmeanfieldimprovedth} by approximating $N_{t_z}(T_N\epsilon)$ in (\ref{mAFeq3}) by its leading asymptotics for small $T_N\epsilon$. Since $N_{0}(\epsilon)$ has a logarithmic singularity at $\epsilon = 0$ whereas $N_{t_z}(\epsilon)$ has no such singularity when $t_z > 0$, qualitatively different asymptotic formulas are obtained for $t_z = 0$ and $t_z > 0$. 

\subsection{Three lemmas}
Before turning to the proofs of the theorems, we prove three lemmas that will be used repeatedly. The first lemma establishes an upper bound on $\hat{m}$. Recall that, by convention, $\tanh(\cdot/y) = 1$ for $y = 0$.

\begin{lemma}\label{mhatboundlemma}
For $U> 0$, $t_z \geq 0$, and $T \in [0, T_N(U,t_z)]$, it holds that $\hat{m}(U, t_z, T) \in [0, 2)$.
\end{lemma}
\begin{proof}
The function $\epsilon \mapsto N_{t_z}(T_N\epsilon)$ is non-negative for all $\epsilon > 0$ and strictly positive for all sufficiently small $\epsilon > 0$. Consequently, in view of (\ref{mAFeq3}), it is enough to show that
\begin{align}\label{tanhinequality}
\frac{\tanh(\frac{\sqrt{x^2 + \epsilon^2}}{2y})}{\sqrt{x^2 + \epsilon^2}}  -
\frac{\tanh(\frac{\epsilon}{2})}{\epsilon} < 0 \qquad \text{whenever $\epsilon > 0$, $x \geq 2$, and $y \in [0,1]$}.
\end{align}

If $\epsilon > 0$, $x \geq 2$, and $y \in [0,1]$, then
\begin{align*}
\frac{\tanh(\frac{\sqrt{x^2 + \epsilon^2}}{2y})}{\sqrt{x^2 + \epsilon^2}}  -
\frac{\tanh(\frac{\epsilon}{2})}{\epsilon} 
\leq f(\epsilon), \qquad \text{where $f(\epsilon) := \frac{1}{\sqrt{4 + \epsilon^2}}  -
\frac{\tanh(\frac{\epsilon}{2})}{\epsilon}$}. 
\end{align*}
Thus (\ref{tanhinequality}) will follow if we can show that $f(\epsilon) < 0$ for all $\epsilon > 0$. As $\epsilon \to 0$, we have $f(\epsilon) = -\frac{\epsilon^2}{48} + O(\epsilon^4)$, so $f(\epsilon) < 0$ for all sufficiently small $\epsilon$. Moreover, $f(\epsilon) \neq 0$ for all $\epsilon > 0$, because $f(\epsilon) = 0$ if and only if $g(\epsilon) = 0$, where $g(\epsilon) := \arctanh(\frac{\epsilon}{\sqrt{4 + \epsilon^2}}) - \frac{\epsilon}{2}$.
But $g(\epsilon) = -\frac{\epsilon^3}{48} + O(\epsilon^5)$ as $\epsilon \to 0$ and $g'(\epsilon) = \frac{1}{\sqrt{4 + \epsilon^2}}  - \frac{1}{2} < 0$ for $\epsilon > 0$, so $g(\epsilon) < 0$ for all $\epsilon > 0$. This shows that $f(\epsilon) < 0$ for all $\epsilon > 0$ and thus completes the proof.
\end{proof}

The next lemma shows that the values of the integrals
$$\int_{0}^{2} \frac{\tanh(\frac{\sqrt{x^2 + \epsilon^2}}{2y})}{\sqrt{x^2 + \epsilon^2}} d\epsilon \geq 0 \quad \text{and} \quad
\int_{0}^{2} (1+|\ln{\epsilon}|) \frac{\tanh(\frac{\sqrt{x^2 + \epsilon^2}}{2y})}{\sqrt{x^2 + \epsilon^2}} d\epsilon \geq 0$$ 
tend to $+\infty$ if and only if $x \in [0,2]$ and $y \in [0,1]$ tend to $0$ simultaneously.

\begin{lemma}\label{int02lemma}
$(a)$ For any $r > 0$, there is a constant $C > 0$ such that
\begin{align}
\int_{0}^{2} (1+|\ln{\epsilon}|) \frac{\tanh(\frac{\sqrt{x^2 + \epsilon^2}}{2y})}{\sqrt{x^2 + \epsilon^2}} d\epsilon \leq C
\end{align}
for all $x \in [0,2]$ and $y \in [0,1]$ such that $\sqrt{x^2 + y^2} \geq r$. 

$(b)$ There exist constants $c > 0$ and $r > 0$ such that
\begin{align}
\int_{0}^{2} \frac{\tanh(\frac{\sqrt{x^2 + \epsilon^2}}{2y})}{\sqrt{x^2 + \epsilon^2}} d\epsilon \geq c|\ln(x^2+y^2)|
\end{align}
for all $x \in [0,2]$ and $y \in [0,1]$ such that $0 < \sqrt{x^2 + y^2} \leq r$. 
\end{lemma}
\begin{proof}
(a) Since $\tanh(a) \leq 1$ for $a \geq 0$, we have
$$\int_{0}^{2}(1+|\ln{\epsilon}|) \frac{\tanh(\frac{\sqrt{x^2 + \epsilon^2}}{2y})}{\sqrt{x^2 + \epsilon^2}} d\epsilon
\leq \int_{0}^{2}(1+|\ln{\epsilon}|) \frac{\tanh(\frac{\sqrt{x^2 + \epsilon^2}}{2y})}{x} d\epsilon
\leq \frac{C}{x}$$
and, since $\tanh(a) \leq a$ for $a \geq 0$, we have
$$\int_{0}^{2}(1+|\ln{\epsilon}|) \frac{\tanh(\frac{\sqrt{x^2 + \epsilon^2}}{2y})}{\sqrt{x^2 + \epsilon^2}} d\epsilon
\leq \int_{0}^{2}(1+|\ln{\epsilon}|) \frac{\frac{\sqrt{x^2 + \epsilon^2}}{2y}}{\sqrt{x^2 + \epsilon^2}} d\epsilon
\leq \frac{C}{2y}.$$
Hence
$$\int_{0}^{2}(1+|\ln{\epsilon}|) \frac{\tanh(\frac{\sqrt{x^2 + \epsilon^2}}{2y})}{\sqrt{x^2 + \epsilon^2}} d\epsilon
\leq \frac{C}{\max(x,y)},$$
which shows that the integral remains bounded if both $x$ and $y$ stay away from $0$. Thus, assertion $(a)$ follows.

Using that $\tanh(a) \geq a/2$ for $0 \leq a \leq 1$, and $\tanh(a) \geq 1/2$ for $a \geq 1$, we obtain, for $x \in [0,2]$ and $y \in [0,1]$,
\begin{align}\label{phiintphi}
\int_{0}^{2} \frac{\tanh(\frac{\sqrt{x^2 + \epsilon^2}}{2y})}{\sqrt{x^2 + \epsilon^2}} d\epsilon
\geq \frac{\phi(x,y)}{2}
\end{align}
where
$$\phi(x,y) := \begin{cases}
\int_{0}^{\sqrt{4y^2 - x^2}} \frac{1}{2y} d\epsilon
+ \int_{\sqrt{4y^2 - x^2}}^2 \frac{1}{\sqrt{x^2 + \epsilon^2}} d\epsilon & \text{if $4y^2 \geq x^2$},
	\\
\int_0^2 \frac{1}{\sqrt{x^2 + \epsilon^2}} d\epsilon & \text{if $4y^2 \leq x^2$}.
\end{cases}
$$
If $4y^2 \geq x^2 \geq 0$, then
$$\phi(x,y) = \sqrt{1 - \frac{x^2}{4y^2}}
+ \frac{1}{2} \ln\bigg(\frac{(\sqrt{4+x^2}+2) (1-\sqrt{1-\frac{x^2}{4 y^2}})}{(\sqrt{4+x^2}-2)(1 + \sqrt{1-\frac{x^2}{4 y^2}})}\bigg),$$
so for all sufficiently small $x, y \geq 0$ satisfying $4y^2 \geq x^2 \geq 0$, we have
\begin{align}\label{phigeq1}
\phi(x,y) \geq c\bigg| \ln\bigg(\frac{1-\sqrt{1-\frac{x^2}{4 y^2}}}{\sqrt{4+x^2}-2}\bigg)\bigg|
\geq c\bigg| \ln\bigg(\frac{\frac{x^2}{4 y^2}}{x^2}\bigg)\bigg| \geq c|\ln(y^2)| \geq c|\ln(x^2+y^2)|.
\end{align}
On the other hand, if $0 \leq 4y^2 \leq x^2$, then
$$\phi(x,y) = \frac{1}{2} \ln\bigg(\frac{\sqrt{4+x^2}+2}{\sqrt{4+x^2}-2}\bigg),$$
so for all sufficiently small $x, y \geq 0$ satisfying $0 \leq 4y^2  \leq x^2$, we have
\begin{align}\label{phigeq2}
\phi(x,y) \geq c|\ln(x^2)| \geq c|\ln(x^2 + y^2)|.
\end{align}
Assertion $(b)$ follows from (\ref{phiintphi}), (\ref{phigeq1}), and (\ref{phigeq2}). 
\end{proof}

We will also need the following lemma. 

\begin{lemma}\label{intxylemma}
$(a)$ There exists a $C > 0$ such that
\begin{align}\label{intepsilonxy}
\int_{2}^{\infty} \big( 1 + |\ln \epsilon |\big)
\bigg|\frac{\tanh(\frac{\sqrt{x^2 + \epsilon^2}}{2y})}{\sqrt{x^2 + \epsilon^2}}  -
\frac{\tanh(\frac{\epsilon}{2})}{\epsilon} \bigg| d\epsilon \leq C
\end{align}
for $x \in [0,2]$ and $y \in [0,1]$.

$(b)$ There exists a $C > 0$ such that, for $\delta = 0, 1$,
\begin{align}\label{intepsilonxyz}
\int_{z}^{\infty} 
\big( 1 + \delta |\ln \epsilon |\big) \bigg|\frac{\tanh(\frac{\sqrt{x^2 + \epsilon^2}}{2y})}{\sqrt{x^2 + \epsilon^2}}  -
\frac{\tanh(\frac{\epsilon}{2})}{\epsilon} \bigg| d\epsilon \leq C( 1 + \delta \ln z) \bigg(\frac{e^{-z}}{z}  + \frac{1}{z^2}\bigg),
\end{align}
for  $x \in [0,2]$, $y \in [0,1]$, and $z \geq 2$.

$(c)$ There exists a $C > 0$ such that, for $\delta = 0, 1$,
\begin{align}\label{intepsilon2xyz}
\int_{2}^{z} \epsilon^2 \big( 1 + \delta |\ln \epsilon |\big)
\bigg|\frac{\tanh(\frac{\sqrt{x^2 + \epsilon^2}}{2y})}{\sqrt{x^2 + \epsilon^2}}  -
\frac{\tanh(\frac{\epsilon}{2})}{\epsilon} \bigg| d\epsilon \leq C +  C|\ln z |\big( 1 + \delta |\ln z |\big)
\end{align}
for $x \in [0,2]$, $y \in [0,1]$, and $z \geq 2$. 
\end{lemma}
\begin{proof}
$(a)$ Let $h_1(x,y)$ be the left-hand side of (\ref{intepsilonxy}). For $x \in [0,2]$ and $y \in [0,1]$, we have
\begin{align*}
h_1(&x,y) \leq \int_{2}^{\infty} \big( 1 + |\ln \epsilon |\big)
\bigg(\bigg|\frac{\tanh(\frac{\sqrt{x^2 + \epsilon^2}}{2y}) - 1}{\sqrt{x^2 + \epsilon^2}}\bigg|  +
\bigg|\frac{\tanh(\frac{\epsilon}{2}) - 1}{\epsilon}\bigg| + 
\bigg|\frac{1}{\sqrt{x^2 + \epsilon^2}} - \frac{1}{\epsilon} \bigg| \bigg)d\epsilon.
\end{align*}
We next use the three inequalities (i) $|\tanh(a) - 1 | \leq 2 e^{-2a}$ for $a \geq 0$, (ii) $\sqrt{x^2 + \epsilon^2} \geq \epsilon y$ for $x \geq 0$ and $y \in [0,1]$, and (iii) $|\frac{1}{\sqrt{a + 1}} - 1| \leq \frac{a}{2}$ for $a \geq 0$. This gives
\begin{align}\nonumber
h_1(x,y) & \leq \int_{2}^{\infty} \big( 1 + |\ln \epsilon |\big)
\bigg(\frac{2 e^{-\frac{\sqrt{x^2 + \epsilon^2}}{y}}}{\sqrt{x^2 + \epsilon^2}} +
\frac{2 e^{-\epsilon}}{\epsilon} + 
 \frac{1}{\epsilon} \bigg|\frac{1}{\sqrt{\frac{x^2}{\epsilon^2} + 1}} - 1 \bigg| \bigg)d\epsilon 
	\\ \label{hestimate}
& \leq \int_{2}^{\infty} \big( 1 + |\ln \epsilon |\big)\bigg(\frac{4 e^{-\epsilon}}{\epsilon} + \frac{x^2}{2\epsilon^3} \bigg)d\epsilon 
\leq C
\end{align}
for $x \in [0,2]$ and $y \in [0,1]$, which proves $(a)$. 

$(b)$ The same steps that led to (\ref{hestimate}) yield, for $\delta = 0,1$,
\begin{align*}
\int_{z}^{\infty} \big( 1 + \delta |\ln \epsilon |\big)
\bigg|\frac{\tanh(\frac{\sqrt{x^2 + \epsilon^2}}{2y})}{\sqrt{x^2 + \epsilon^2}}  -
\frac{\tanh(\frac{\epsilon}{2})}{\epsilon} \bigg| d\epsilon
\leq \int_{z}^{\infty} \big( 1 + \delta |\ln \epsilon |\big) \bigg(\frac{4 e^{-\epsilon}}{\epsilon} + \frac{x^2}{2\epsilon^3} \bigg)d\epsilon
\end{align*}
for $x \in [0,2]$, $y \in [0,1]$, and $z \geq 2$. Since
$$\int_{z}^{\infty} \big( 1 + \delta |\ln \epsilon |\big) \bigg(\frac{4 e^{-\epsilon}}{\epsilon} + \frac{x^2}{2\epsilon^3} \bigg)d\epsilon 
\leq
C( 1 + \delta \ln z) \bigg(\frac{e^{-z}}{z}  + \frac{x^2}{z^2}\bigg),$$
the desired estimate follows.

$(c)$ The same steps that led to (\ref{hestimate}) yield, for $\delta = 0,1$,
\begin{align*}
\int_{2}^{z} \epsilon^2 \big( 1 + \delta |\ln \epsilon |\big)
\bigg|\frac{\tanh(\frac{\sqrt{x^2 + \epsilon^2}}{2y})}{\sqrt{x^2 + \epsilon^2}}  -
\frac{\tanh(\frac{\epsilon}{2})}{\epsilon} \bigg| d\epsilon
 \leq  \int_{2}^{z} \epsilon^2 \big( 1 + \delta |\ln \epsilon |\big)\bigg(\frac{4 e^{-\epsilon}}{\epsilon} + \frac{x^2}{2\epsilon^3} \bigg)d\epsilon 
\end{align*}
for $x \in [0,2]$, $y \in [0,1]$, and $z \geq 2$. Since
$$\int_{2}^{z} \epsilon^2 \big( 1 + \delta |\ln \epsilon |\big)\bigg(\frac{4 e^{-\epsilon}}{\epsilon} + \frac{x^2}{2\epsilon^3} \bigg)d\epsilon 
\leq C +  C|\ln z |\big( 1 + \delta |\ln z |\big)x^2,$$
the desired estimate follows.
\end{proof}

\subsection{Proof of Theorem \ref{2Dmeanfieldth}}\label{2Dmeanfieldsubsec}
Assume that $t_z=0$. We first show that $\hat{m}$ and $T/T_N$ cannot both be small.

\begin{lemma}\label{2DhatmTlemma}
There is an $r > 0$ such that
\begin{align}
\sqrt{\hat{m}(U, 0, T)^2 + \Big(\frac{T}{T_N(U,0)}\Big)^2} \geq r
\end{align}
for all sufficiently small $U > 0$ and all $T \in [0, T_N(U,0)]$.
\end{lemma}
\begin{proof}
By Theorem \ref{2DNeelth}, $T_N(U,0) \downarrow 0$ as $U \downarrow 0$.
Hence, in view of (\ref{mAFeq3}), it is enough to show that there is an $r>0$ such that
\begin{align}\label{2Dintxynonzero}
A(x,y,\tau) := \int_0^{\frac{4}{\tau}} N_{0}(\tau \epsilon)  
\bigg(\frac{\tanh(\frac{\sqrt{x^2 + \epsilon^2}}{2y})}{\sqrt{x^2 + \epsilon^2}}  -
\frac{\tanh(\frac{\epsilon}{2})}{\epsilon} \bigg) d\epsilon \neq 0
\end{align}
whenever $\tau \in [0,1/4]$ and $x, y \geq 0$ are such that $\sqrt{x^2 + y^2} \leq r$.
By Theorem \ref{2Ddensityth}, there is a $c_1 > 0$ such that 
$$N_{0}(\tau \epsilon) > c_1 (\ln 2 + |\ln(\tau \epsilon)|) > c_1 |\ln{\tau}|$$ 
for all $\tau \in [0,1/4]$ and $\epsilon \in [0,2]$.
Therefore, 
\begin{align*}
A(x,y,\tau) \geq &\; c_1 |\ln{\tau}| \int_0^2 \frac{\tanh(\frac{\sqrt{x^2 + \epsilon^2}}{2y})}{\sqrt{x^2 + \epsilon^2}} d\epsilon
- B(x,y,\tau)
\end{align*}
where
\begin{align*}
B(x,y,\tau) := \int_0^2 N_{0}(\tau \epsilon) \frac{\tanh(\frac{\epsilon}{2})}{\epsilon} d\epsilon
 + \int_2^\infty N_{0}(\tau \epsilon)  
\bigg|\frac{\tanh(\frac{\sqrt{x^2 + \epsilon^2}}{2y})}{\sqrt{x^2 + \epsilon^2}}  -
\frac{\tanh(\frac{\epsilon}{2})}{\epsilon} \bigg| d\epsilon \geq 0.
\end{align*}
Utilizing Lemma \ref{intxylemma} $(a)$, we find, for $x \in [0,2]$, $y \in [0,1]$, and $\tau \in [0,1/4]$,
\begin{align*}
B(x,y,\tau) \leq &\; C\int_0^2 (1 + |\ln{\tau}| + |\ln{\epsilon}|) \frac{\tanh(\frac{\epsilon}{2})}{\epsilon} d\epsilon
	\\
& + C\int_2^\infty (1 + |\ln{\tau}| + |\ln{\epsilon}|)  
\bigg|\frac{\tanh(\frac{\sqrt{x^2 + \epsilon^2}}{2y})}{\sqrt{x^2 + \epsilon^2}}  -
\frac{\tanh(\frac{\epsilon}{2})}{\epsilon} \bigg| d\epsilon
\leq C_1 |\ln \tau|, 
\end{align*}
for some $C_1 > 0$ independent of $x$, $y$, and $\tau$.
Lemma \ref{int02lemma} $(b)$ then shows that there are constants $c > 0$ and $r_1 > 0$ such that
\begin{align}
A(x,y,\tau) \geq |\ln{\tau}| (c |\ln(x^2+y^2)| - C_1) \geq c |\ln(x^2+y^2)| - C_1
\end{align}
for all $\tau \in [0,1/4]$ and all $x, y \geq 0$ such that $0 < \sqrt{x^2 + y^2} \leq r_1$. 
It follows that (\ref{2Dintxynonzero}) holds for all $\tau \in [0,1/4]$ and all $x, y \geq 0$ with $\sqrt{x^2 + y^2} \leq r$ if $r > 0$ is sufficiently small.
\end{proof}

Let $K$ be a compact subset of $[0,1)$ and assume that $T/T_N \in K$. Since $t_z=0$, (\ref{mAFeq3}) can be written as
\begin{align}\label{mAFeq4}
& \frac{\ln(\frac{1}{T_N})}{2 \pi^2} \int_0^{\frac{4}{T_N}} 
\bigg(\frac{\tanh(\frac{\sqrt{\hat{m}^2 + \epsilon^2}}{2T/T_N})}{\sqrt{\hat{m}^2 + \epsilon^2}}  -
\frac{\tanh(\frac{\epsilon}{2})}{\epsilon} \bigg) d\epsilon 
= R(U,T),
\end{align}
where $\hat{m} = \hat{m}(U, 0, T)$, $T_N = T_N(U,0)$, and
$$R(U,T)
:= -\int_{0}^{\frac{4}{T_N}} \bigg(N_0(T_N\epsilon)  - \frac{\ln(\frac{16}{T_N\epsilon})}{2 \pi^2} + \frac{\ln(\frac{16}{\epsilon})}{2 \pi^2} \bigg)  
\bigg(\frac{\tanh(\frac{\sqrt{\hat{m}^2 + \epsilon^2}}{2T/T_N})}{\sqrt{\hat{m}^2 + \epsilon^2}}  -
\frac{\tanh(\frac{\epsilon}{2})}{\epsilon} \bigg) d\epsilon.$$
Equation (\ref{mAFeq4}) implies that
\begin{align}\label{mAFeq5}
& J(\hat{m}, T/T_N)
= \frac{2 \pi^2 R(U,T)}{\ln(\frac{1}{T_N})} + \int_{\frac{4}{T_N}}^\infty \bigg(\frac{\tanh(\frac{\sqrt{\hat{m}^2 + \epsilon^2}}{2T/T_N})}{\sqrt{\hat{m}^2 + \epsilon^2}}  -
\frac{\tanh(\frac{\epsilon}{2})}{\epsilon} \bigg) d\epsilon,
\end{align}
where $J$ is the function defined in (\ref{Jdef}). 
By Theorem \ref{2Ddensityth}, we have 
$$\bigg|N_0(T_N\epsilon)  - \frac{\ln(\frac{16}{T_N\epsilon})}{2 \pi^2} \bigg| \leq C \qquad \text{for $\epsilon \in [0, 4/T_N]$},$$
and hence
$$|R(U,T)| \leq C \int_{0}^{\frac{4}{T_N}} \big( 1 + |\ln \epsilon |\big)
\bigg|\frac{\tanh(\frac{\sqrt{\hat{m}^2 + \epsilon^2}}{2T/T_N})}{\sqrt{\hat{m}^2 + \epsilon^2}}  -
\frac{\tanh(\frac{\epsilon}{2})}{\epsilon} \bigg| d\epsilon.$$

By Lemma \ref{mhatboundlemma}, we have $\hat{m} \in [0,2)$, and by Lemma \ref{2DhatmTlemma}, $\sqrt{\hat{m}^2 + (T/T_N)^2}$ is bounded away from $0$ for all sufficiently small $U > 0$. Hence we may apply Lemma \ref{int02lemma} $(a)$ with $x = \hat{m}$ and $y = T/T_N$ to see that
$$\int_{0}^2 \big( 1 + |\ln \epsilon |\big)
\bigg|\frac{\tanh(\frac{\sqrt{\hat{m}^2 + \epsilon^2}}{2T/T_N})}{\sqrt{\hat{m}^2 + \epsilon^2}}  -
\frac{\tanh(\frac{\epsilon}{2})}{\epsilon} \bigg| d\epsilon \leq C$$
for all small $U > 0$ and $T \in [0, T_N]$. Applying also parts $(a)$ and $(b)$ of Lemma \ref{intxylemma} with $x = \hat{m}$, $y = T/T_N$, $z = 4/T_N$, and $\delta = 0$, we obtain the following estimates for all small enough $U > 0$ and all $T \in [0, T_N]$:
\begin{align}\label{RUTbounded}
& |R(U,T)| \leq C,
	\\ \label{intsmallerthanTN2}
& \bigg|\int_{\frac{4}{T_N}}^\infty \bigg(\frac{\tanh(\frac{\sqrt{\hat{m}^2 + \epsilon^2}}{2T/T_N})}{\sqrt{\hat{m}^2 + \epsilon^2}}  -
\frac{\tanh(\frac{\epsilon}{2})}{\epsilon} \bigg) d\epsilon \bigg|
\leq C \big(T_N e^{-\frac{4}{T_N}}  + T_N^2 \big) \leq C T_N^2.
\end{align}
Consequently, using (\ref{mAFeq5}) and (\ref{TNexpansion2D}),
\begin{align}\label{Jestimate}
& |J(\hat{m}, T/T_N)|
\leq \frac{C}{\ln(\frac{1}{T_N})} + CT_N^2
\leq \frac{C}{\ln(\frac{1}{T_N})} \leq C \sqrt{U},
\end{align}
for all sufficiently small $U > 0$ and all $T \in [0, T_N]$.

The inequality (\ref{Jestimate}) shows that $J(\hat{m}, T/T_N) = O(\sqrt{U})$  for small $U$. Since $f_{\BCS}(y)$ is defined as the unique solution of $J(f_{\BCS}(y), y) =0$, this suggests that $\hat{m}(U,0,T)$ tends to $f_{\BCS}(T/T_N)$ as $U \downarrow 0$. We will use the following lemma to make this precise.

\begin{lemma}\label{GJlemma}
For all $x \geq 0$ and all $y \in [0,1]$ such that $e^{J(x, y)} y \leq 1$, it holds that
\begin{align}\label{fBCSJrelation}
x = f_{\BCS}(e^{J(x, y)} y) e^{-J(x, y)}.
\end{align}
\end{lemma}
\begin{proof}
Let $x \geq 0$ and $y \in [0,1]$.
We write the definition (\ref{Jdef}) of $J$ as 
$$J(x,y) = \lim_{M \to +\infty}\bigg\{\int_0^M \frac{\tanh(\frac{\sqrt{x^2 + \epsilon^2}}{2y})}{\sqrt{x^2 + \epsilon^2 }} d\epsilon
- \int_0^M \frac{\tanh(\frac{\epsilon}{2})}{\epsilon} d\epsilon\bigg\}.$$
Integrating by parts in both integrals using that
$$\frac{d}{d\epsilon}\ln(\epsilon + \sqrt{x^2 + \epsilon^2}) = \frac{1}{\sqrt{x^2 + \epsilon^2}},$$
we find
\begin{align*}
J(&x,y) =  \lim_{M \to +\infty}\bigg\{\bigg[\ln \left(\epsilon + \sqrt{x^2 + \epsilon^2}\right) \tanh\Big(\frac{\sqrt{x^2 + \epsilon^2}}{2y}\Big)\bigg]_{\epsilon = 0}^M
 - \bigg[\ln(\epsilon) \tanh\Big(\frac{\epsilon}{2}\Big)\bigg]_{\epsilon = 0}^M
	\\
&\hspace{1.3cm}   - \int_0^M \frac{ \ln\big(\epsilon + \sqrt{x^2 + \epsilon^2}\big) \epsilon}{\cosh^2(\frac{\sqrt{x^2 + \epsilon^2}}{2y}) 2y \sqrt{x^2 + \epsilon^2}} d\epsilon
+ \int_0^M  \frac{\ln(\epsilon)}{2\cosh^2(\epsilon/2)} d\epsilon \bigg\}
	\\
= &\; \ln(2) - \ln(x) \tanh\Big(\frac{x}{2y}\Big)
- \int_0^\infty \frac{\ln \big(\epsilon + \sqrt{x^2 + \epsilon^2}\big) \epsilon}{\cosh^2(\frac{\sqrt{x^2 +  \epsilon^2}}{2y}) 2y \sqrt{x^2 + \epsilon^2}} d\epsilon
+ \int_0^\infty \frac{\ln(\epsilon)}{2\cosh^2(\epsilon/2)} d\epsilon .
\end{align*}
The last integral can be computed exactly and equals $\ln \left(\frac{\pi }{2}\right)- \gamma$.
The change of variables $\epsilon = \sqrt{s^2 - x^2}$ in the first integral on the right-hand side gives
\begin{align*}
J(x,y) = &\; \ln(2)  - \ln(x) \tanh\Big(\frac{x}{2y}\Big)
- \int_{x}^\infty \frac{\ln\big(s + \sqrt{s^2-x^2}\big) }{2 \cosh^2(\frac{s}{2y})y} ds
 + \ln \left(\frac{\pi }{2}\right) - \gamma.
\end{align*}
Changing variables $s \to sx$, we arrive at
\begin{align*}
J(x, y) = &\; \ln(\pi) - \gamma - H\Big(\frac{x}{2y}\Big) - \ln x,
\end{align*}
where $H:[0,+\infty) \to [0,+\infty)$ is defined by
$$H(a) = a \int_{1}^\infty \frac{\ln \big(s + \sqrt{s^2 -1}\big)}{\cosh^2(s a)} ds.$$
It follows that
$$J(xe^{J(x, y)}, e^{J(x, y)} y) 
= \ln(\pi) - \gamma - H\Big(\frac{x}{2y}\Big) - \ln x - J(x, y) = 0.$$
Since, by definition, $f_{\BCS}(e^{J(x, y)} y)$ is the unique solution of $J(f_{\BCS}(e^{J(x, y)} y), e^{J(x, y)} y) =0$, we conclude that $xe^{J(x, y)} = f_{\BCS}(e^{J(x, y)} y),$ which is the desired conclusion.
\end{proof}

Taylor's theorem applied to $z \mapsto f_{\BCS}(e^{z}y)e^{-z}$ gives the identity
\begin{align*}
f_{\BCS}(e^{z}y)e^{-z} = &\; f_{\BCS}(y) + \big(y f_{\BCS}'(y) - f_{\BCS}(y)\big)z 
	\\
& + \int_0^{z} \Big(y^2 f_{\BCS}''(e^t y)e^t  -y f_{\BCS}'(e^t y)+f_{\BCS}(e^t y)e^{-t}  \Big) (z-t)  dt
\end{align*}
for all $z \in \R$ and $y \in [0,1)$ such that $e^{z}y < 1$. Since $f_{\BCS}(y), f_{\BCS}'(y), f_{\BCS}''(y)$ are bounded on any compact subset of $[0, 1)$, we find
\begin{align}\label{fBCStaylor}
f_{\BCS}(e^{z}y)e^{-z} = &\; f_{\BCS}(y) + \big(y f_{\BCS}'(y) - f_{\BCS}(y)\big)z + O(z^2) \qquad \text{as $z \to 0$}
\end{align}
uniformly for $y \in K$.
By (\ref{fBCSJrelation}),
$$\hat{m} = f_{\BCS}\Big(e^{J(\hat{m}, T/T_N)}\frac{T}{T_N}\Big) e^{-J(\hat{m}, T/T_N)},$$
where, by (\ref{Jestimate}), $J(\hat{m}, T/T_N) = O(\sqrt{U})$ uniformly for $T/T_N \in K$ as $U \downarrow 0$. 
Together with (\ref{fBCStaylor}), this yields
\begin{align}\label{mhattaylor}
\hat{m} = &\; f_{\BCS}\Big(\frac{T}{T_N}\Big) + \bigg(\frac{T}{T_N} f_{\BCS}'\Big(\frac{T}{T_N}\Big) - f_{\BCS}\Big(\frac{T}{T_N}\Big)\bigg)J\Big(\hat{m},\frac{T}{T_N}\Big) + O(U),
\end{align}
as $U \downarrow 0$ uniformly for $T/T_N \in K$.
Applying also (\ref{mAFeq5}) and (\ref{intsmallerthanTN2}), we infer that
\begin{align}\label{hatmexpansion}
\hat{m} = &\; f_{\BCS}\Big(\frac{T}{T_N}\Big) + \bigg(\frac{T}{T_N} f_{\BCS}'\Big(\frac{T}{T_N}\Big) - f_{\BCS}\Big(\frac{T}{T_N}\Big)\bigg)\bigg(\frac{2 \pi^2 R(U,T)}{\ln(\frac{1}{T_N})} + O(T_N^2)\bigg) + O(U)
\end{align}
as $U \downarrow 0$ uniformly for all $T$ such that $T/T_N \in K$.

The next step in the proof is to compute the small $U$ limit of $R(U,T)$. To do this, we need the following lemma. 

\begin{lemma}\label{tanhsqrtlemma}
As $U \downarrow 0$, 
\begin{align}\label{tanhsqrtexpansion}
\frac{\tanh(\frac{\sqrt{\hat{m}^2 + \epsilon^2}}{2T/T_N})}{\sqrt{\hat{m}^2 + \epsilon^2}} 
= \frac{\tanh(\frac{\sqrt{f_{\BCS}(T/T_N)^2 + \epsilon^2}}{2T/T_N})}{\sqrt{f_{\BCS}(T/T_N)^2 + \epsilon^2}}\bigg(1 + O\Big(\frac{\sqrt{U}}{1 + \epsilon^2}\Big)\bigg)
\end{align}
uniformly for $\epsilon > 0$ and $T$ such that $T/T_N \in K$.
\end{lemma}
\begin{proof}
By (\ref{hatmexpansion}) and (\ref{RUTbounded}), we have $\hat{m} = f_{\BCS}(T/T_N) + O(\sqrt{U})$ as $U \downarrow 0$ uniformly for $T/T_N \in K$, and hence, by Taylor's theorem,
\begin{align*}
\tanh\bigg(\frac{\sqrt{\hat{m}^2 + \epsilon^2}}{2T/T_N}\bigg) 
& = \tanh\bigg(\frac{\sqrt{f_{\BCS}(T/T_N)^2 + \epsilon^2}}{2T/T_N}\bigg) 
+ \frac{\xi \sech^2(\frac{\sqrt{\xi^2 + \epsilon^2}}{2T/T_N}) }{2\frac{T}{T_N} \sqrt{\xi^2 + \epsilon^2}}(\hat{m} - f_{\BCS}(T/T_N)),
\end{align*}
where $\xi = \xi(U,T,\epsilon)$ lies in the interval between $\hat{m}$ and $f_{\BCS}(T/T_N)$.
For $\epsilon > 0$, $T/T_N \in K$, and $U > 0$ sufficiently small, we have
\begin{align*}
\frac{1}{\tanh(\frac{\sqrt{f_{\BCS}(T/T_N)^2 + \epsilon^2}}{2T/T_N}) } \frac{\xi \sech^2(\frac{\sqrt{\xi^2 + \epsilon^2}}{2T/T_N}) }{2\frac{T}{T_N} \sqrt{\xi^2 + \epsilon^2}}
& \leq C \frac{e^{-\frac{\sqrt{\xi ^2+\epsilon ^2}}{T/T_N}}}{2\frac{T}{T_N} \sqrt{\xi ^2+\epsilon ^2}}
	\\
& \leq 
C \frac{\sqrt{\xi ^2+\epsilon ^2}}{T/T_N} e^{-\frac{\sqrt{\xi ^2+\epsilon ^2}}{T/T_N}}
 \frac{1}{\xi^2+\epsilon ^2}
 \leq \frac{C}{1+\epsilon ^2}
\end{align*}
and therefore
\begin{align}\label{tanhexpansion}
\tanh\bigg(\frac{\sqrt{\hat{m}^2 + \epsilon^2}}{2T/T_N}\bigg) 
& = \tanh\bigg(\frac{\sqrt{f_{\BCS}(T/T_N)^2 + \epsilon^2}}{2T/T_N}\bigg) 
\bigg(1+ O\Big(\frac{\sqrt{U}}{1 + \epsilon^2}\Big)\bigg)
\end{align}
uniformly for $\epsilon > 0$ and $T/T_N \in K$ as $U \downarrow 0$. Similarly,
$$\sqrt{\hat{m}^2 + \epsilon^2}
= \sqrt{f_{\BCS}(T/T_N)^2 + \epsilon^2} + 
\frac{\eta}{\sqrt{\eta^2 + \epsilon^2}} (\hat{m} - f_{\BCS}(T/T_N))
$$
where $\eta = \eta(U,T,\epsilon)$ lies in the interval between $\hat{m}$ and $f_{\BCS}(T/T_N)$.
For $\epsilon > 0$, $T/T_N \in K$, and $U > 0$ sufficiently small, we have
$$\frac{\eta}{\sqrt{f_{\BCS}(T/T_N)^2 + \epsilon^2} \sqrt{\eta^2 + \epsilon^2}}
\leq \frac{C}{1 + \epsilon^2}$$
and therefore
\begin{align}\label{sqrtexpansion}
\sqrt{\hat{m}^2 + \epsilon^2}
= \sqrt{f_{\BCS}(T/T_N)^2 + \epsilon^2}\bigg(1+ O\Big(\frac{\sqrt{U}}{1 + \epsilon^2}\Big)\bigg).
\end{align}
uniformly for $\epsilon > 0$ and $T/T_N \in K$ as $U \downarrow 0$.
The desired conclusion follows by combining (\ref{tanhexpansion}) and (\ref{sqrtexpansion}).
\end{proof}

We can now determine the limiting behavior of $R(U,T)$ as $U$ tends to $0$.

\begin{lemma} \label{Rlemma}
As $U \downarrow 0$,
$$R(U,T) = -\int_{0}^\infty \frac{\ln(\frac{16}{\epsilon})}{2 \pi^2}
\bigg(\frac{\tanh(\frac{\sqrt{f_{\BCS}(T/T_N)^2 + \epsilon^2}}{2T/T_N})}{\sqrt{f_{\BCS}(T/T_N)^2 + \epsilon^2}}  
- \frac{\tanh(\frac{\epsilon}{2})}{\epsilon} \bigg) d\epsilon
+ O(\sqrt{U})$$
uniformly for all $T$ such that $T/T_N\in K$.
\end{lemma}
\begin{proof}
We have $R(U,T) = R_1(U,T) + R_2(U,T)$, where
\begin{align}
& R_1(U,T) = -\int_{0}^{\frac{4}{T_N}} \bigg(N_0(T_N\epsilon)  - \frac{\ln(\frac{16}{T_N\epsilon})}{2 \pi^2}\bigg)  
\bigg(\frac{\tanh(\frac{\sqrt{\hat{m}^2 + \epsilon^2}}{2T/T_N})}{\sqrt{\hat{m}^2 + \epsilon^2}}  -
\frac{\tanh(\frac{\epsilon}{2})}{\epsilon} \bigg) d\epsilon,
	\\ \label{R2def}
& R_2(U,T) = -\int_{0}^{\frac{4}{T_N}} \frac{\ln(\frac{16}{\epsilon})}{2 \pi^2}
\bigg(\frac{\tanh(\frac{\sqrt{\hat{m}^2 + \epsilon^2}}{2T/T_N})}{\sqrt{\hat{m}^2 + \epsilon^2}}  -
\frac{\tanh(\frac{\epsilon}{2})}{\epsilon} \bigg) d\epsilon.
\end{align}
Using the expansion of $N_0(\epsilon)$ proved in Theorem \ref{2Ddensityth}, we see that 
$$|R_1(U,T)| \leq C\int_{0}^{\frac{4}{T_N}} (T_N\epsilon)^2 \big( 1 + |\ln(T_N \epsilon) |\big)  
\bigg|\frac{\tanh(\frac{\sqrt{\hat{m}^2 + \epsilon^2}}{2T/T_N})}{\sqrt{\hat{m}^2 + \epsilon^2}}  -
\frac{\tanh(\frac{\epsilon}{2})}{\epsilon} \bigg| d\epsilon.$$
Recall that $\hat{m} \in [0,2)$ and $\sqrt{\hat{m}^2 + (T/T_N)^2} \geq r > 0$ by Lemma \ref{mhatboundlemma} and Lemma \ref{2DhatmTlemma}. Thus, applying Lemma \ref{int02lemma} $(a)$ and Lemma \ref{intxylemma} $(c)$ with $x = \hat{m}$, $y = T/T_N$, and $z = 4/T_N$, and using also the fact that $T_N = O(e^{-\frac{2 \pi}{\sqrt{U}}})$ as $U \downarrow 0$ by (\ref{TNexpansion2D}), we obtain
\begin{align}\label{R1estimate}
|R_1(U,T)| \leq C T_N^2 (1 + |\ln T_N|^2) \leq C U^{-1} e^{-\frac{4 \pi}{\sqrt{U}}}
\end{align}
for all sufficiently small $U > 0$ and all $T$ such that $T/T_N \in K$.

Part $(b)$ of Lemma \ref{intxylemma} with $\delta = 1$ implies that the upper limit of integration in (\ref{R2def}) can be replaced by $\infty$ with an error of order
$$O\Big(( 1 + |\ln T_N|) \big(T_N e^{-\frac{4}{T_N} } + T_N^2 \big)\Big)
= O\Big( U^{-1/2} e^{-\frac{4\pi}{\sqrt{U}}} \Big).$$
Employing also Lemma \ref{tanhsqrtlemma}, it transpires that
\begin{align}\nonumber
R_2(U,T) = & -\int_{0}^\infty \frac{\ln(\frac{16}{\epsilon})}{2 \pi^2}  
\bigg(\frac{\tanh(\frac{\sqrt{f_{\BCS}(T/T_N)^2 + \epsilon^2}}{2T/T_N})}{\sqrt{f_{\BCS}(T/T_N)^2 + \epsilon^2}}  
- \frac{\tanh(\frac{\epsilon}{2})}{\epsilon} \bigg) d\epsilon
	\\\nonumber
&+ O\bigg(\int_{0}^\infty \bigg| \frac{\ln(\frac{16}{\epsilon})}{2 \pi^2}  
\frac{\tanh(\frac{\sqrt{f_{\BCS}(T/T_N)^2 + \epsilon^2}}{2T/T_N})}{\sqrt{f_{\BCS}(T/T_N)^2 + \epsilon^2}} \bigg| \frac{\sqrt{U}}{1 + \epsilon^2} d\epsilon \bigg)
+ O\Big( U^{-1/2} e^{-\frac{4\pi}{\sqrt{U}}} \Big)
	\\ \label{R2estimate}
= & -\int_{0}^\infty \frac{\ln(\frac{16}{\epsilon})}{2 \pi^2}  
\bigg(\frac{\tanh(\frac{\sqrt{f_{\BCS}(T/T_N)^2 + \epsilon^2}}{2T/T_N})}{\sqrt{f_{\BCS}(T/T_N)^2 + \epsilon^2}}  
- \frac{\tanh(\frac{\epsilon}{2})}{\epsilon} \bigg) d\epsilon
+ O(\sqrt{U})
\end{align}
uniformly for $T$ such that $T/T_N \in K$.
Since $R = R_1 + R_2$, the desired conclusion follows from (\ref{R1estimate}) and (\ref{R2estimate}).
\end{proof}

Substituting the expansion of $R(U,T)$ established in Lemma \ref{Rlemma} into (\ref{hatmexpansion}), we infer that
\begin{align}\nonumber
\hat{m}(U,0,T) = &\; f_{\BCS}\Big(\frac{T}{T_N}\Big) - \bigg(\frac{T}{T_N} f_{\BCS}'\Big(\frac{T}{T_N}\Big) - f_{\BCS}\Big(\frac{T}{T_N}\Big)\bigg)\frac{2 \pi^2}{\ln(\frac{1}{T_N})}
	\\
&\times \int_{0}^\infty \frac{\ln(\frac{16}{\epsilon})}{2 \pi^2}  
\bigg(\frac{\tanh(\frac{\sqrt{f_{\BCS}(T/T_N)^2 + \epsilon^2}}{2T/T_N})}{\sqrt{f_{\BCS}(T/T_N)^2 + \epsilon^2}}  
- \frac{\tanh(\frac{\epsilon}{2})}{\epsilon} \bigg) d\epsilon  + O(U)
\end{align}
as $U \downarrow 0$ uniformly for $T/T_N \in K$.
Since, by (\ref{TNexpansion2D}), 
$$\frac{1}{\ln\frac{1}{T_N}} = \frac{\sqrt{U}}{2\pi} + O(U) \qquad \text{as $U \downarrow 0$},$$
the expansion in (\ref{2Dhatmexpansion}) follows. Since $f_{\BCS}(0) = \pi e^{-\gamma}$, the expansion in (\ref{c1at0}) follows when $T = 0$. This completes the proof of Theorem \ref{2Dmeanfieldth}.

\subsection{Proof of Theorem \ref{3Dmeanfieldth}}\label{3Dmeanfieldsubsec}
Let $K_1$ and $K_2$ be compact subsets of $(0, 2)$ and $[0,1)$, respectively, and suppose that $t_z \in K_1$ and $T/T_N(U, t_z) \in K_2$. The next lemma is analogous to Lemma \ref{2DhatmTlemma}.

\begin{lemma}\label{3DhatmTlemma}
There is an $r > 0$ such that
$$\sqrt{\hat{m}(U, t_z, T)^2 + \Big(\frac{T}{T_N(U,t_z)}\Big)^2} \geq r$$ 
for all sufficiently small $U > 0$, all $t_z \in K_1$, and all $T$ such that $T/T_N(U, t_z) \in K_2$.
\end{lemma}
\begin{proof}
By Theorem \ref{3DNeelth}, $T_N(U,t_z) \downarrow 0$ as $U \downarrow 0$ uniformly for $t_z \in K_1$.
Hence, in view of (\ref{mAFeq3}), it is enough to show that there is an $r>0$ such that
\begin{align}\label{3Dintxynonzero}
A(t_z, x,y,\tau) := \int_0^{\frac{4 + 2t_z}{\tau}} N_{t_z}(\tau \epsilon)  
\bigg(\frac{\tanh(\frac{\sqrt{x^2 + \epsilon^2}}{2y})}{\sqrt{x^2 + \epsilon^2}}  -
\frac{\tanh(\frac{\epsilon}{2})}{\epsilon} \bigg) d\epsilon \neq 0
\end{align}
whenever $t_z \in K_1$, $\tau \in [0,1/4]$, and $x, y \geq 0$ are such that $\sqrt{x^2 + y^2} \leq r$.
By Theorem \ref{3Ddensityth}, there is a $c_1 > 0$ such that 
$$N_{t_z}(\tau \epsilon) > c_1$$ 
for all $\tau \in [0,1/4]$, $\epsilon \in [0,2]$, and $t_z \in K_1$.
Therefore, 
\begin{align*}
A(t_z, x,y,\tau) \geq &\; c_1 \int_0^2 \frac{\tanh(\frac{\sqrt{x^2 + \epsilon^2}}{2y})}{\sqrt{x^2 + \epsilon^2}} d\epsilon
- B(t_z, x,y,\tau)
\end{align*}
where
\begin{align*}
B(t_z, x,y,\tau) := \int_0^2 N_{t_z}(\tau \epsilon) \frac{\tanh(\frac{\epsilon}{2})}{\epsilon} d\epsilon
 + \int_2^\infty N_{t_z}(\tau \epsilon)  
\bigg|\frac{\tanh(\frac{\sqrt{x^2 + \epsilon^2}}{2y})}{\sqrt{x^2 + \epsilon^2}}  -
\frac{\tanh(\frac{\epsilon}{2})}{\epsilon} \bigg| d\epsilon \geq 0.
\end{align*}
Utilizing Lemma \ref{intxylemma} $(a)$, we find, for $t_z \in K_1$, $x \in [0,2]$, $y \in [0,1]$, and $\tau \in [0,1/4]$,
\begin{align*}
B(t_z, x,y,\tau) \leq &\; C\int_0^2  \frac{\tanh(\frac{\epsilon}{2})}{\epsilon} d\epsilon
+ C\int_2^\infty  
\bigg|\frac{\tanh(\frac{\sqrt{x^2 + \epsilon^2}}{2y})}{\sqrt{x^2 + \epsilon^2}}  -
\frac{\tanh(\frac{\epsilon}{2})}{\epsilon} \bigg| d\epsilon
\leq C_1, 
\end{align*}
for some $C_1 > 0$ independent of $t_z$, $x$, $y$, and $\tau$.
Lemma \ref{int02lemma} $(b)$ then shows that there are constants $c > 0$ and $r_1 > 0$ such that
\begin{align}
A(t_z, x,y,\tau) \geq c |\ln(x^2+y^2)| - C_1
\end{align}
for all $t_z \in K_1$, $\tau \in [0,1/4]$, and all $x,y \geq 0$ such that $0 < \sqrt{x^2 + y^2} \leq r_1$. 
It follows that (\ref{3Dintxynonzero}) holds for all $t_z \in K_1$, $\tau \in [0,1/4]$, and all $x,y \geq 0$ with $\sqrt{x^2 + y^2} \leq r$ if $r > 0$ is sufficiently small.
\end{proof}

Let us write (\ref{mAFeq3}) as
\begin{align}\label{mAFeq43D}
 \int_0^{\frac{4 + 2t_z}{T_N}} 
\bigg(\frac{\tanh(\frac{\sqrt{\hat{m}^2 + \epsilon^2}}{2T/T_N})}{\sqrt{\hat{m}^2 + \epsilon^2}}  -
\frac{\tanh(\frac{\epsilon}{2})}{\epsilon} \bigg) d\epsilon 
= E(U,t_z,T),
\end{align}
where $\hat{m} = \hat{m}(U, t_z, T)$, $T_N = T_N(U,t_z)$, and
$$E(U,t_z,T)
:= -\int_0^{\frac{4 + 2t_z}{T_N}} \frac{N_{t_z}(T_N\epsilon)  - N_{t_z}(0)}{N_{t_z}(0)} 
\bigg(\frac{\tanh(\frac{\sqrt{\hat{m}^2 + \epsilon^2}}{2T/T_N})}{\sqrt{\hat{m}^2 + \epsilon^2}}  -
\frac{\tanh(\frac{\epsilon}{2})}{\epsilon} \bigg) d\epsilon.$$
Equation (\ref{mAFeq43D}) implies that
\begin{align}\label{mAFeq6}
& J(\hat{m}, T/T_N)
= E(U,t_z,T) + \int_{\frac{4 + 2t_z}{T_N}}^\infty \bigg(\frac{\tanh(\frac{\sqrt{\hat{m}^2 + \epsilon^2}}{2T/T_N})}{\sqrt{\hat{m}^2 + \epsilon^2}}  -
\frac{\tanh(\frac{\epsilon}{2})}{\epsilon} \bigg) d\epsilon,
\end{align}
where $J$ is the function defined in (\ref{Jdef}). 
By Theorem \ref{3DNeelth},
\begin{align}\label{3DTNbound}
T_N(U, t_z) = O\big(e^{- \frac{1}{N_{t_z}(0) U}}\big)
\qquad \text{as $U \downarrow 0$}
\end{align}
uniformly for $t_z \in K_1$. 
From the expression for $N_{t_z}(0)$ given in Theorem \ref{3Ddensityth} we deduce that there is a $C > 0$ such that $1/C \leq N_{t_z}(0) \leq C$ for all $t_z \in K_1$. In particular, $T_N(U, t_z) \to 0$ as $U \downarrow 0$ uniformly for $t_z \in K_1$. Furthermore, by Lemma \ref{mhatboundlemma}, $\hat{m} \in [0, 2)$. 
Thus part $(b)$ of Lemma \ref{intxylemma} with $\delta = 0$ shows that the integral on the right-hand side of (\ref{mAFeq6}) is bounded above by
\begin{align}\label{estimate1}
\int_{\frac{4 + 2t_z}{T_N}}^\infty \bigg|\frac{\tanh(\frac{\sqrt{\hat{m}^2 + \epsilon^2}}{2T/T_N})}{\sqrt{\hat{m}^2 + \epsilon^2}}  -
\frac{\tanh(\frac{\epsilon}{2})}{\epsilon} \bigg| d\epsilon 
\leq C\big(T_N e^{-\frac{4}{T_N}}  + T_N^2 \big)
\leq C T_N^2  \leq Ce^{- \frac{2}{N_{t_z}(0) U}}
\end{align}
for all sufficiently small $U> 0$, uniformly for $t_z \in K_1$ and $T/T_N \in K_2$.
On the other hand, by Theorem \ref{3Ddensityth}, we have 
$$\bigg|\frac{N_{t_z}(T_N\epsilon)  - N_{t_z}(0)}{N_{t_z}(0)} \bigg| \leq C T_N^2 \epsilon^2 \qquad \text{for $\epsilon \in \bigg[0, \frac{4 + 2t_z}{T_N}\bigg]$ and $t_z \in K_1$},$$
and hence
$$|E(U,t_z,T)| \leq C T_N^2 \int_{0}^{\frac{4 + 2t_z}{T_N}} \epsilon^2
\bigg|\frac{\tanh(\frac{\sqrt{\hat{m}^2 + \epsilon^2}}{2T/T_N})}{\sqrt{\hat{m}^2 + \epsilon^2}}  -
\frac{\tanh(\frac{\epsilon}{2})}{\epsilon} \bigg| d\epsilon$$
for all sufficiently small $U> 0$ and uniformly for $t_z \in K_1$ and $T/T_N \in K_2$.
Since $\hat{m} \in [0,2)$ and $\sqrt{\hat{m}^2 + (T/T_N)^2} \geq r > 0$ by Lemma \ref{mhatboundlemma} and Lemma \ref{3DhatmTlemma}, we may employ Lemma \ref{int02lemma} $(a)$ and Lemma \ref{intxylemma} $(c)$ with $\delta = 0$ to get
\begin{align}\label{Eestimate}
|E(U,t_z,T)| \leq C T_N^2 \big(1 + |\ln T_N | \big)
\leq C T_N^2 |\ln T_N |
\leq C U^{-1}e^{- \frac{2}{N_{t_z}(0) U}}
\end{align}
for all sufficiently small $U> 0$ and uniformly for $t_z \in K_1$ and $T/T_N \in K_2$.
Utilizing (\ref{estimate1}) and (\ref{Eestimate}) in (\ref{mAFeq6}), we conclude that 
\begin{align}\label{mAFeq7}
& J(\hat{m}, T/T_N) = O\Big(U^{-1} e^{- \frac{2}{N_{t_z}(0) U}} \Big)
\end{align}
and hence, by the same argument that led to (\ref{mhattaylor}),
\begin{align}\nonumber
\hat{m} & = f_{\BCS}\Big(\frac{T}{T_N}\Big) 
+ \bigg(\frac{T}{T_N} f_{\BCS}'\Big(\frac{T}{T_N}\Big) - f_{\BCS}\Big(\frac{T}{T_N}\Big)\bigg)J(\hat{m}, T/T_N)
+ O\bigg(\frac{e^{- \frac{4}{N_{t_z}(0) U}}}{U^2} \bigg)
	\\ \label{hatmfBCSJ}
& = f_{\BCS}\Big(\frac{T}{T_N}\Big) + O\Big(U^{-1} e^{- \frac{2}{N_{t_z}(0) U}} \Big)
\end{align}
as $U \downarrow 0$ uniformly for $t_z \in K_1$ and $T\geq 0$ such that $T/T_N \in K_2$.
This completes the proof of Theorem \ref{3Dmeanfieldth}.

\subsection{Proof of Theorem \ref{3Dmeanfieldimprovedth}}\label{3Dmeanfieldimprovedsubsec}
As in the proof of Theorem \ref{3Dmeanfieldth}, we suppose that $t_z \in K_1$ and $T/T_N(U, t_z) \in K_2$ where $K_1$ and $K_2$ are compact subsets of $(0, 2)$ and $[0,1)$, respectively. 
Using Theorem \ref{3Dmeanfieldth}, the next lemma is proved in the same way as Lemma \ref{tanhsqrtlemma}.

\begin{lemma}\label{tanhsqrtlemma2}
As $U \downarrow 0$, 
\begin{align}\label{tanhsqrtexpansion2}
\frac{\tanh(\frac{\sqrt{\hat{m}^2 + \epsilon^2}}{2T/T_N})}{\sqrt{\hat{m}^2 + \epsilon^2}}
= \frac{\tanh(\frac{\sqrt{f_{\BCS}(T/T_N)^2 + \epsilon^2}}{2T/T_N})}{\sqrt{f_{\BCS}(T/T_N)^2 + \epsilon^2}}\bigg(1 + O\bigg(\frac{U^{-1} e^{- \frac{2}{N_{t_z}(0) U}}}{1 + \epsilon^2}\bigg)\bigg)
\end{align}
uniformly for $\epsilon > 0$, $t_z \in K_1$, and $T$ such that $T/T_N \in K_2$.
\end{lemma}

The function $J(\hat{m}, T/T_N)$ in (\ref{hatmfBCSJ}) is given by
\begin{align}\label{JhatmTTN}
 J(\hat{m}, T/T_N)
=& -\int_0^\infty \frac{N_{t_z}(T_N\epsilon)  - N_{t_z}(0)}{N_{t_z}(0)} 
\bigg(\frac{\tanh(\frac{\sqrt{\hat{m}^2 + \epsilon^2}}{2T/T_N})}{\sqrt{\hat{m}^2 + \epsilon^2}}  -
\frac{\tanh(\frac{\epsilon}{2})}{\epsilon} \bigg) d\epsilon.
\end{align}
Utilizing (\ref{tanhsqrtexpansion2}) in (\ref{JhatmTTN}), we get
\begin{align}\nonumber
 J(\hat{m}, T/T_N)
=& -\int_0^\infty \frac{N_{t_z}(T_N \epsilon)  - N_{t_z}(0)}{N_{t_z}(0)} 
\bigg(\frac{\tanh(\frac{\sqrt{f_{\BCS}(T/T_N)^2 + \epsilon^2}}{2T/T_N})}{\sqrt{f_{\BCS}(T/T_N)^2 + \epsilon^2}}  -
\frac{\tanh(\frac{\epsilon}{2})}{ \epsilon}\bigg) d \epsilon
	\\ \label{JhatmTTN2}
& + O\Big(U^{-1} e^{- \frac{2}{N_{t_z}(0) U}} S(U,t_z,T)\Big) \qquad \text{as $U \downarrow 0$}
\end{align}
uniformly for $t_z \in K_1$ and  $T/T_N \in K_2$, where 
$$S(U,t_z,T) := \int_0^\infty \bigg|\frac{N_{t_z}(T_N \epsilon)  - N_{t_z}(0)}{N_{t_z}(0)} 
\frac{\tanh(\frac{\sqrt{f_{\BCS}(T/T_N)^2 + \epsilon^2}}{2T/T_N})}{\sqrt{f_{\BCS}(T/T_N)^2 + \epsilon^2}} \bigg| \frac{1}{1 + \epsilon^2}d \epsilon.$$

\begin{lemma}\label{Slemma}
The function $S(U,t_z,T)$ obeys the estimate
\begin{align*}
S(U,t_z,T) = O\Big(U^{-1} e^{- \frac{2}{N_{t_z}(0) U}} \Big)\qquad \text{as $U \downarrow 0$}
\end{align*}
uniformly for $t_z \in K_1$ and  $T/T_N \in K_2$.
\end{lemma}
\begin{proof}
Since $|\tanh x| \leq 1$ for $x \in \R$ and since there is a $c > 0$ such that $f_{\BCS}(T/T_N) > c$ for $T/T_N \in K_2$, we have
$$S(U,t_z,T) \leq C \int_0^\infty \bigg|\frac{N_{t_z}(T_N \epsilon)  - N_{t_z}(0)}{N_{t_z}(0)} 
\bigg| \frac{1}{1 + \epsilon^3}d \epsilon$$
for $t_z \in K_1$, $T/T_N \in K_2$, and all sufficiently small $U > 0$.
Since $N_{t_z}(\epsilon) = 0$ for $\epsilon \geq 4 + 2t_z$, this can be expressed as
$$S(U,t_z,T) \leq C \int_0^{\frac{4 + 2t_z}{T_N}} \bigg|\frac{N_{t_z}(T_N \epsilon)  - N_{t_z}(0)}{N_{t_z}(0)} 
\bigg| \frac{1}{1 + \epsilon^3}d \epsilon
+ C \int_{\frac{4 + 2t_z}{T_N}}^\infty \frac{1}{1 + \epsilon^3}d \epsilon.$$
By Theorem \ref{3Ddensityth}, we can use the estimate
$$\bigg|\frac{N_{t_z}(T_N \epsilon)  - N_{t_z}(0)}{N_{t_z}(0)} 
\bigg| \leq C T_N^2 \epsilon^2$$
in the first integral, which gives
$$S(U,t_z,T) \leq C T_N^2 \int_0^{\frac{4 + 2t_z}{T_N}} \frac{\epsilon^2}{1 + \epsilon^3}d \epsilon
+ C T_N^2
\leq C T_N^2 |\ln T_N|.$$
In light of (\ref{3DTNbound}), the lemma follows.
\end{proof}

Theorem \ref{3Dmeanfieldimprovedth} is a direct consequence of (\ref{hatmfBCSJ}), (\ref{JhatmTTN2}), and Lemma \ref{Slemma}.

\bigskip\noindent
{\bf Acknowledgements} {\it We thank J. Henheik and A. Lauritsen for useful discussions. 
E.L. acknowledges support from the Swedish Research Council, Grant No. 2023-04726. 
J.L. acknowledges support from the Swedish Research Council, Grant No.\ 2021-03877.}

\bibliographystyle{plain}
\bibliography{is}

\end{document}